\documentclass[11pt,a4paper]{article}
\usepackage{amssymb}
\usepackage{amsmath}
\usepackage{amscd}
\usepackage[all,cmtip]{xy}
\usepackage{amsthm}
\usepackage{amsfonts}
\usepackage{mathrsfs}
\usepackage{marvosym}
\usepackage{bbm}
\usepackage{graphicx}
\usepackage{a4wide}
\usepackage[utf8]{inputenc}
\usepackage{hyperref}
\usepackage[margin=2.2cm]{geometry}
\usepackage{color}
\usepackage{comment}

\usepackage{stmaryrd}

\definecolor{darkred}{rgb}{0.8,0.1,0.1}
\hypersetup{
     colorlinks=true,         
     linkcolor=darkred,
     citecolor=blue,
}

\theoremstyle{plain}
\newtheorem{theo}{Theorem}[section]
\newtheorem{lem}[theo]{Lemma}
\newtheorem{propo}[theo]{Proposition}

\theoremstyle{definition}
\newtheorem{defi}[theo]{Definition}

\newtheorem{rem}[theo]{Remark}

\numberwithin{equation}{section}

\def\nn{\nonumber}

\def\bbR{\mathbb{R}}
\def\bbC{\mathbb{C}}
\def\bbN{\mathbb{N}}
\def\bbZ{\mathbb{Z}}
\def\bbT{\mathbb{T}}
\def\Loc{\mathsf{Loc}}
\def\LocSrc{\mathsf{LocSrc}}
\def\LocTop{\mathsf{LocTop}}
\def\bfM{\boldsymbol{M}}
\def\o{\mathfrak{o}}
\def\t{\mathfrak{t}}
\def\g{\boldsymbol{g}}
\def\h{\boldsymbol{h}}
\def\j{\boldsymbol{j}}
\def\x{\boldsymbol{x}}
\def\bfJ{\boldsymbol{J}}
\def\bfJt{\widetilde{\boldsymbol{J}}}

\def\Vec{\mathsf{Vec}}
\def\Aff{\mathsf{Aff}}

\def\PreSymp{\mathsf{PreSymp}}

\def\astAlg{{}^\ast\mathsf{Alg}}

\def\PoisAlg{\mathsf{PoisAlg}}

\def\Cinfty{\mathfrak{C}^\infty}
\def\Ainfty{\mathfrak{A}^\infty}
\def\PhSp{\mathfrak{PS}}
\def\Sol{\mathfrak{Sol}}
\def\NN{\mathfrak{N}}
\def\CCR{\mathfrak{CCR}}
\def\CanPois{\mathfrak{CanPois}}
\def\CPA{\mathfrak{CPA}}
\def\PA{\mathfrak{PA}}
\def\CQA{\mathfrak{CQA}}
\def\QA{\mathfrak{QA}}
\def\Z{\mathfrak{Z}}

\def\End{\mathrm{End}}

\def\Aut{\mathrm{Aut}}

\def\id{\mathrm{id}}
\def\dd{\mathrm{d}}

\def\supp{\mathrm{supp}}

\def\dim{\mathrm{dim}}

\def\vol{\mathrm{vol}}
\def\1{\mathbbm{1}}

\newcommand{\ip}[2]{\left\langle #1,#2 \right\rangle}
\newcommand{\sip}[2]{\left\langle\left\langle #1,#2\right\rangle\right\rangle}
\newcommand{\omi}[1]{\buildrel { \buildrel{#1}\over{\vee} } \over .}

\def\sk{\vspace{2mm}}

\newcommand{\pPreSymp}{\bullet\PreSymp}
\newcommand{\pPhSp}{\bullet\PhSp}
\newcommand{\pCCR}{\bullet\CCR}
\newcommand{\obulplus}{\mathbin{\ooalign{$\oplus$\cr\hidewidth\raise0.17ex\hbox{$\scriptstyle\bullet\mkern4.48mu$}}}}

\newcommand{\rce}{\mathrm{rce}}

\title{%
Locally covariant quantum field theory with external sources\\
{\Large Relative Cauchy evolution, automorphisms and dynamical locality}
}

\author{%
Christopher J. Fewster$^{1,a}$ and Alexander Schenkel$^{2,b}$\thanks{Present address:
Department of Mathematics, Heriot-Watt University,
 Edinburgh EH14 4AS, United Kingdom. \newline
Maxwell Institute for Mathematical Sciences, Edinburgh, United Kingdom. \newline
The Tait Institute, Edinburgh, United Kingdom.\newline
Email: \texttt{as880@hw.ac.uk}\newline
Supported by a Research Fellowship of Deutsche Forschungsgemeinschaft (DFG).}\vspace{4mm}\\
{\small $^1$Department of Mathematics}\\ 
{\small University of York, Heslington, York, YO10 5DD, UK.}\vspace{2mm}\\
{\small $^2$ Fachgruppe Mathematik}\\
{\small Bergische Universit\"at Wuppertal, Gau\ss stra\ss e 20, 42119 Wuppertal, Germany.}\vspace{4mm}\\
 {\footnotesize  ~$^a$ \texttt{chris.fewster@york.ac.uk}~,~$^b$ \texttt{schenkel@math.uni-wuppertal.de} }
 }

\date{July 2014}

\begin{document}

\maketitle

\begin{abstract}
We provide a detailed analysis of the classical and quantized theory of a multiplet of inhomogeneous Klein--Gordon fields, which 
couple to the spacetime metric and also to an external source term; 
thus the solutions form an affine space. 
Following the formulation of affine field theories
in terms of presymplectic vector spaces as proposed in $[$Annales Henri Poincar{\'e} {\bf 15}, 171 (2014)$]$, 
we determine the relative Cauchy evolution induced by metric as well as source term perturbations
and compute the automorphism group of natural isomorphisms of the presymplectic vector space functor. 
Two pathological features of this formulation are revealed: 
the automorphism group contains elements that cannot be interpreted
as global gauge transformations of the theory; moreover, the presymplectic formulation does not respect a natural requirement on 
composition of subsystems. We therefore propose a systematic strategy to improve the original description
of affine field theories at the classical and quantized level,
first passing to a Poisson algebra description in the classical case.
The idea is to consider state spaces on the classical and quantum algebras suggested by the physics of the theory (in the classical case, we use the affine solution space). The state spaces are not separating for the algebras, indicating a redundancy in the description. 
Removing this redundancy by a quotient, a functorial
theory is obtained that is free of the above mentioned pathologies.
These techniques are applicable to general affine field theories and Abelian gauge theories. The resulting quantized theory is shown to be dynamically local. 
\end{abstract}
\paragraph*{Keywords:}
locally covariant quantum field theory, 
relative Cauchy evolution,
quantum field theory on curved spacetimes,
affine quantum field theory
\paragraph*{MSC 2010:}81T20, 81T05


\section{\label{sec:intro}Introduction}
Our understanding of quantum field theories on Lorentzian manifolds has made tremendous
developments since the principle of general local covariance was introduced in \cite{Brunetti:2001dx}.
Its underlying physical idea, which roughly speaking says that
any reasonable quantum field theory should be defined coherently on {\it all} spacetimes instead of focusing 
on formulations in individual spacetimes, is expressed mathematically in terms of category theory.
The basic structure of interest is that of a covariant functor from a category of spacetimes (possibly with extra 
data such as fibre bundles) to a category of algebras, which is supposed to describe the association of
observable algebras to spacetimes.
The benefits from this new perspective on quantum field theory
are substantial: On the one hand, many structural problems have been addressed and solved, for example
the generalization of the famous spin-statistics theorem to curved spacetimes \cite{Verch:2001bv}
and the perturbative renormalization of quantum field theories, see e.g.\ \cite{Brunetti:2009qc,Hollands:2001nf,Hollands:2001fb}
for general developments and \cite{Fredenhagen:2011mq,Brunetti:2013maa} for perturbative gauge and gravity theories.
On the other hand, the locally covariant framework also has had an impact
on applications of quantum field theory to e.g.\ quantum energy inequalities
\cite{Fewster:2006kt,Fewster:2006iy} and cosmology \cite{Verch:2011bx,Pinamonti:2013zba,Pinamonti:2013wya}.
\sk

Another new and interesting aspect arising in the locally covariant framework
is that internal symmetries of (quantum) field theories can be promoted to the functorial level.
It has been proposed recently by one of us \cite{Fewster:2012yc} that
the automorphism group (of natural isomorphisms) of a locally covariant quantum field theory functor
is a suitable generalization to curved spacetimes of the global gauge group in Minkowski 
space algebraic quantum field theory. Besides clarifying general properties of such automorphism
groups, it has been shown in \cite{Fewster:2012yc} that this concept indeed captures the usual orthogonal 
symmetries of the quantum theory of a multiplet of Klein--Gordon fields with equal masses.

\sk
In this paper we investigate how some of the above features are
modified when the basic category is enriched from a category of 
spacetimes to include additional external sources. This particularly
influences the {\em relative Cauchy evolution} \cite{Brunetti:2001dx}, which measures
sensitivity of a theory to perturbations of the background structure,
and plays a key role in the classification of the automorphism group
\cite{Fewster:2012yc}
and also in defining the notion of a {\em dynamically local} theory  \cite{Fewster:2011pe}. External sources provide additional degrees
of freedom for the relative Cauchy evolution to exploit, leading to
a richer framework. 
\sk

Our investigation is conducted in the context of an example, namely the
theory of a multiplet of {\it inhomogeneous} Klein--Gordon fields interacting with an external source, with underlying 
Lagrangian 
\begin{flalign}\label{eqn:actionintro}
\mathcal{L} = \sqrt{\vert \g \vert}\left(\frac{1}{2}\ip{\nabla_a\phi}{\nabla^a\phi}- \frac{1}{2}m^2\ip{\phi}{\phi} 
- \ip{\bfJ}{\phi}\right)~,
\end{flalign}
where $\bfJ\in C^\infty(M,\bbR^p)$ is a classical and non-dynamical source term.
The interest in this model comes from two directions. Physically, it represents an approximation to a model
 in which $\phi$ is coupled both to gravity and to other fields, but with the simplifying assumption that not 
 only the metric but also the other fields have been `frozen' as background structure represented by the external source; this would be an appropriate
  approximation in situations where the back-reaction of $\phi$ on both the metric and other fields can be neglected. 
  Mathematically, interest arises because, in contrast to the homogeneous theory with $\bfJ=0$,
the equation of motion corresponding to this Lagrangian is not linear, but affine, and as a consequence
the space of solutions is not a vector space, but rather an affine space. On the one hand, replacing linear structures
by affine ones leads to the simplest `non-linear' models of quantum field theories, which still can be treated
exactly without the need for perturbative expansions \cite{Benini:2012vi}. On the other hand, these
 affine structures are unavoidable in gauge theories
\cite{Benini:2013tra,Benini:2013ita} as the space of connections on a principal bundle is intrinsically an affine space.
Hence, inhomogeneous theories such as (\ref{eqn:actionintro}) can be regarded as toy-models
for gauge theories, which reflect parts of their geometric structure.
For these reasons, a general study of affine field theories was recently undertaken in \cite{Benini:2012vi}.
See also \cite{Sanders:2012sf} for a recent study of the $p$-form fields (in particular, the Maxwell field) coupled to an external source term.
\sk

Formulating the (classical and quantum) inhomogeneous Klein--Gordon theory according to 
\cite{Benini:2012vi}, however, we have found that the resulting functors have two serious pathologies. 
First, their automorphism groups do not reflect the expected symmetries of this model: From the Lagrangian
(\ref{eqn:actionintro}) one expects that the usual orthogonal symmetries are broken due to the presence
of the (arbitrary) source terms $\bfJ$, so that only a translation symmetry $\phi \mapsto\phi + \mu $,
$\mu\in\bbR^p$, remains in the massless case $m=0$. By contrast, the functor constructed in 
\cite{Benini:2012vi} always has a $\bbZ_2$ (sub)group of automorphisms, which has no corresponding interpretation
at the level of the Lagrangian. Second, the functors of \cite{Benini:2012vi} provide a description 
of a multiplet of $p$ (mutually noninteracting) fields that is inequivalent to 
what it would provide for the composition of $p$ copies of a single field. These
defects convince us that there is a flaw in this earlier description
of affine field theories and that the corresponding functor has to be modified. 
\sk

A second theme of this paper, then, is to propose a systematic way to improve the construction of the classical and quantum theory of the inhomogeneous 
multiplet of Klein--Gordon fields. The essential ingredient is the use of suitable state spaces 
(in the classical theory given by the solution space) and the characterization of the vanishing ideals induced by these state spaces
in the abstract algebras considered in \cite{Benini:2012vi}. These
ideals reflect redundancies in that description, and we therefore
quotient by the vanishing ideals to obtain improved functors. 
Our construction is free of the pathologies of \cite{Benini:2012vi}: 
the functors have the expected automorphism groups and satisfy
a natural composition property with respect to the size of the multiplet.
Furthermore, we will see that the functorial perspective can illuminate some aspects of the structure of these models that is obscured in other approaches (see Remark~\ref{rem:splitting} in particular).
This work provides the foundation for a broader discussion of the properties of these models.  
It is worth emphasizing that our improved classical theory is not formulated in terms of presymplectic vector spaces,
but in terms of Poisson algebras.\footnote{In Appendix~\ref{sec:pointed}, however, we explain a 
formulation using {\em pointed} presymplectic spaces.} Similarly, the improved quantum 
theory is not given by CCR-algebras, but by quotients of such algebras. 
Finally, we would like to mention that the techniques developed in this work 
are important for and can be applied to generic affine field theories \cite{Benini:2012vi} and, with slight modifications due to
the presence of gauge invariance, also to Abelian gauge theories \cite{Benini:2013tra,Benini:2013ita}. 
\sk

The outline of this paper is as follows: 
In Section \ref{sec:prelim} we shall review briefly the techniques required for studying affine field theories
\cite{Benini:2012vi}, focusing for simplicity on the explicit example given by the inhomogeneous multiplet of Klein--Gordon fields.
The relative Cauchy evolution for this model is discussed in detail in Section \ref{sec:rce}; as a new feature compared
to earlier studies, we study perturbations of both the metric $\g$ and  the external source
$\bfJ$. The derivative of the relative Cauchy evolution along metric perturbations is calculated
and it is shown how to identify it with the stress-energy tensor corresponding to the action given by (\ref{eqn:actionintro}).
Furthermore, the derivative of the relative Cauchy evolution along external source perturbations
is determined and identified with the $\bfJ$-variation of the action given by (\ref{eqn:actionintro}).
In Section \ref{sec:automorphisms} we compute the automorphism group of the functor
describing the presymplectic vector spaces of the classical theory of a multiplet of inhomogeneous Klein--Gordon fields.
We find that all endomorphisms of this functor (embeddings of the theory as a
subtheory of itself) are in fact automorphisms (global gauge transformations), and that the automorphism group
is isomorphic to $\bbZ_2$ in the massive case and to $\bbZ_2\times \bbR^p$ for $m=0$.
The nontrivial $\bbZ_2$ automorphism does not describe a symmetry of the Lagrangian (\ref{eqn:actionintro}), 
suggesting that inhomogeneous field theories are not appropriately described by the presymplectic vector space 
functor developed in \cite{Benini:2012vi}. This suggestion is strengthened in Section \ref{sec:compproperty}, 
where we study a composition
property: Any pair $(\bfM,\bfJ)$
consisting of a spacetime $\bfM$ with source term $\bfJ\in C^\infty(M,\bbR^p)$ may be split in a functorial way 
 into two pairs $\big((\bfM,\bfJ^q),(\bfM,\bfJ^{p-q})\big)$,
where the source $\bfJ$ is split into the first $q$ and last $p-q$ components. Treating the two pairs individually 
by the presymplectic vector space functor of \cite{Benini:2012vi}, we get a separate description
 of the first $q$ and last $p-q$ components
of the inhomogeneous Klein--Gordon field. We observe that the direct sum of these two presymplectic vector spaces
is not isomorphic to the original presymplectic vector space, and as a consequence the theory
obtained in the direct way is not naturally isomorphic to the one  obtained after splitting.
 As the individual components of the inhomogeneous Klein--Gordon field have no mutual interactions,
  this behavior is pathological and strengthens our claim that the presymplectic vector space functor
 is not a satisfactory description of the inhomogeneous theory of a multiplet of Klein--Gordon fields.
\sk
 
In  Section \ref{sec:poisson} we show how to resolve these issues by passing from  the category 
of presymplectic vector spaces to that of Poisson algebras. The presymplectic vector space
 of \cite{Benini:2012vi} has a canonical corresponding (abstract) 
Poisson algebra which can be represented naturally as an algebra of
functionals on the affine space of solutions to the inhomogeneous Klein--Gordon equation.
 In this representation there arises a kernel, which has no corresponding analog at the level of the presymplectic
vector spaces. We show that these kernels are natural Poisson ideals and hence we can modify our Poisson algebra
functor by quotienting them out. The resulting improved Poisson algebra functor is shown to have
the expected automorphism group (i.e.\ the trivial group for $m\neq 0$ and $\bbR^p$ for $m=0$)
and to satisfy the composition property. Hence, it is a better description of the classical theory
of a multiplet of inhomogeneous Klein--Gordon fields. We extend these constructions to the quantum level in
Section \ref{sec:quantization}. The main idea is to characterize suitable state spaces
for the CCR-algebras obtained by canonical quantization of our presymplectic vector spaces,
which reflect the fact that the latter describe affine functionals on the solution space of 
the inhomogeneous theory. Quotienting by the intersection of the kernels of 
corresponding GNS representations, we obtain our improved (functorial) quantized theory, which has the correct automorphism
group and satisfies the composition property. We  also prove that
our improved theory satisfies the dynamical locality property introduced in \cite{Fewster:2011pe,Fewster:2011pn}.
Furthermore, we compare our improved algebras
with the algebras for inhomogeneous theories constructed
by Hollands and Wald \cite{Hollands:2004yh} as a quotient of the Borchers--Uhlmann algebra \cite{Borchers,Uhlmann} 
and show that they are naturally isomorphic; nonetheless, as
we describe in more detail in Section \ref{sec:conclusion}, our methods provide a wider perspective on this approach. 
The somewhat special case of the massless multiplet of inhomogeneous Klein--Gordon fields and its interpretation as
a rather simple kind of gauge theory is discussed in Section \ref{sec:massless}.
In Section \ref{sec:conclusion} we add some concluding remarks, which should show that the techniques developed
in this paper can be readily applied to generic affine field theories in the sense of \cite{Benini:2012vi} and also to  Abelian 
gauge theories \cite{Benini:2013tra,Benini:2013ita}. 
There are three appendices. 
Appendix \ref{app:Cauchy} includes details on how to take the derivative of the relative Cauchy evolution
and the stress-energy tensor.
In Appendix \ref{sec:pointed} we give an alternative solution to the problems arising with the presymplectic
vector space functor by introducing a category of {\it pointed} presymplectic spaces. Finally, Appendix \ref{sec:fedosov} treats 
 the quantization of our model by deformation methods. It turns out that our improved classical
Poisson algebra is amenable to direct deformation quantization; 
alternatively, one may also apply an algebraic version of Fedosov's method -- both lead to the improved quantum theory discussed
in the text.  
We comment on the relationship between our approach
and that of the recent paper~\cite{Sanders:2012sf}.


\section{\label{sec:prelim}Preliminaries}
\subsection{Basics and notations}
The model we study throughout this work is given by a multiplet of $p\in \bbN$ real scalar fields (with the same mass),
which are minimally coupled to the Lorentzian metric and in addition coupled to an external source.
We shall exclusively work in a category theoretical setting, which is an extension of
the framework of locally covariant quantum field theory developed in \cite{Brunetti:2001dx}, see also \cite{Fewster:2011pe}.
The basic category entering our construction is given by the following
\begin{defi}\label{defi:LocSrcp}
The category $\LocSrc_{p}$ consists of the following objects and morphisms:
\begin{itemize}
\item The objects in $\LocSrc_{p}$ are pairs $(\bfM,\bfJ)$, where $\bfM = (M,\o,\g,\t)$ is 
a manifold $M$ (of arbitrary but finite dimension and with finitely many connected components) with orientation $\o$,
globally hyperbolic Lorentzian metric $\g$ (of signature $(+,-,\cdots,-)$) and time-orientation $\t$,
and $\bfJ\in C^\infty(M,\bbR^p)$. 
\item The morphisms $f:(\bfM_1,\bfJ_1)\to (\bfM_2,\bfJ_2)$ in $\LocSrc_{p}$ are orientation and time-orientation preserving
isometric embeddings $f:M_1\to M_2$, such that $f[M_1]\subseteq M_2$ is causally compatible and open
and such that $f^\ast(\bfJ_2) = \bfJ_1$, where $f^\ast$ denotes the pull-back.
\end{itemize}
\end{defi}
Any morphism whose image contains a Cauchy surface of the codomain 
will be called a {\em Cauchy morphism}; any functor from $\LocSrc_p$ to some other category is 
said to obey the {\em time-slice axiom} if it maps every Cauchy morphism to an isomorphism. 
\sk

The configuration space of a multiplet of $p\in \bbN$ real scalar fields is given by the following contravariant
functor $\Cinfty_{p} :\LocSrc_{p}\to \Vec $ to the category of real vector spaces:
To any object $(\bfM,\bfJ)$ in $\LocSrc_{p}$ we associate $\Cinfty_p(\bfM,\bfJ) := C^\infty(M,\bbR^p)$ 
and to any morphism $f:(\bfM_1,\bfJ_1)\to (\bfM_2,\bfJ_2)$ in $\LocSrc_{p}$ we associate the pull-back
$\Cinfty_{p}(f):=f^\ast:\Cinfty_p(\bfM_2,\bfJ_2)\to \Cinfty_p(\bfM_1,\bfJ_1)$.
\sk

We model the equations of motion for our theory, given by inhomogeneous Klein--Gordon equations,
 by a certain natural transformation. As a first step, remember that
the homogeneous Klein--Gordon equation is described by the natural transformation
$\mathrm{KG}= \{\mathrm{KG}_{\bfM} \}: \Cinfty_{p} \Rightarrow \Cinfty_{p}$ given by the Klein--Gordon
operators 
\begin{flalign}
\mathrm{KG}_{\bfM} : \Cinfty_{p}(\bfM,\bfJ) \to \Cinfty_{p}(\bfM,\bfJ)~,~~\phi\mapsto \mathrm{KG}_{\bfM}(\phi) 
= \square_{\bfM}(\phi) + m^2\phi~.
\end{flalign}
Here $\square_{\bfM}$ is the d'Alembert operator corresponding to $\bfM$
 and $m\ge 0$ is a fixed  mass.
In order to couple the theory to the sources $\bfJ$, which are  part of the data of the category
$\LocSrc_{p}$, we have to relax the condition that morphisms in $\Vec$ are linear. 
Let us therefore introduce the category $\Aff$ of affine spaces over real vector spaces 
with affine maps as morphisms and the obvious forgetful functor $\mathfrak{Forget}:\Vec \to \Aff$.
We compose $\Cinfty_p$ with the functor $\mathfrak{Forget}$
and obtain as result a contravariant functor $\Ainfty_p:= \mathfrak{Forget}\circ \Cinfty_p : \LocSrc_p \to \Aff$.
The inhomogeneous Klein--Gordon operators are then described by  the natural transformation
$\mathrm{P} = \{\mathrm{P}_{(\bfM,\bfJ)}\} :  \Ainfty_p \Rightarrow \Ainfty_p$  given by
\begin{flalign}
\mathrm{P}_{(\bfM,\bfJ)}: \Ainfty_{p}(\bfM,\bfJ) \to \Ainfty_{p}(\bfM,\bfJ)~,~~\phi \mapsto \mathrm{P}_{(\bfM,\bfJ)}(\phi) 
=\square_{\bfM}(\phi)  + m^2\phi + \bfJ~.
\end{flalign}
The solution spaces for these equations can be given a functorial form. Note that we do 
{\em not} assume that the solutions have spacelike compact support (there is no assumption on the support of $\bfJ$). 
\begin{defi}
The contravariant functor $\Sol_p: \LocSrc_p \to \Aff$ is defined as follows: To any object 
$(\bfM,\bfJ)$ in $\LocSrc_p$ it associates the solution space 
\begin{flalign}
\Sol_p(\bfM,\bfJ):= 
\{\phi\in C^\infty(M,\bbR^p) : \mathrm{P}_{(\bfM,\bfJ)}(\phi) =\mathrm{KG}_{\bfM}(\phi) + \bfJ =0\}~,
\end{flalign} 
which is an affine space
over the vector space $\Sol_p^\mathrm{lin} (\bfM) := \{\underline{\phi} \in C^\infty(M,\bbR^p): \mathrm{KG}_{\bfM}(\underline{\phi})=0\}$.
To any morphism $f:(\bfM_1,\bfJ_1)\to(\bfM_2,\bfJ_2)$ in $\LocSrc_p$,
$\Sol_p$ associates the $\Aff$-morphism 
given by the pull-back $\Sol_p(f):= f^\ast : \Sol_p(\bfM_2,\bfJ_2)\to\Sol_p(\bfM_1,\bfJ_1)$.
\end{defi}

\subsection{The presymplectic vector space functor}\label{sec:psvsf}
We follow the prescription of \cite{Benini:2012vi} in order to construct a covariant functor
$\PhSp_p : \LocSrc_p\to \PreSymp$ associating presymplectic vector spaces to objects in $\LocSrc_p$, 
whose role is to label certain affine functionals on $\Sol_p(\bfM,\bfJ)$, i.e.\ observables of the theory.\footnote{
In \cite{Benini:2012vi} this functor was denoted by $\mathfrak{PhSp}$ and it was called the ``phase space functor''.
We avoid this notation in our present work, since the presymplectic vector spaces obtained by $\mathfrak{PhSp}$
are just labeling spaces for functionals and not what one typically calls the phase space 
(i.e.\ the space of initial data or the space of solutions). 
} Here $\PreSymp$ denotes the category of real presymplectic vector spaces, with all
 morphisms being assumed to be injective. The aim is to have sufficiently many observables to separate the solutions, while
also removing redundancy by identifying observables that vanish on all solutions. Accordingly, 
to any object $(\bfM,\bfJ)$ in $\LocSrc_p$ we associate the object $\PhSp_p(\bfM,\bfJ)$ in $\PreSymp$
 given by the following construction:  As a vector space,
\begin{flalign}\label{eqn:PhSp_VS}
\PhSp_p(\bfM,\bfJ):= \big(C^\infty_0(M, \bbR^{p})\oplus\bbR\big)/\mathrm{P}^\ast_{(\bfM,\bfJ)}[C^\infty_0(M,\bbR^p)]~,
\end{flalign}
where, for all $h\in C^\infty_0(M,\bbR^p)$,
\begin{flalign}\label{eqn:equivalence}
\mathrm{P}^\ast_{(\bfM,\bfJ)}(h) = \left(
\mathrm{KG}_{\bfM}(h),\int_{M} \ip{\bfJ}{h} \,\vol_{\bfM} \right) \in C^\infty_0(M,\bbR^{p})\oplus \bbR~.
\end{flalign}
(One may also understand this construction as follows: 
$\PhSp_p(\bfM,\bfJ)$ is (isomorphic to) the vector space of compactly supported sections of the vector dual bundle
of our configuration bundle $M\times \bbR^p\stackrel{\mathrm{pr}_1}{\to} M$ (in the category of affine bundles)
modulo the quotient which identifies with zero those elements corresponding
to functionals which act trivially on all solutions. This viewpoint, which also leads naturally to the definitions of 
the presymplectic structure and morphisms given below, is spelled out in 
more detail in \cite[Section 4 and Section 5]{Benini:2012vi}.)
\sk

The presymplectic structure in $\PhSp_p(\bfM,\bfJ)$ is defined by, for all $[(\varphi,\alpha)],[(\psi,\beta)] \in \PhSp_p(\bfM,\bfJ)$,
\begin{flalign}\label{eqn:PhSp_PS}
\sigma_{(\bfM,\bfJ)}\big([(\varphi,\alpha)],[(\psi,\beta)] \big) := \int_{M} \ip{\varphi}{\mathrm{E}_{\bfM}(\psi)} \,\vol_{\bfM}~,
\end{flalign}
where $\mathrm{E}_{\bfM} = \mathrm{E}^-_{\bfM} - \mathrm{E}^+_{\bfM}$ is the advanced-minus-retarded
Green's operator for $\mathrm{KG}_{\bfM}$, and the Green's operators 
obey $\supp (\mathrm{E}^\pm_{\bfM}(\varphi)) \subseteq  J^\pm_{\bfM}(\supp(\varphi))$.
\sk

To any morphism $f: (\bfM_1,\bfJ_1)\to(\bfM_2,\bfJ_2)$ in $\LocSrc_p$ the functor $\PhSp_p$
associates the morphism $\PhSp_p(f) : \PhSp_p(\bfM_1,\bfJ_1)\to \PhSp_p(\bfM_2,\bfJ_2)$ in
$\PreSymp$ that is canonically induced by the push-forward,
\begin{flalign}\label{eqn:PhSp_mor}
\PhSp_p(f)\big([(\varphi,\alpha)]\big) := [(f_\ast(\varphi),\alpha)]~,
\end{flalign} 
for any $[(\varphi,\alpha)] \in \PhSp_p(\bfM_1,\bfJ_1)$, which is 
well-defined because 
\begin{flalign}
\left(f_\ast\big(\mathrm{KG}_{\bfM_1} (h)\big), \int_{M_1} \ip{\bfJ_1}{h} \,\vol_{\bfM_1}\right) = \mathrm{P}^\ast_{(\bfM_2,\bfJ_2)}\big(f_\ast (h)\big)~,
\end{flalign}
and injective because of the general result in \cite[Theorem 5.4.]{Benini:2012vi}.
As mentioned, the role of the covariant functor $\PhSp_p$ is to label affine functionals on the contravariant functor $\Sol_p$. 
This manifests itself in a natural dual pairing: For each object $(\bfM,\bfJ)$ in $\LocSrc_p$ the evaluation map 
\begin{flalign}\label{eqn:pairing}
\sip{\,\cdot\,}{\,\cdot\,}_{(\bfM,\bfJ)} : \PhSp_p(\bfM,\bfJ)\times \Sol_p(\bfM,\bfJ)\to \bbR~,~~
\big([(\varphi,\alpha)],\phi\big) \mapsto \left( \int_M\ip{\varphi}{\phi}\,\vol_{\bfM} \right)+ \alpha~
\end{flalign}
is well-defined and linear in the left and affine in the right entry. Naturality means that the following diagram commutes for any morphism
$f: (\bfM_1,\bfJ_1)\to (\bfM_2,\bfJ_2)$ in $\LocSrc_p$ 
\begin{flalign}\label{eqn:pairingdiagramphsp}
\xymatrix{
\PhSp_p(\bfM_1,\bfJ_1)\times \Sol_p(\bfM_2,\bfJ_2)\ar[d]_-{\PhSp_p(f)\times \id_{\Sol_p(\bfM_2,\bfJ_2)}} \ar[rrrr]^-{\id_{\PhSp_p(\bfM_1,\bfJ_1)}\times \Sol_p(f)}&&&&\PhSp_p(\bfM_1,\bfJ_1)\times \Sol_p(\bfM_1,\bfJ_1)
\ar[d]^-{\sip{\,\cdot\,}{\,\cdot\,}_{(\bfM_1,\bfJ_1)}}\\
\PhSp_p(\bfM_2,\bfJ_2)\times \Sol_p(\bfM_2,\bfJ_2) \ar[rrrr]^-{\sip{\,\cdot\,}{\,\cdot\,}_{(\bfM_2,\bfJ_2)}}&&&&\bbR
}
\end{flalign}
Furthermore, the presymplectic structure in
$\PhSp_p(\bfM,\bfJ)$ coincides precisely with the Peierls bracket \cite{Peierls:1952} for the theory \eqref{eqn:actionintro}, 
on regarding elements of $\PhSp_p(\bfM,\bfJ)$ as observables in this way.
\sk

We summarize the main properties of the covariant functor $\PhSp_p$ defined by  
\eqref{eqn:PhSp_VS}, \eqref{eqn:PhSp_PS} and \eqref{eqn:PhSp_mor}, which follow immediately from
the general treatment of affine field theories in \cite{Benini:2012vi}.
\begin{propo}\label{propo:phspproperties}
\begin{itemize}
\item[a)] Let $(\bfM,\bfJ)$ be any object in $\LocSrc_p$. Then the null space $\NN_p(\bfM,\bfJ)$ of the presymplectic structure
in $\PhSp_p(\bfM,\bfJ)$ is isomorphic to $\bbR$.
\item[b)] The null space is functorial, i.e.~$\NN_p:\LocSrc_p\to \Vec$ is a covariant functor.
\item[c)] The covariant functor $\PhSp_p:\LocSrc_p \to \PreSymp$ satisfies the causality property and the time-slice axiom.
\end{itemize}
\end{propo}
\begin{proof}
The proof of a) follows from \cite[Corollary 4.5.]{Benini:2012vi} and b) follows from \cite[Lemma 7.3.]{Benini:2012vi}.
Item c) is a consequence of \cite[Theorem 5.5. and Theorem 5.6.]{Benini:2012vi}.
\end{proof}
Due to the nontrivial null space of the presymplectic structure (cf.~item a)) this theory has distinct features which are
not present in the homogeneous Klein--Gordon theory, where the null space is trivial.

\subsection{Quantization}
The theory $\PhSp_p: \LocSrc_p\to \PreSymp$  may be quantized by composing it with the 
canonical commutation relation (CCR) functor (either in Weyl or polynomial form). Since these quantization functors
preserve locality, causality and the time-slice axiom, we obtain a locally covariant quantum field theory 
in the sense of \cite{Brunetti:2001dx,Fewster:2011pe}, 
with the difference that our underlying geometric category
is enhanced from $\Loc$ to $\LocSrc_p$. 
This construction results in a functorial assignment of quantized observable algebras to objects
in $\LocSrc_p$, i.e.\ a covariant functor $\LocSrc_p\to \astAlg$ to the category of unital $\ast$-algebras. 
We would like to stress that in locally covariant quantum field theory
the primary focus is on the functor describing the observables algebras, while
quantum fields (being certain natural transformations with values in this functor, cf.\ \cite{Brunetti:2001dx}) 
are secondary concepts which can be assigned once this functor is specified.
Hence, we shall mostly focus in this work on the functor itself and refer the reader to
Remark \ref{rem:fields} (and also Remark \ref{rem:classical_fields}) for an example of
a locally covariant quantum (and classical) field for our model.
For more details on the Weyl quantization functors for presymplectic vector spaces
(and more general also presymplectic Abelian groups) we refer to the Appendix of 
\cite{Benini:2013ita}. The quantized theory of a multiplet of $p\in\bbN$ inhomogeneous
Klein--Gordon fields is studied in detail in Section \ref{sec:quantization}. 


\section{\label{sec:rce}Relative Cauchy evolution of the functor $\PhSp_p$}
Relative Cauchy evolution encodes the sensitivity of a theory to variations
in the background structures; in this it closely resembles the action. 
Apart from its intrinsic interest, understanding the relative Cauchy evolution
is also an integral step in characterizing the automorphism groups of our functors
in Section \ref{sec:automorphisms}.
We base our analysis on the refined construction given in \cite{Fewster:2011pe}, 
which we now review and adapt to our present setting. 
\sk

Given any object $(\bfM,\bfJ)$ in $\LocSrc_p$, we can consider its perturbation
by elements $(\h,\j)\in \Gamma^\infty_0(T^\ast M\vee T^\ast M) \times C^\infty_0(M,\bbR^p)$,
where $\Gamma^\infty_0(T^\ast M\vee T^\ast M)$ denotes the vector space of compactly supported 
sections of the symmetric tensor product of the cotangent bundle (i.e.~symmetric tensor fields).
Explicitly, given $(\h,\j)\in \Gamma^\infty_0(T^\ast M\vee T^\ast M)\times C^\infty_0(M,\bbR^p)$ such that $\g+\h$
is a time-orientable Lorentzian metric,
we define $(\bfM[\h],\bfJ[\j]) := \big((M,\o,\g+\h,\t_{\h}),\bfJ+\j\big)$,
where $\t_{\h}$ is the unique time-orientation for $\g+\h$, such that $\t_{\h}=\t$ outside
the support of $\h$. If $(\bfM[\h],\bfJ[\j])$ is an object in $\LocSrc_p$, i.e.~if $\bfM[\h]$ is globally hyperbolic,\footnote{It could happen that
 $\bfM[\h]$ is {\em not} globally hyperbolic: consider $\g=dt\otimes dt-\sum_{k=1}^{n-1} d\theta_k\otimes d\theta_k$
on $\bbR\times\bbT^{n-1}$ with $\h=\frac{1}{2}\eta(t) (dt\otimes d\theta_1+d\theta_1\otimes dt+2d\theta_1\otimes d\theta_1)$ where $\eta\in C_0^\infty(\bbR)$,
$0\le\eta\le 1$, and $\eta(0)=1$. Then $\g$ is globally hyperbolic, 
while $\g+\h$ is Lorentzian and time-orientable but admits a closed null curve in the hypersurface $t=0$.} we say that $(\h,\j)$ is a globally hyperbolic perturbation and write $(\h,\j)\in H(\bfM,\bfJ)$. 
Evidently $H(\bfM,\bfJ)$ contains an open neighborhood of $\{0\}\times C^\infty_0(M,\bbR^p) $
 in the usual test-section topology.
\sk

For any object  $(\bfM,\bfJ)$ in $\LocSrc_p$ and any $(\h,\j)\in H(\bfM,\bfJ)$ 
we define the sets 
\begin{equation}
M^\pm := M\setminus J_{\bfM}^{\mp}\big(\supp(\h)\cup \supp(\j)\big),
\end{equation}
which are causally compatible, open and globally hyperbolic subsets of $\bfM$ and $\bfM[\h]$. We further define
$\bfM^\pm[\h,\j]:= \bfM\vert_{M^\pm} = \bfM[\h]\vert_{M^\pm}$
and $\bfJ^\pm[\h,\j]:= \bfJ\vert_{M^\pm} = (\bfJ+\j)\vert_{M^\pm}$. 
Notice that $(\bfM^\pm[\h,\j],\bfJ^\pm[\h,\j])$ are objects in $\LocSrc_p$
and further that the canonical inclusions of underlying manifolds yield Cauchy morphisms 
\begin{subequations}
\begin{flalign}
i_{(\bfM,\bfJ)}^\pm[\h,\j]&: \big(\bfM^\pm[\h,\j],\bfJ^\pm[\h,\j]\big) \to \big(\bfM,\bfJ\big)~,\\
j_{(\bfM,\bfJ)}^\pm[\h,\j]&: \big(\bfM^\pm[\h,\j],\bfJ^\pm[\h,\j]\big) \to \big(\bfM[\h],\bfJ[\j]\big)~.
\end{flalign}
\end{subequations}
Since, by Proposition \ref{propo:phspproperties}, $\PhSp_p : \LocSrc_p \to \PreSymp$ satisfies
 the time-slice axiom, we can construct isomorphisms $\tau_{(\bfM,\bfJ)}^\pm[\h,\j]:
\PhSp_p(\bfM,\bfJ)\to \PhSp_p(\bfM[\h],\bfJ[\j])$ in $\PreSymp$ by
\begin{flalign}\label{eqn:rce_defa}
\tau_{(\bfM,\bfJ)}^\pm[\h,\j]:= \PhSp_p(j_{(\bfM,\bfJ)}^\pm[\h,\j])\circ \big(\PhSp_p(i_{(\bfM,\bfJ)}^\pm[\h,\j])\big)^{-1}~.
\end{flalign}
The relative Cauchy evolution of $\PhSp_p$ induced by $(\h,\j)\in H(\bfM,\bfJ)$ is defined as the automorphism
\begin{flalign}\label{eqn:rce_defb}
\mathrm{rce}_{(\bfM,\bfJ)}^{(\PhSp_p)}[\h,\j]:= \big(\tau_{(\bfM,\bfJ)}^{-}[\h,\j]\big)^{-1} \circ \tau^+_{(\bfM,\bfJ)}[\h,\j] \in \Aut(
\PhSp_p(\bfM,\bfJ))~,
\end{flalign}
and may be computed as follows.  Owing to the time-slice axiom, 
any element in $\PhSp_p(\bfM,\bfJ)$ may be written in the 
form\footnote{For clarity, in this discussion we shall indicate the $\LocSrc_p$ object with 
respect to which the equivalence relation is understood.} $[(\varphi,\alpha)]_{_{(\bfM,\bfJ)}}$ with 
$\varphi$ supported in $M^+$, whereupon
\begin{flalign}
\tau^{+}_{(\bfM,\bfJ)}[\h,\j]([(\varphi,\alpha)]_{_{(\bfM,\bfJ)}})= [(\varphi,\alpha)]_{_{(\bfM[\h],\bfJ[\j])}}~.
\end{flalign}
In turn, given a representative
$(\varphi^\prime,\alpha^\prime)\in [(\varphi,\alpha)]_{_{(\bfM[\h],\bfJ[\j])}}$
so that $\varphi^\prime$ has support in $M^-$, the relative Cauchy evolution of $[(\varphi,\alpha)]_{_{(\bfM,\bfJ)}}$ is 
\begin{flalign}
\mathrm{rce}_{(\bfM,\bfJ)}^{(\PhSp_p)}[\h,\j]\big([(\varphi,\alpha)]_{_{(\bfM,\bfJ)}}\big) =\big(\tau^{-}_{(\bfM,\bfJ)}[\h,\j]\big)^{-1}([(\varphi,\alpha)]_{_{(\bfM[\h],\bfJ[\j])}})=
 [(\varphi^\prime,\alpha^\prime)]_{_{(\bfM,\bfJ)}}~.
\end{flalign}
Thus it remains to find a suitable representative $(\varphi',\alpha')$. 
By a standard argument, see e.g.~\cite[Lemma 3.1.]{Fewster:2011pn}, we can find a smooth function $\chi\in C^\infty(M)$,
such that $\varphi^\prime = \varphi - \mathrm{KG}_{\bfM[\h]} \big(\chi\,\mathrm{E}^-_{\bfM[\h]}(\varphi)\big)$
has  support in $M^-$ and such that $\chi\,\mathrm{E}^-_{\bfM[\h]}(\varphi)$ has compact support.
Explicitly, we take any two Cauchy surfaces $\Sigma^{\pm}$ in $\bfM^{-}[\h,\j]$ such that
$\Sigma^{+}\cap \Sigma^- =\emptyset$, $\Sigma^+$ is in the future of $\Sigma^-$
and $\supp(\varphi)\cup \supp(\h)\cup\supp(\j)$ is in the future of $\Sigma^+$. 
Any $\chi$ such that $\chi\vert_{J^+_{\bfM[\h]}(\Sigma^+)}\equiv 1$
and $\chi\vert_{J^-_{\bfM[\h]}(\Sigma^-)}\equiv 0$ then leads 
by the formula above  to a $\varphi^\prime$ with the desired properties.
Using (\ref{eqn:equivalence}) and now dropping the labels on equivalence classes 
(which from now on are all taken with respect to $(\bfM,\bfJ)$)
we obtain for the relative Cauchy evolution 
\begin{flalign}\label{eqn:tmpcauchy}
\mathrm{rce}_{(\bfM,\bfJ)}^{(\PhSp_p)}[\h,\j]\big([(\varphi,\alpha)]\big) =  
\Big[\Big(\varphi - \mathrm{KG}_{\bfM[\h]}\big(\chi\,\mathrm{E}^-_{\bfM[\h]}(\varphi) \big)
,\alpha - \int_{M} \ip{\bfJ+\j}{\chi\,\mathrm{E}^-_{\bfM[\h]}(\varphi)}\, \vol_{\bfM[\h]} \Big)\Big]~.
\end{flalign}
As $\chi\,\mathrm{E}^-_{\bfM[\h]}(\varphi)$ is compactly supported, 
we may use the equivalence relation with respect to $(\bfM,\bfJ)$
to obtain
\begin{multline}
\Big(\mathrm{rce}_{(\bfM,\bfJ)}^{(\PhSp_p)}[\h,\j]-\id_{\PhSp_p(\bfM,\bfJ)}\Big)\big([(\varphi,\alpha)]\big)  \\
=\Big[\Big((\mathrm{KG}_{\bfM}-\mathrm{KG}_{\bfM[\h]})\big(\chi\,\mathrm{E}^-_{\bfM[\h]}(\varphi) \big)
,\int_{M}
\ip{(1-\rho_{\h}) \bfJ -\rho_{\h}\,\j}{\chi\,\mathrm{E}^-_{\bfM[\h]}(\varphi)} \,\vol_{\bfM}  \Big)\Big]~,
\end{multline}
in which we have also written $\rho_{\h}\in C^\infty_0(M)$ for the unique function such that $\vol_{\bfM[\h]} = \rho_{\h}\,\vol_{\bfM}$,
explicitly $\rho_{\h} = \sqrt{\vert \g+\h\vert}/\sqrt{\vert \g\vert}$.
Noting that $\chi=1$ and $\mathrm{E}^+_{\bfM[\h]}(\varphi)=0$ on $\supp (\j) \cup \supp (\h)$, we may replace both occurrences of
$\chi\,\mathrm{E}^-_{\bfM[\h]}(\varphi)$ by 
$\mathrm{E}_{\bfM[\h]}(\varphi)$, obtaining
\begin{multline}\label{eq:rce_reform}
\Big(\mathrm{rce}_{(\bfM,\bfJ)}^{(\PhSp_p)}[\h,\j]-\id_{\PhSp_p(\bfM,\bfJ)}\Big)\big([(\varphi,\alpha)]\big)  \\
=\Big[\Big((\mathrm{KG}_{\bfM}-\mathrm{KG}_{\bfM[\h]})
\big(\mathrm{E}_{\bfM[\h]}(\varphi) \big)
,\int_{M}  \Big(\ip{-\j}{\mathrm{E}_{\bfM[\h]}(\varphi)}+
\ip{(1-\rho_{\h}) (\bfJ +\j)}{\mathrm{E}_{\bfM[\h]}(\varphi)} \Big)\, \vol_{\bfM} \Big)\Big]~,
\end{multline}
after a further slight rearrangement. Note that \eqref{eq:rce_reform} applies 
only for representatives where $\varphi$ is supported in $M^+$. In this form, it is easy to see
what the functional derivative
of the relative Cauchy evolution with respect to $\h$ and $\j$ {\em ought} to be, 
simply by expanding to first order in these quantities. This procedure gives
\begin{flalign}\label{eq:rce_deriv}
\left.\frac{d}{ds} \mathrm{rce}_{(\bfM,\bfJ)}^{(\PhSp_p)}[s\h,s\j] \big([(\varphi,\alpha)]\big)  \right|_{s=0} &=: -
 \left(\mathcal{T}_{(\bfM,\bfJ)}[\h]+\mathcal{J}_{(\bfM,\bfJ)}[\j]\right)\big([(\varphi,\alpha)]\big)  ~,
\end{flalign}
where  
\begin{subequations}
\begin{flalign}\label{eqn:Tdefinition} 
\mathcal{T}_{(\bfM,\bfJ)}[\h]\big([(\varphi,\alpha)]\big) & = 
\Big[\Big( 
\mathrm{KG}^{\prime}_{\bfM[\h]}\big(\mathrm{E}_{\bfM}(\varphi)\big), \int_M\frac{1}{2}\,g^{ab}\,h_{ab}\,\ip{\bfJ}{\mathrm{E}_{\bfM}(\varphi)} \,\vol_{\bfM} \Big)\Big]~, \\
\label{eqn:Jdefinition}
\mathcal{J}_{(\bfM,\bfJ)}[\j]\big([(\varphi,\alpha)]\big) & = 
\Big[\Big( 0,\int_M\ip{\j}{\mathrm{E}_{\bfM}(\varphi)} \,\vol_{\bfM} \Big)\Big]~,
\end{flalign}
\end{subequations}
and $\mathrm{KG}^{\prime}_{\bfM[\h]} = \frac{d}{ds} \mathrm{KG}_{\bfM[s\h]}\big\vert_{s=0}$.\footnote{Note that
the derivative of the relative Cauchy evolution involves {\em minus} the derivative
of the Klein--Gordon operator. The BFV paper~\cite{Brunetti:2001dx}  contains
an error [or unconventional terminology] on p.61, where an advanced
solution is given support in the causal future of the source, leading to the opposite  
overall sign in the analogous expression for the derivative on p.62 and hence 
in their Theorem~4.3.} Formulae~\eqref{eqn:Tdefinition} and~\eqref{eqn:Jdefinition}
hold for arbitrary representatives $(\varphi,\alpha)$.
Note that elements in $\PhSp_p(\bfM,\bfJ)$ which are of the form $[(0,\alpha)]$, $\alpha\in\bbR$,
are left unchanged under the relative Cauchy evolution. In Appendix~\ref{app:Cauchy}, we shall show how (\ref{eq:rce_deriv})
holds rigorously in the weak-$*$ topology induced by the pairing \eqref{eqn:pairing}. Moreover, we obtain the formula
\begin{flalign}
\sip{\mathcal{T}_{(\bfM,\bfJ)}[\h]\big([(\varphi,\alpha)]\big)}{\phi}_{(\bfM,\bfJ)} =
\frac{1}{2}\left.\frac{d}{ds} \int_{M}h_{ab} \,T_{(\bfM,\bfJ)}^{ab}[\phi+s \,\mathrm{E}_{\bfM}(\varphi)]\, \vol_{\bfM} \right|_{s=0}  ~,\label{eq:tmp_stresstensorfuncderiv_short}
\end{flalign}
where the stress-energy tensor is\footnote{The minus sign before the functional 
derivative appears because we differentiate with respect to the covariant form of the metric.}
\begin{flalign}\label{eqn:fullSET}
T_{(\bfM,\bfJ)}^{ab}[\phi] :=  -\frac{2}{\sqrt{\vert \g\vert }} \frac{\delta S}{\delta g_{ab}(x)}=  \ip{\nabla^a \phi }{\nabla^b \phi} - \frac{1}{2}g^{ab} \ip{\nabla_c \phi }{\nabla^c \phi}
+\frac{1}{2}m^2 g^{ab}\ip{ \phi }{ \phi}  +
g^{ab} \ip{\bfJ}{\phi}~,
\end{flalign}
and $S$ is the classical action obtained from the Lagrangian 
\eqref{eqn:actionintro}.
Similarly, it is clear from \eqref{eqn:Jdefinition} that 
\begin{flalign} 
\sip{\mathcal{J}_{(\bfM,\bfJ)}[\j] \big([(\varphi,\alpha)]\big)}{\phi}_{(\bfM,\bfJ)}
= \left.\frac{d}{ds}\int_M \ip{\j}{\phi+s \,\mathrm{E}_{\bfM}(\varphi)}\, \vol_{\bfM}\right|_{s=0}~.
\end{flalign}
These formulae establish a close link between the relative Cauchy evolution and the action; indeed,
\begin{flalign}\label{eq:rce_deriv2}
\left.\frac{d}{ds}\sip{\mathrm{rce}_{(\bfM,\bfJ)}^{(\PhSp_p)}[s\h,s\j] \big([(\varphi,\alpha)]\big) }{\phi}_{(\bfM,\bfJ)} \right|_{s=0} &= 
\frac{\delta^2 S}{\delta \phi\delta \g}\big(\mathrm{E}_{\bfM}(\varphi)\otimes\h\big) + 
\frac{\delta^2 S}{\delta \phi\delta \bfJ}\big(\mathrm{E}_{\bfM}(\varphi)\otimes\j\big)~,
\end{flalign}
where the functional derivatives are evaluated at $\phi\in\Sol_p(\bfM,\bfJ)$, and 
on the background $(\bfM,\bfJ)$, and we differentiate with respect to the covariant metric tensor.  
\sk

At this point the following remark is in order: The stress-energy tensor (\ref{eqn:fullSET}) is {\it not} covariantly conserved
for generic $(\bfM,\bfJ)$ and generic solutions $\phi$ of the inhomogeneous Klein--Gordon equation, since
\begin{flalign}\label{eqn:tmp_nonconservationlaw}
\nabla_a T^{ab}_{(\bfM,\bfJ)}[\phi] = \ip{\nabla^b\bfJ}{\phi}~.
\end{flalign}
Modifying $T^{ab}_{(\bfM,\bfJ)}$ by adding a constant functional, which would not change the derivative of the relative Cauchy evolution
given in (\ref{eq:tmp_stresstensorfuncderiv_short}), does not change this fact.
Repeating the arguments given in \cite[\S 4]{Brunetti:2001dx}, the non-conservation law
(\ref{eqn:tmp_nonconservationlaw})  (up to constant functionals) can also be derived directly from the relative Cauchy evolution.
This perhaps unpleasant feature can be understood as follows: diffeomorphism
invariance of the action derived from \eqref{eqn:actionintro}  entails the identity
\begin{flalign}
\frac{\delta S}{\delta\g}(\pounds_X \g) + 
\frac{\delta S}{\delta\bfJ}(\pounds_X \bfJ) +
\frac{\delta S}{\delta\phi}(\pounds_X \phi) = 0
\end{flalign}
for all compactly supported vector fields $X$. When $\phi$ is on-shell, the
last term vanishes and the identity implies \eqref{eqn:tmp_nonconservationlaw}.
We cannot expect conservation of the stress-energy tensor in our theory, because
$\bfJ$ is non-dynamical; indeed \eqref{eqn:tmp_nonconservationlaw} is the correct generalized 
conservation law in this case. (Were we to modify the theory, so that $\bfJ$ became  dynamical,
 the additional Euler--Lagrange equation $\phi=0$  
would rather trivially restore conservation of the stress-energy tensor.)


\section{\label{sec:automorphisms}Automorphism group of the functor $\PhSp_p$}
Given any covariant functor from $\LocSrc_p$ to $\PreSymp$ it is interesting
to study its endomorphisms and automorphisms \cite{Fewster:2012yc}. The latter typically
sheds light on possible symmetries of the theory at the functorial level,
which is comparable to the global gauge group of Minkowski algebraic quantum field theory.
In \cite{Fewster:2012yc}, the automorphism group of a theory describing 
a multiplet of $p\in\bbN$ classical real scalar fields satisfying the minimally coupled Klein--Gordon equation
was found to be the orthogonal group $O(p)$ if all masses coincide and are nonzero, or $O(p)\ltimes\mathbb{R}^p$ if they all vanish. 
As mentioned in the Introduction, we expect the source terms in the inhomogeneous Klein--Gordon theory to break 
(at least for the massive case $m\neq 0$)
all the symmetries of the homogeneous Klein--Gordon theory.  It therefore comes as a surprise 
that the functor $\PhSp_p$ turns out to have a nontrivial automorphism group for any mass $m$.
\sk

We shall briefly fix some notation. Given any covariant functor $\mathfrak{F} : \LocSrc_p \to \PreSymp$,
an endomorphism of $\mathfrak{F}$ is a natural transformation $\eta : \mathfrak{F} \Rightarrow \mathfrak{F}$, 
i.e.~a collection of morphisms  $\eta_{(\bfM,\bfJ)}: \mathfrak{F}(\bfM,\bfJ) \to \mathfrak{F}(\bfM,\bfJ)$
in $\PreSymp$, such that for any morphism $f: (\bfM_1,\bfJ_1)\to (\bfM_2,\bfJ_2)$ in $\LocSrc_p$ the following 
diagram commutes
\begin{flalign}\label{eqn:endodiagram}
\xymatrix{
\ar[d]_-{\eta_{(\bfM_1,\bfJ_1)}} \mathfrak{F}(\bfM_1,\bfJ_1)  \ar[rr]^-{\mathfrak{F}(f)} && \mathfrak{F}(\bfM_2,\bfJ_2) \ar[d]^-{\eta_{(\bfM_2,\bfJ_2)}}\\
\mathfrak{F}(\bfM_1,\bfJ_1) \ar[rr]^-{\mathfrak{F}(f)}&& \mathfrak{F}(\bfM_2,\bfJ_2) 
}
\end{flalign}
We denote the collection of all endomorphisms of $\mathfrak{F}$ by $\End(\mathfrak{F})$.
An automorphism of $\mathfrak{F}$ is an endomorphism $\eta\in \End(\mathfrak{F})$, such that
all $\eta_{(\bfM,\bfJ)}$ are isomorphisms. Under composition, the automorphisms of $\mathfrak{F}$
form a group denoted by $\Aut(\mathfrak{F})$.
\sk

The goal of this section is to characterize the automorphism group of the functor
 $\PhSp_p:\LocSrc_p\to \PreSymp$ introduced in Section \ref{sec:prelim}. Due to the 
 following general statement,  $\Aut(\PhSp_p)$ is nontrivial.
\begin{propo}\label{propo:Z2flip}
Let $\mathfrak{F}: \LocSrc_p \to \PreSymp$ be any covariant functor. 
Then there exists a faithful homomorphism $\eta:\bbZ_2 \to \Aut(\mathfrak{F})$
given by $\eta(\sigma) = \{\sigma\,\id_{\mathfrak{F}(\bfM,\bfJ)}\} $, where 
$\sigma\in \bbZ_2 = \{-1,+1\}$.
\end{propo}
\begin{proof}
Injectivity of $\eta$ and the group law $\eta(\sigma)\circ \eta(\sigma^\prime) = \eta(\sigma\,\sigma^\prime)$ are obvious.
All $\eta(\sigma)_{(\bfM,\bfJ)}$ are clearly linear automorphisms and since $\sigma^2=1$ they preserve
the presymplectic structure in $\mathfrak{F}(\bfM,\bfJ)$ (this follows from bilinearity of any presymplectic structure).
For any morphism $f$ in $\LocSrc_p$, $\eta(\sigma)$ satisfies the diagram in (\ref{eqn:endodiagram}), 
since $\mathfrak{F}(f)$ are in particular linear maps and hence commute with the multiplication by $\sigma$.
\end{proof}

The previous proposition in particular shows that $\Aut(\PhSp_p)$ contains a $\mathbb{Z}_2$ subgroup
for all values of $m$. In the massless case we can say more. 
\begin{propo}\label{propo:Z2xR}
If $m=0$ there exists a faithful homomorphism $\eta:\mathbb{Z}_2\times\bbR^p \to \Aut(\PhSp_p)$
given by $\eta(\sigma,\mu) = \{\eta(\sigma,\mu)_{(\bfM,\bfJ)}\}$, where, for all $[(\varphi,\alpha)]\in \PhSp_p(\bfM,\bfJ)$,
\begin{flalign}
\eta(\sigma,\mu)_{(\bfM,\bfJ)}\big([(\varphi,\alpha)] \big) =\left[\left(\sigma\, \varphi,\sigma\,\alpha+  \sigma\int_{M} \ip{\varphi}{\mu} \, \vol_{\bfM} \right)\right]~.
\label{eqn:Z2xR}
\end{flalign}
Here we have identified $\mu\in \bbR^p$  with the corresponding constant function in $C^\infty(M,\bbR^p)$.
\end{propo}
\begin{proof} 
The main burden is to show that \eqref{eqn:Z2xR} does define a 
natural $\eta(\sigma,\mu)\in\End(\PhSp_p)$ for each $(\sigma,\mu)\in\bbZ_2\times\bbR^p$, because 
injectivity of $\eta$ is obvious and it is easy to establish the group law $\eta(\sigma,\mu)\circ \eta(\sigma^\prime,\mu^\prime) 
= \eta(\sigma\,\sigma^\prime,\mu+\mu^\prime)$, whereupon it is clear that each $\eta(\sigma,\mu)$ is a linear automorphism.
We notice that $\eta(\sigma,\mu)_{(\bfM,\bfJ)}$ is compatible with the quotient in $\PhSp_p(\bfM,\bfJ)$, since,
for all $h\in C^\infty_0(M,\bbR^p)$,
\begin{flalign}
\nn\eta(\sigma,\mu)_{(\bfM,\bfJ)}\big(\mathrm{P}^\ast_{(\bfM,\bfJ)}(h)\big)&=\eta(\sigma,\mu)_{(\bfM,\bfJ)}\left(\mathrm{KG}_{\bfM}(h),\int_M\ip{\bfJ}{h} \,\vol_{\bfM}\right) \\
\nn &= \left(\sigma \,\mathrm{KG}_{\bfM}(h) ,\sigma \int_M \ip{\bfJ}{h} \,\vol_{\bfM}+ \sigma\int_{M} \ip{\mathrm{KG}_{\bfM}(h)}{\mu}\, \vol_{\bfM}\right) \\
&= \mathrm{P}^\ast_{(\bfM,\bfJ)}(\sigma\,h)~,
\end{flalign}
where in the last equality we have used that $\int_{M}\ip{\mathrm{KG}_{\bfM}(h)}{\mu} \,\vol_{\bfM} = 
\int_{M} \ip{h}{\mathrm{KG}_{\bfM}(\mu)} \, \vol_{\bfM} =0$ for the massless Klein--Gordon operator.
It is easily seen that the linear map $\eta(\sigma,\mu)_{(\bfM,\bfJ)}$ preserves the presymplectic 
structure in $\PhSp_p(\bfM,\bfJ)$ and that it is injective (indeed, invertible, as already mentioned). 
Thus  $\eta(\sigma,\mu)_{(\bfM,\bfJ)}\in\Aut(\PhSp_p(\bfM,\bfJ))$. 

Now suppose that $f: (\bfM_1,\bfJ_1)\to
(\bfM_2,\bfJ_2)$ is a morphism in $\LocSrc_p$. Then,
for all $[(\varphi,\alpha)]\in \PhSp_p(\bfM_1,\bfJ_1)$,
\begin{flalign}
\eta(\sigma,\mu)_{(\bfM_2,\bfJ_2)}\circ \PhSp_p(f)\big(
 \big[(\varphi,\alpha)\big]\big)  &= 
 \left[\left(\sigma\,f_\ast(\varphi), \sigma\,\alpha+\sigma \int_{M_2} 
 \ip{f_\ast(\varphi)}{\mu} \, \vol_{\bfM_2}\right)\right] \nonumber \\ 
 &=   \PhSp_p(f) \circ\eta(\sigma,\mu)_{(\bfM_1,\bfJ_1)}
\big( \big[(\varphi,\alpha)\big]\big)~,
\end{flalign}
because $\int_{M_2} \ip{f_\ast(\varphi)}{\mu} \,\vol_{\bfM_2}=
\int_{M_1}  \ip{\varphi}{f^\ast(\mu)} \,\vol_{\bfM_1}=
\int_{M_1} \ip{\varphi}{\mu} \,\vol_{\bfM_1}$. Hence, naturality is proved.  
\end{proof}

Our aim is now to
prove that we actually have an isomorphism $ \Aut(\PhSp_p)\simeq \mathbb{Z}_2$ in the massive case
and an isomorphism $\Aut(\PhSp_p)\simeq \mathbb{Z}_2\times \bbR^p$ for $m=0$. For this endeavor 
we generalize the results of \cite[\S 2.2.]{Fewster:2012yc}, which shall allow us to prove that every
endomorphism $\eta\in \End(\PhSp_p)$ is uniquely determined by its component $\eta_{(\bfM,\bfJ)}$ on any given
object $(\bfM,\bfJ)$ in $\LocSrc_p$.
\sk

We start by collecting some useful lemmas, which were obtained in \cite[\S  2.2.]{Fewster:2012yc} for the 
category $\Loc$.
\begin{lem}\label{lem:endorcecompatibility}
Let $\eta\in \End(\PhSp_p)$ be any endomorphism and $(\bfM,\bfJ)$ any object in $\LocSrc_p$.
Then for all globally hyperbolic perturbations $(\h,\j)\in H(\bfM,\bfJ)$ we have that
\begin{flalign}
\mathrm{rce}_{(\bfM,\bfJ)}^{(\PhSp_p)}[\h,\j] \circ \eta_{(\bfM,\bfJ)} = \eta_{(\bfM,\bfJ)} \circ \mathrm{rce}_{(\bfM,\bfJ)}^{(\PhSp_p)}[\h,\j] ~.
\end{flalign}
\end{lem}
\begin{proof}
The essential idea is naturality of $\eta$. The detailed steps can be found in
\cite[Proposition 3.8.]{Fewster:2011pe}.
\end{proof}
\begin{lem}\label{lem:endoproperties}
Let $\eta,\eta^\prime\in \End(\PhSp_p)$ and suppose that $\eta_{(\bfM,\bfJ)} = \eta^\prime_{(\bfM,\bfJ)}$
for some object $(\bfM,\bfJ)$ in $\LocSrc_p$. Then the following hold true:
\begin{itemize}
\item[(i)] If $f: (\boldsymbol{L},\bfJ_{\boldsymbol{L}})\to (\bfM,\bfJ)$ is a morphism in $\LocSrc_p$, then
$\eta_{(\boldsymbol{L},\bfJ_{\boldsymbol{L}})} = \eta^\prime_{(\boldsymbol{L},\bfJ_{\boldsymbol{L}})} $.
\item[(ii)] If $f: (\bfM,\bfJ) \to  (\boldsymbol{N},\bfJ_{\boldsymbol{N}})$ is a Cauchy morphism in
$\LocSrc_p$, then $\eta_{(\boldsymbol{N},\bfJ_{\boldsymbol{N}})} = \eta^\prime_{(\boldsymbol{N},\bfJ_{\boldsymbol{N}})} $.
\item[(iii)] $\eta_{(\boldsymbol{L},\bfJ_{\boldsymbol{L}})} = \eta^\prime_{(\boldsymbol{L},\bfJ_{\boldsymbol{L}})} $ for any object
$(\boldsymbol{L},\bfJ_{\boldsymbol{L}})$ in $\LocSrc_p$, such that the Cauchy surfaces of $\boldsymbol{L}$
are oriented diffeomorphic to those of $\bfM\vert_{O}$, for some $O\in \mathscr{O}(\bfM)$. Here
$\mathscr{O}(\bfM)$ is the set of all causally compatible, open and globally hyperbolic subsets of $\bfM$
with finitely many connected components all of which are mutually causally disjoint.
\end{itemize}
\end{lem}
\begin{proof}
Item (i) and (ii) are simple consequences of naturality of $\eta,\eta^\prime$ and the fact that
$\PhSp_p(f)$ is injective for (i) and even an isomorphism for (ii), cf.~\cite[Proof of Lemma 2.2.]{Fewster:2012yc}.
 Item (iii) follows from a generalization of the ``Cauchy wedge connectedness'' property to the category
$\LocSrc_p$ that we shall discuss now in detail. Forgetting the source terms, more precisely, applying the forgetful functor
from $\LocSrc_p$ to $\Loc$, it was shown in \cite[Proposition 2.4.]{Fewster:2011pe} that under our hypotheses 
there is a chain of morphisms in $\Loc$
\begin{flalign}\label{eqn:cwcon}
\xymatrix{
\boldsymbol{L} & \ar[l]_-{c} \boldsymbol{L}^\prime \ar[r]^-{c}& \boldsymbol{L}^{\prime\prime} &\ar[l]_-{c} 
\boldsymbol{L}^{\prime\prime\prime} \ar[r]^-{c} & \bfM\vert_{O} \ar[r]^-{\iota_{\bfM;O}} & \bfM
}~,
\end{flalign}
where $\iota_{\bfM;O}$ denotes the canonical inclusion and by `$c$' we indicate Cauchy morphisms.
If we could construct from this chain a chain of morphisms in $\LocSrc_p$ of the form 
\begin{flalign}\label{eqn:cwcon2}
\xymatrix{
(\boldsymbol{L},\bfJ_{\boldsymbol{L}}) & \ar[l]_-{c} (\boldsymbol{L}^\prime,\bfJ_{\boldsymbol{L}^\prime}) \ar[r]^-{c}& 
(\boldsymbol{L}^{\prime\prime},\bfJ_{\boldsymbol{L}^{\prime\prime}}) &\ar[l]_-{c} 
(\boldsymbol{L}^{\prime\prime\prime},\bfJ_{\boldsymbol{L}^{\prime\prime\prime}}) \ar[r]^-{c} 
& (\bfM\vert_{O},\bfJ\vert_O) \ar[r]^-{\iota_{\bfM;O}} & (\bfM,\bfJ)
}~,
\end{flalign}
then the same argument as in \cite[Proof of Lemma 2.2.]{Fewster:2012yc} (combining results (i) and (ii))
 would prove (iii). Since the morphisms in (\ref{eqn:cwcon}) uniquely fix
$\bfJ_{\boldsymbol{L}^\prime}$ and $\bfJ_{\boldsymbol{L}^{\prime\prime\prime}}$
 via pulling back, respectively, $\bfJ_{\boldsymbol{L}}$ and $\bfJ\vert_{O}$, the remaining step is to show that
 we can construct a $\bfJ_{\boldsymbol{L}^{\prime\prime}}$ completing the chain (\ref{eqn:cwcon2}).
This is indeed possible if we equip the spacetime $\boldsymbol{L}^{\prime\prime}$
constructed in \cite[Proof of Proposition 2.4.]{Fewster:2011pe} with a $\bfJ_{\boldsymbol{L}^{\prime\prime}}$
that is obtained by gluing $\bfJ_{\boldsymbol{L}^{\prime}}$ and $\bfJ_{\boldsymbol{L}^{\prime\prime\prime}}$
via a suitable cutoff function (e.g.\ the one used to construct
the metric in $ {\boldsymbol{L}^{\prime\prime}}$).
\end{proof}
We are now ready to prove the analog of  \cite[Theorem 2.6.]{Fewster:2012yc} for our specific functor 
$\PhSp_p:\LocSrc_p\to \PreSymp$. Since we are working with a specific model we can give a direct proof and we 
 do not have to dwell on the technicalities concerning abstract categorical unions
 and equalizers appearing in \cite{Fewster:2012yc}.
\begin{theo}\label{theo:endobycomponent}
Every $\eta\in \End(\PhSp_p)$ is uniquely determined by its component $\eta_{(\bfM,\bfJ)}$
on any object $(\bfM,\bfJ)$ in $\LocSrc_p$.
\end{theo}
\begin{proof}
Suppose that $\eta^\prime \in \End(\PhSp_p)$ agrees with $\eta$ on the object $(\bfM,\bfJ)$ in $\LocSrc_p$, 
i.e.~$\eta_{(\bfM,\bfJ)} = \eta^\prime_{(\bfM,\bfJ)}$. 
Let $(\boldsymbol{N},\bfJ_{\boldsymbol{N}})$ be any object in $\LocSrc_p$
and let $D$ be any diamond in $\boldsymbol{N}$. Then $\boldsymbol{N}\vert_D$
has Cauchy surfaces which are oriented diffeomorphic to any diamond in $\bfM$. Then
 $\eta_{(\boldsymbol{N}\vert_D,\bfJ_{\boldsymbol{N}}\vert_D)} = 
\eta_{(\boldsymbol{N}\vert_D,\bfJ_{\boldsymbol{N}}\vert_D)}^\prime $ by Lemma \ref{lem:endoproperties} (iii).
Using the canonical inclusion $\iota_{\boldsymbol{N};D}: (\boldsymbol{N}\vert_D,\bfJ_{\boldsymbol{N}}\vert_D) \to
 (\boldsymbol{N},\bfJ_{\boldsymbol{N}})$ we obtain by naturality
 \begin{flalign}
 \nn \eta_{(\boldsymbol{N},\bfJ_{\boldsymbol{N}})} \circ \PhSp_p(\iota_{\boldsymbol{N};D})
 &= \PhSp_p(\iota_{\boldsymbol{N};D})\circ \eta_{(\boldsymbol{N}\vert_D,\bfJ_{\boldsymbol{N}}\vert_D)} 
 =\PhSp_p(\iota_{\boldsymbol{N};D})\circ \eta^\prime_{(\boldsymbol{N}\vert_D,\bfJ_{\boldsymbol{N}}\vert_D)} \\
 &=\eta^\prime_{(\boldsymbol{N},\bfJ_{\boldsymbol{N}})} \circ \PhSp_p(\iota_{\boldsymbol{N};D})~.
 \end{flalign}
This equation implies that $\eta_{(\boldsymbol{N},\bfJ_{\boldsymbol{N}})} \big([(\varphi,\alpha)]\big)
=\eta^\prime_{(\boldsymbol{N},\bfJ_{\boldsymbol{N}})} \big([(\varphi,\alpha)]\big)$, 
for all $[(\varphi,\alpha)] \in \PhSp_p(\boldsymbol{N},\bfJ_{\boldsymbol{N}})$
for which there exists a representative $(\varphi,\alpha)$ such that $\varphi$ has compact support in $D$.
We shall now prove that $\eta_{(\boldsymbol{N},\bfJ_{\boldsymbol{N}})} \big([(\varphi,\alpha)]\big)
=\eta^\prime_{(\boldsymbol{N},\bfJ_{\boldsymbol{N}})} \big([(\varphi,\alpha)]\big)$, 
for all $[(\varphi,\alpha)] \in \PhSp_p(\boldsymbol{N},\bfJ_{\boldsymbol{N}})$, and hence that $\eta^\prime =\eta$ 
since $(\boldsymbol{N},\bfJ_{\boldsymbol{N}})$ was arbitrary. This proof follows from a partition of unity argument:
Given any $[(\varphi,\alpha)] \in \PhSp_p(\boldsymbol{N},\bfJ_{\boldsymbol{N}})$ we take some representative
$(\varphi,\alpha) \in C_0^\infty(N,\bbR^p)\oplus \bbR$. Since $\supp(\varphi)$ is compact,
the open cover of $N$ given by the set of all diamonds in $N$ has a finite subcover of $\supp(\varphi)$, which we denote by 
$\{D_i\}_{i=1,\dots,n}$. Using a partition of unity subordinated to $\{D_i\}_{i=1,\dots,n}$
we can write $(\varphi,\alpha) = \sum_{i=1}^n (\varphi_i,\alpha/n)$, where $\varphi_i$ has compact support in
$D_i$. Hence, 
\begin{flalign}
\nn \eta_{(\boldsymbol{N},\bfJ_{\boldsymbol{N}})} \big([(\varphi,\alpha)]\big) &= \sum\limits_{i=1}^n 
\eta_{(\boldsymbol{N},\bfJ_{\boldsymbol{N}})} \big([(\varphi_i,\alpha/n)]\big) 
= \sum\limits_{i=1}^n 
\eta^\prime_{(\boldsymbol{N},\bfJ_{\boldsymbol{N}})} \big([(\varphi_i,\alpha/n)]\big) \\
&=\eta^\prime_{(\boldsymbol{N},\bfJ_{\boldsymbol{N}})} \big([(\varphi,\alpha)]\big) ~.
\end{flalign}
\end{proof}

We now come to the main statement of this section. 
\begin{theo} \label{theo:automorphismgroup}
Every endomorphism of the functor $\PhSp_p$ is an automorphism and
\begin{flalign}
\End(\PhSp_p) = \Aut(\PhSp_p)\simeq 
\begin{cases} \mathbb{Z}_2 & ~,~~\text{for }m\neq 0~,\\  \mathbb{Z}_2 \times\bbR^p &~,~~\text{for } m= 0~,\end{cases}
\end{flalign}
where the action is given for $m\neq 0$ by Proposition \ref{propo:Z2flip}, and for
$m=0$ by Proposition~\ref{propo:Z2xR}.
\end{theo}
\begin{proof}
Due to Theorem \ref{theo:endobycomponent}, any $\eta\in\End(\PhSp_p)$ is uniquely determined by
its component $\eta_{(\bfM_0,0)}$, where $\bfM_0$ is Minkowski spacetime and we have chosen $\bfJ_0 =0$.
The presymplectic vector space $\PhSp_p(\bfM_0,0)$ can be expressed 
more conveniently using the following linear isomorphism for the underlying vector space
\begin{flalign}
\nn \big(C^\infty_0(M_0,\bbR^p)\oplus\bbR\big)/\mathrm{P}_{(\bfM_0,0)}^\ast\big[C_0^\infty(M_0,\bbR^p)\big] &= 
\big(C^\infty_0(M_0,\bbR^p)\oplus\bbR\big)/\big(\mathrm{KG}_{\bfM_0}\big[C^\infty_0(M_0,\bbR^p)\big]\oplus \{0\}\big)\\
\nn &=\big(C^\infty_0(M_0,\bbR^p)/\mathrm{KG}_{\bfM_0}\big[C^\infty_0(M_0,\bbR^p)\big]\big)\oplus \bbR\\
\label{eqn:tmpsolspaceiso}&\simeq \Sol_\mathrm{sc}(\bfM_0) \oplus \bbR~,
\end{flalign}
where $\Sol_\mathrm{sc}(\bfM_0) := \{\phi\in C^\infty_\mathrm{sc}(M_0,\bbR^p): \mathrm{KG}_{\bfM_0}(\phi) =0\}$
is the space of solutions of the Klein--Gordon equation with spacelike compact support. The isomorphism in the last line of
 (\ref{eqn:tmpsolspaceiso}) is the usual one given
by the advanced-minus-retarded Green's operator $\mathrm{E}_{\bfM_0}$.
The induced presymplectic structure on $\Sol_\mathrm{sc}(\bfM_0) \oplus \bbR$ is given by,
for all $(\phi,\alpha),(\psi,\beta)\in \Sol_\mathrm{sc}(\bfM_0) \oplus \bbR$,
\begin{flalign}
\sigma_{(\bfM_0,0)}\big((\phi,\alpha),(\psi,\beta)\big) = \widetilde{\sigma}_{\bfM_0}(\phi,\psi)~,
\end{flalign}
where $\widetilde{\sigma}_{\bfM_0}$ is the usual symplectic structure on $\Sol_{\mathrm{sc}}(\bfM_0)$. 
\sk

Via the isomorphism \eqref{eqn:tmpsolspaceiso}, $\eta_{(\bfM_0,0)}$ induces an 
endomorphism $\widetilde{\eta}$ of  $\Sol_{\mathrm{sc}}(\bfM_0)\oplus \bbR$ which, by naturality,
commutes with the action of all Poincar\'e transformations. By  Lemma \ref{lem:endorcecompatibility},
 $\widetilde{\eta}$ commutes with the relative Cauchy evolution,
and therefore with its derivatives given in (\ref{eqn:Tdefinition}) and (\ref{eqn:Jdefinition}).
Taking into account the isomorphism (\ref{eqn:tmpsolspaceiso}) and our specific choice of object $(\bfM_0,0)$
they read, for all $(\phi,\alpha)\in \Sol_\mathrm{sc}(\bfM_0) \oplus \bbR$,
\begin{subequations}\label{eqn:tmpTandJexpliciteMinkowski}
\begin{flalign}
\mathcal{T}_{(\bfM_0,0)}[\h](\phi,\alpha)&=\Big(\mathrm{E}_{\bfM_0}\big(\mathrm{KG}^\prime_{\bfM_0[\h]}(\phi)\big),0\Big)~,\\
\mathcal{J}_{(\bfM_0,0)}[\j](\phi,\alpha) &= \Big(0,\int_{M_0}\ip{\j}{\phi} \, \vol_{\bfM_0}\Big)~.
\end{flalign}
\end{subequations}
Since $\widetilde{\eta}: \Sol_{\mathrm{sc}}(\bfM_0)\oplus \bbR \to \Sol_{\mathrm{sc}}(\bfM_0)\oplus \bbR$ is a linear
map, it decomposes into linear maps $L_{11} : \Sol_{\mathrm{sc}}(\bfM_0)\to \Sol_{\mathrm{sc}}(\bfM_0)$,
 $L_{12}: \bbR\to \Sol_{\mathrm{sc}}(\bfM_0)$, $L_{21}: \Sol_{\mathrm{sc}}(\bfM_0)\to \bbR$ 
 and $L_{22}: \bbR\to \bbR$.
As $\widetilde{\eta}$ commutes with the maps in (\ref{eqn:tmpTandJexpliciteMinkowski}), 
we obtain the following conditions on the $L_{ij}$, 
for all  $(\phi,\alpha)\in \Sol_\mathrm{sc}(\bfM_0) \oplus \bbR$ and $(\h,\j)\in H(\bfM_0,0)$,
\begin{subequations}
\begin{flalign}\label{eqn:tmpconditionT}
\Big(\mathrm{E}_{\bfM_0}\big(\mathrm{KG}_{\bfM_0[\h]}^\prime\big(L_{11}(\phi) + L_{12}(\alpha)\big)\big),0\Big)
= \Big(L_{11}\big(\mathrm{E}_{\bfM_0}\big(\mathrm{KG}_{\bfM[\h]}^\prime(\phi)\big)\big),L_{21}\big(\mathrm{E}_{\bfM_0}\big(\mathrm{KG}_{\bfM[\h]}^\prime(\phi)\big)\big)\Big)
\end{flalign}
and
\begin{flalign}\label{eqn:tmpconditionJ}
\Big(0,\int_{M_0} \ip{\j}{L_{11}(\phi) + L_{12}(\alpha)}\, \vol_{\bfM_0} \Big)  = 
\Big(L_{12}\Big(\int_{M_0}\ip{\j}{\phi}\, \vol_{\bfM_0} \Big), L_{22}\Big(\int_{M_0} \ip{\j}{\phi} \, \vol_{\bfM_0}\Big)\Big)~.
\end{flalign}
\end{subequations}
From the first component of (\ref{eqn:tmpconditionJ}) we obtain that $L_{12} =0$.
Substituting into (\ref{eqn:tmpconditionT}) we obtain from the first component
 the same condition that is present for multiplets of  
homogeneous Klein--Gordon fields. The only solution of this condition (supplemented by additional
conditions stemming from the Poincar{\'e} invariance of Minkowski spacetime)
is that $L_{11}$ is an $O(p)$ transformation acting in the obvious way on the components of
 $\phi$, cf.~\cite[Theorem 5.2.]{Fewster:2012yc}. Using this information the second component
 of (\ref{eqn:tmpconditionJ}) implies that $L_{11} = \sigma\, \id_{\Sol_{\mathrm{sc}}(\bfM_0)}$ and $L_{22} =\sigma\, \id_{\bbR}$, with
 $\sigma\in\bbZ_2 = \{-1,+1\}$.
Finally, the fact that endomorphisms commute with the Poincar{\'e} transformations 
entails that $L_{21}: \Sol_{\mathrm{sc}}(\bfM_0)\to \bbR$ 
is Poincar{\'e} invariant. By Lemma~\ref{lem:transinvfnl} below, we deduce that $L_{21}=0$ for 
$m\neq 0$, or that $L_{21}(\,\cdot\,) =\widetilde{\sigma}_{\bfM_0}(\mu,\,\cdot\,)$ for some $\mu\in\bbR^p$ if $m=0$
 (here $\mu$ denotes a constant solution).
Combining these facts, we have
\begin{flalign}
\widetilde{\eta}(\phi,\alpha) = \begin{cases} (\sigma\,\phi,\sigma\,\alpha) &~,~~\text{for } m\neq 0~,\\
(\sigma\,\phi,\sigma\,(\alpha+\widetilde{\sigma}_{\bfM_0}(\mu,\phi))) & ~,~~\text{for }m=0~,
\end{cases}
\end{flalign}
for some $\sigma\in\bbZ_2$ (and $\mu\in\bbR^p$ if $m=0$). Undoing the isomorphism  \eqref{eqn:tmpsolspaceiso},
this means that for $m=0$  we have $\eta_{(\bfM_0,0)} = \eta(\sigma,\mu)_{(\bfM_0,0)}$ and 
hence $\eta  = \eta(\sigma,\mu)$ by Theorem~\ref{theo:endobycomponent}.
Similarly, $\eta= \eta(\sigma) = \{\sigma\,\id_{\PhSp_p(\bfM,\bfJ)}\}$ if $m\neq 0$. This proves the result. 
\end{proof}

It remains to prove the following
\begin{lem} \label{lem:transinvfnl}
Suppose $L:\Sol_{\mathrm{sc}}(\bfM_0)\to \bbR$ is linear and translationally invariant. 
If $m\neq 0$ then $L=0$. If $m=0$ then there exists $\mu\in\bbR^p$ such that
$L(\phi) =\widetilde{\sigma}_{\bfM_0}(\mu,\phi)$.
\end{lem}
\begin{proof}
We use the automatic continuity result of Meisters~\cite{Meisters:1971} that the translationally invariant linear
 functionals on $C_0^\infty(\bbR^{k},\bbR^p)$ are precisely the scalar multiples of the integral. Passing to Cauchy data
on any surface $t=\text{const.}$, $L$ decomposes into two  linear
functionals on $C_0^\infty(\bbR^{k},\bbR^p)$ that are each translationally invariant, owing to 
spatial translational invariance of $L$. Hence, for each $t$ there are $\alpha_t, \beta_t\in\bbR^p$
such that, for any $\phi\in \Sol_{\mathrm{sc}}(\bfM_0)$,
\begin{flalign}
L(\phi) = \int_{\bbR^k} d^{k}\x \left(\ip{\alpha_t }{\phi(t,\x)} + \ip{\beta_t}{\frac{\partial\phi}{\partial t}(t,\x)}\right)~.
\end{flalign}
By time-translational invariance of $L$, $\alpha_t\equiv\alpha$ and $\beta_t\equiv\beta$ 
are independent of $t$. Further, differentiating with respect to $t$, using the Klein--Gordon equation
and Gauss' theorem, we obtain, for all $\phi\in \Sol_{\mathrm{sc}}(\bfM_0)$,
\begin{flalign}
0 = \int_{\bbR^k} d^k\x \left(\ip{\alpha}{\frac{\partial\phi}{\partial t}(t,\x) }- \ip{\beta  m^2}{\phi(t,\x)} \right)
\end{flalign}
and hence that $\alpha=0$, $\beta m^2 = 0$. The result follows (with $\mu=\beta$). 
\end{proof}

The results of this section reveal that the functor $\PhSp_p$ has a larger automorphism 
group than one would expect for the global gauge group of
the inhomogeneous theory. To close the section, we mention that the $\bbZ_2$ 
factor of $\Aut(\PhSp_p)$ can be removed if we change the category
in which $\PhSp_p$ takes values to reflect the fact that the underlying vector
space of $\PhSp_p(\bfM,\bfJ)$ is a linear subspace of the algebraic affine dual of the space of solutions
 $\Sol_p(\bfM,\bfJ)$ (i.e.\ the vector space of affine maps $\Sol_p(\bfM,\bfJ)\to\bbR$). Were we to restrict to morphisms  arising as
restrictions of duals to affine maps, we would be left with a trivial automorphism group
for $m\neq 0$ and $\bbR^p$ for $m=0$. 
We do not develop this line of thought in detail, because the next section shows
 that $\PhSp_p$ has further pathologies, which would not be eliminated by this device. 
Nevertheless, it is worth noting that the unexpected symmetries appear because 
we have discarded information about the action of the observables in $\PhSp_p(\bfM,\bfJ)$ on the solution space.


\section{\label{sec:compproperty}Violation of the composition property of the functor $\PhSp_p$}
The theory we aim to construct consists of $p$ inhomogeneous Klein--Gordon fields without mutual 
interactions. One would expect that an
equivalent formulation would be produced if the multiplet were decomposed into 
independent submultiplets of $0<q<p$ and $p-q$ fields, which are treated separately 
according to the general prescription and then recombined. In this section, we describe how 
the splitting and recombination may be formalized and then show that the functor $\PhSp_p$ fails to respect this composition property. 
We would like to emphasize that for proving the violation of the composition property 
our functorial language is not strictly necessary; in fact, as we will show in Proposition \ref{propo:compprop}, 
the composition property is already violated on any fixed object $(\bfM,\bfJ)$ in $\LocSrc_p$, even if we choose the Minkowski
space $\bfM_0$ with trivial source $\bfJ_0=0$. However, we find it interesting to develop functorial techniques in order 
to formulate and disprove the composition property, as this point of view is the natural one in the 
spirit of locally covariant field theory and it is flexible enough for axiomatization, which might play a role in future studies in
this field. Moreover, we will show later that the improved formulation
we describe in the following sections does satisfy the composition property in this functorial
sense.
\sk

Let $p\geq 2$ and let $\Pi_q: \bbR^p\to\bbR^p\,,~(a_1,\dots,a_p)\mapsto (a_1,\dots,a_q,0,\dots,0)$ be the projection
onto the first $q$ dimensions, where $0<q<p$. Given any object $(\bfM,\bfJ)$ of $\LocSrc_p$, 
we can split $\bfJ = \Pi_q(\bfJ) + (\id_{\bbR^p}-\Pi_q)(\bfJ) =: \bfJ^q + \bfJ^{p-q}$
and identify $ \bfJ^q$ as an element in $C^\infty(M,\bbR^q)$ and $\bfJ^{p-q}$ as an element in $C^\infty(M,\bbR^{p-q})$.
Hence, we can associate to  $(\bfM,\bfJ)$ the object 
$\mathfrak{Split}_{p,q}(\bfM,\bfJ) :=\big((\bfM,\bfJ^q), (\bfM,\bfJ^{p-q})\big)$ in the product category $\LocSrc_q\times \LocSrc_{p-q}$.
Moreover, given any morphism $f:(\bfM_1,\bfJ_1)\to (\bfM_2,\bfJ_2)$ in $\LocSrc_p$
we associate a morphism in $\LocSrc_q\times \LocSrc_{p-q}$ via
$\mathfrak{Split}_{p,q}(f) := (f,f) : \mathfrak{Split}_{p,q}(\bfM_1,\bfJ_1)\to \mathfrak{Split}_{p,q}(\bfM_2,\bfJ_2)$,
where with a slight abuse of notation we have denoted the smooth map underlying the morphism
by the same symbol $f:M_1\to M_2$. In this way we obtain a covariant
functor $\mathfrak{Split}_{p,q}: \LocSrc_p \to \LocSrc_q\times \LocSrc_{p-q}$ representing the decomposition into submultiplets.
\sk

Treating each submultiplet according to the prescription of Section~\ref{sec:psvsf}, 
we compose $\mathfrak{Split}_{p,q}$ with the covariant functor $\PhSp_q\times \PhSp_{p-q}$
to obtain $\big(\PhSp_q\times \PhSp_{p-q}\big)\circ \mathfrak{Split}_{p,q} : \LocSrc_p \to \PreSymp\times\PreSymp$.
Finally, we recombine the resulting theories by composing with the covariant functor 
$\oplus: \PreSymp\times \PreSymp \to \PreSymp$ that forms the direct sum of two presymplectic vector spaces -- 
the monoidal structure in this category. Explicitly, for any object 
$((V,\sigma_V),(W,\sigma_W))$ in $\PreSymp\times \PreSymp$ we define
$\oplus\big((V,\sigma_V),(W,\sigma_W)\big) := (V\oplus W,\sigma_{V\oplus W})$, where
$V\oplus W$ is the direct sum of vector spaces and, for all $(v,w),(v^\prime,w^\prime)\in V\oplus W$,
\begin{flalign}\label{eqn:sumpresympstructure}
\sigma_{V\oplus W}\big((v,w),(v^\prime,w^\prime)\big) := \sigma_{V}(v,v^\prime) + \sigma_{W}(w,w^\prime)~.
\end{flalign}
On morphisms, $\oplus(L,K):=L\oplus K$. The resulting covariant functor is
\begin{flalign}\label{eqn:PhSppq}
\PhSp_{p,q}:= \oplus\circ \big(\PhSp_q\times\PhSp_{p-q}\big)\circ \mathfrak{Split}_{p,q}:\LocSrc_p \to \PreSymp~.
\end{flalign}

Since the covariant functor $\PhSp_p:\LocSrc_p\to \PreSymp$ satisfies the causality property and the time-slice axiom for
all $p\in\bbN$ (cf.~Proposition \ref{propo:phspproperties}), it is not hard to see that the same holds true for the covariant functor
$\PhSp_{p,q}:\LocSrc_p\to \PreSymp$.
\begin{propo}
The covariant functor $\PhSp_{p,q}:\LocSrc_p\to \PreSymp$ satisfies the causality property and the time-slice axiom
for all $2\leq p\in\bbN$ and $0<q<p$.
\end{propo}
\begin{proof}
Causality holds owing to causality of $\PhSp_{p-q}$ and $\PhSp_q$
and the following property of the direct sum: if  $\PreSymp$-morphisms 
$L_i : (V_i,\sigma_{V_i})\to (V,\sigma_{V})$, $i=1,2$, have symplectically orthogonal images, 
and so do $K_i : (W_i,\sigma_{W_i})\to (W,\sigma_{W})$, $i=1,2$, then $L_1\oplus K_1$ and $L_2\oplus K_2$
have symplectically orthogonal images in $(V\oplus W,\sigma_{V\oplus W})$ because 
$\sigma_{V\oplus W}= \sigma_V\oplus\sigma_W$ by (\ref{eqn:sumpresympstructure}).
The time-slice axiom holds simply because the direct sum of
two presymplectic isomorphisms is itself a presymplectic isomorphism. 
\end{proof}
We shall now prove that the theories $\PhSp_{p,q}$ and $\PhSp_p$ 
are inequivalent for $0<q<p$; that is, the functors are not naturally isomorphic. This will be a consequence of the following simple
\begin{lem}\label{lem:presympiso}
Let $L : (V,\sigma_V)\to (W,\sigma_W)$ be an isomorphism in $\PreSymp$. Then $L$ induces a linear isomorphism 
between the null spaces $N(V,\sigma_V)$ and $N(W,\sigma_W)$.
\end{lem}
\begin{proof}
The $\PreSymp$-isomorphism $L$ induces an injective linear map $L: N(V,\sigma_V)\to W$. The image $L[N(V,\sigma_V)]$
is contained in $N(W,\sigma_W)$, since for all $v\in N(V,\sigma_V)$ and $w\in W$,
\begin{flalign}
\sigma_{W}\big(L(v),w\big)=\sigma_{W}\big(L(v),L(L^{-1}(w))\big)  = \sigma_{V}\big(v,L^{-1}(w)\big)=0~.
\end{flalign}
Hence, $L:  N(V,\sigma_V)\to N(W,\sigma_W)$ is a linear map that is invertible via
$L^{-1}:  N(W,\sigma_W)\to N(V,\sigma_V)$.
\end{proof}
\begin{propo}\label{propo:compprop}
For any $2\leq p\in\bbN$ and $0<q<p$, the covariant functors $\PhSp_{p,q}:\LocSrc_p\to \PreSymp$ 
and $\PhSp_p:\LocSrc_p\to \PreSymp$ are not naturally isomorphic. Indeed, 
there is no object
$(\bfM,\bfJ)$ in $\LocSrc_p$ for which the presymplectic vector spaces $\PhSp_{p,q}(\bfM,\bfJ)$ and $\PhSp_p(\bfM,\bfJ)$
are isomorphic.
\end{propo}
\begin{proof}
We argue by contradiction. Suppose there were a $\PreSymp$-isomorphism 
$L: \PhSp_{p,q}(\bfM,\bfJ)\to \PhSp_p(\bfM,\bfJ)$ for some object $(\bfM,\bfJ)$
in $\LocSrc_p$. Then $L$ induces a linear isomorphism between the null spaces of
$\PhSp_{p,q}(\bfM,\bfJ)$ and $\PhSp_{p}(\bfM,\bfJ)$ by Lemma \ref{lem:presympiso}. However, the latter 
is isomorphic to $\bbR$ (cf.~Proposition \ref{propo:phspproperties}),
while the former is easily seen to be isomorphic to $\bbR^2$. Hence, no such isomorphism exists 
and consequently the functors  $\PhSp_{p,q}$ 
and $\PhSp_p$ are not naturally isomorphic. 
\end{proof}


\section{\label{sec:poisson}The Poisson algebra functor}
In order to resolve the pathological composition property of the functor
$\PhSp_p$ obtained in Section \ref{sec:compproperty}, as well as the mysterious automorphism group
of Theorem \ref{theo:automorphismgroup} (cf.~also the text below Lemma \ref{lem:transinvfnl}),
we introduce further structures. Naively speaking, we aim to make the theory
given by $\PhSp_p$ remember that it describes the affine functionals on the affine space of solutions
of the inhomogeneous Klein--Gordon equation. To realize this idea, we first construct  from 
$\PhSp_p$ a functor describing (abstract) Poisson algebras of observables, which are then represented {\em non-faithfully}
on the solution space. The kernel of this representation is then identified and it is shown that
the quotients of the Poisson algebras by these kernels are described by a covariant functor
which has the desired automorphism group and the composition property.
In Appendix \ref{sec:pointed} we present an alternative strategy for improving the classical theory
of the inhomogeneous multiplet of Klein--Gordon fields by using pointed presymplectic spaces.

\subsection{\label{subsec:canpois}Canonical Poisson algebras}
Recall that a unital Poisson algebra over $\mathbb{R}$ is an associative and commutative unital algebra $\mathcal{A}$ over $\mathbb{R}$
together with a Poisson bracket, i.e.\ a Lie bracket $\{\cdot,\cdot\}_{\mathcal{A}} : \mathcal{A}\times\mathcal{A}\to \mathcal{A}$
which satisfies the derivation property $\{a,b\,c\}_{\mathcal{A}} = \{a,b\}_{\mathcal{A}}\,c + b\,\{a,c\}_{\mathcal{A}}$, for all
$a,b,c\in\mathcal{A}$. A unit-preserving Poisson algebra homomorphism $\kappa : \mathcal{A}\to \mathcal{B}$ between 
two unital Poisson algebras $\mathcal{A},\mathcal{B}$ over $\mathbb{R}$ is a unital algebra homomorphism 
$\kappa : \mathcal{A}\to \mathcal{B}$ that preserves the
Poisson brackets, i.e.\ $\{\cdot,\cdot\}_{\mathcal{B}} \circ (\kappa\times\kappa) = \{\cdot,\cdot\}_{\mathcal{A}}$.
Let us denote by $\PoisAlg$ the
category of unital Poisson algebras over $\mathbb{R}$, with injective unit-preserving Poisson algebra homomorphisms as morphisms. 
We first construct a covariant functor $\CanPois: \PreSymp\to \PoisAlg$ 
that associates canonical Poisson algebras to presymplectic vector spaces:
Given any object $(V,\sigma_V)$ in $\PreSymp$, let us consider the symmetric tensor algebra
$S(V):= \bigoplus_{k=0}^\infty S^k(V)$, where $S^0(V)=\bbR$, $S^1(V)=V$ and $S^k(V):= \bigvee^k V$, for $k\geq 2$, 
is the $k$-th symmetric power of $V$. The product in $S(V)$ is denoted  by juxtaposition and turns
$S(V)$ into an associative and commutative algebra over $\bbR$ with unit $1\in S^0(V)\subset S(V)$.
We define a Poisson bracket $\{\cdot,\cdot\}_{\sigma_V}: S(V)\times S(V) \to S(V)$ by,
for all $\alpha\in S^0(V)$ and $v_1,\dots,v_k,v_1^\prime,\dots,v_l^\prime\in V$,
\begin{subequations}\label{eqn:poissonbracket}
\begin{flalign}
\{\alpha,v_1\cdots v_k\}_{\sigma_V}&=\{v_1\cdots v_k,\alpha\}_{\sigma_V}=0~,\\
\{v_1 \cdots  v_k,v^\prime_1\cdots  v^\prime_l\}_{\sigma_V} &=\sum_{i=1}^k\sum_{j=1}^l
v_1\cdots \omi{i}\cdots  v_k\, v^\prime_1 \cdots \omi{j}\cdots  v_l^\prime ~\sigma_{V}(v_i,v^\prime_j)~.
\end{flalign}
\end{subequations}
The symbols $\omi{i}$ mean the omission of the $i$-th element.
We denote the resulting Poisson algebra  by $\CanPois(V,\sigma_V) := \big(S(V),\{\cdot,\cdot\}_{\sigma_V}\big)$.
Given any morphism $L: (V,\sigma_V)\to (W,\sigma_W)$ in $\PreSymp$
we associate a map
$\CanPois(L): \CanPois(V,\sigma_V)\to \CanPois(W,\sigma_W)$
via $\CanPois(L)(\alpha) = \alpha$, for all $\alpha\in S^0(V)$, and
$\CanPois(L)(v_1\cdots v_k) = L(v_1) \cdots L(v_k)$, for all $v_1,\dots,v_k\in V$.
It is easy to see that $\CanPois(L)$ is an injective Poisson algebra homomorphism.
Thus, we have shown the following
\begin{propo}
The association $\CanPois: \PreSymp\to \PoisAlg$ constructed above is a covariant functor.
\end{propo}
\begin{rem}
Notice that for any object $(V,\sigma_V)$ in $\PreSymp$ the Poisson algebra $\CanPois(V,\sigma_V)$
is $\bbN^0$-graded. Furthermore, for any morphism $L:(V,\sigma_{V})\to (W,\sigma_W)$ in $\PreSymp$
the morphism $\CanPois(L): \CanPois(V,\sigma_V) \to \CanPois(W,\sigma_W)$
is a graded Poisson algebra morphism. Hence, $\CanPois$ is also a covariant functor to the category
of $\bbN^0$-graded Poisson algebras. As will become clear in the next subsection, the 
latter category is too restrictive for our purposes, hence we shall usually disregard this natural grading.
\end{rem}

We can compose our functor $\PhSp_p: \LocSrc_p\to \PreSymp$ 
with $\CanPois:\PreSymp\to \PoisAlg$ 
to obtain the covariant functor $\CPA_p:= \CanPois\circ \PhSp_p : \LocSrc_p\to \PoisAlg$,
which we call the {\em canonical} Poisson algebra functor.
We immediately observe the following 
\begin{propo}\label{propo:PApproperties}
The covariant functor $\CPA_p:\LocSrc_p\to \PoisAlg$ satisfies the causality property and
the time-slice axiom. Moreover, $\Aut(\CPA_p)$ contains a $\bbZ_2$ subgroup for $m\neq 0$ and a $\bbZ_2\times \bbR^p$ subgroup
 for $m=0$.
\end{propo}
\begin{proof}
By Proposition \ref{propo:phspproperties} c) the functor $\PhSp_p$ satisfies these properties. Thus $\CanPois$ automatically obeys
the time-slice property because functors preserve isomorphisms.  
The causality property
 is seen as follows: Given any two $\PreSymp$-morphisms $L_{i}: (V_i,\sigma_{V_i})\to (V,\sigma_V) $, $i=1,2$, 
 such that $\sigma_V\big(L_1[V_1],L_2[V_2]\big) =\{0\}$
 in $(V,\sigma_V)$, then the explicit expression for the Poisson bracket (\ref{eqn:poissonbracket})
 implies that $\{\cdot,\cdot\}_{\sigma_V}$ acts trivially between
 $\CanPois(L_1)\big[\CanPois(V_1,\sigma_1)\big]$
 and  $\CanPois(L_2)\big[\CanPois(V_2,\sigma_2)\big]$. The statement
on the automorphism group follows from Theorem \ref{theo:automorphismgroup}
and the  fact that every automorphism $\eta=\{\eta_{(\bfM,\bfJ)}\}$ of $\PhSp_p$ 
lifts to an automorphism $\{\CanPois(\eta_{(\bfM,\bfJ)})\}$ of $\CPA_p$.
 \end{proof}

Next, we shall show that the functor $\CPA_p$ violates the analog of the composition property for $\PhSp_p$ discussed in
Section \ref{sec:compproperty}, cf.~Proposition \ref{propo:compprop}.
 For $2\leq p\in\bbN$ and $0<q<p$, we define the covariant functor
\begin{flalign}\label{eqn:PApq}
\CPA_{p,q}:= \otimes \circ \big(\CPA_q\times \CPA_{p-q}\big)\circ \mathfrak{Split}_{p,q}: \LocSrc_p\to \PoisAlg~,
\end{flalign}
where $\otimes: \PoisAlg\times\PoisAlg\to\PoisAlg$ is the covariant functor that takes (algebraic) tensor products of
 Poisson algebras. Explicitly, to any object $(\mathcal{A},\mathcal{B})$ in $\PoisAlg\times\PoisAlg$ we associate the object
$\otimes(\mathcal{A},\mathcal{B}):= \mathcal{A}\otimes\mathcal{B}$ in $\PoisAlg$, which is the tensor product
of the underlying commutative and associative unital algebras, equipped with the Poisson bracket specified by linearity and,
for all $a,a^\prime\in\mathcal{A}$ and $b,b^\prime\in\mathcal{B}$,
\begin{flalign}
\{a\otimes b,a^\prime \otimes b^\prime\}_{\mathcal{A}\otimes\mathcal{B}} := \{a,a^\prime\}_{\mathcal{A}}\otimes (b\,b^\prime) 
+ (a\,a^\prime)\otimes \{b,b^\prime\}_{\mathcal{B}}~.
\end{flalign}
To any morphism $(\kappa,\lambda):(\mathcal{A}_1,\mathcal{B}_1)\to(\mathcal{A}_2,\mathcal{B}_2)$ in 
$\PoisAlg\times\PoisAlg$ we associate the morphism $\otimes(\kappa,\lambda) :=\kappa\otimes\lambda:
 \mathcal{A}_1\otimes\mathcal{B}_1\to \mathcal{A}_2\otimes\mathcal{B}_2$
in $\PoisAlg$ specified by linearity and, for all $a\in\mathcal{A}_1$ and $b\in\mathcal{B}_1$,
$(\kappa\otimes\lambda)\big(a\otimes b\big) = \kappa(a)\otimes \lambda(b)$.
To show that the covariant functors $\CPA_{p,q}$ and $\CPA_p$ are inequivalent,
we require two  lemmas.
\begin{lem}\label{lem:presymppoisalgiso}
Let $(V,\sigma_V)$ and $(W,\sigma_W)$ be two objects in $\PreSymp$. 
Then there exists an isomorphism $L: (V,\sigma_V)\to(W,\sigma_W)$ in $\PreSymp$ if and only
if there exists an isomorphism $\kappa: \CanPois(V,\sigma_V)\to \CanPois(W,\sigma_W)$
in $\PoisAlg$. Moreover, the isomorphisms $\kappa$ and $L$ are related by $L = \pi_{S^1(W)}\circ \kappa \vert_{S^1(V)}$,
where $\kappa\vert_{S^1(V)}$ is the restriction of $\kappa$ to the vector subspace $S^1(V)\subseteq \CanPois(V,\sigma_V)$
and $\pi_{S^1(W)}:  \CanPois(W,\sigma_W)\to S^1(W)$ is the projection to the degree one vector subspace $S^1(W)$. 
$L$ is uniquely determined by $\kappa$, but $L$ does not determine $\kappa$.
\end{lem}
\begin{proof}
The direction ``$\Rightarrow$'' is a consequence of functoriality, 
because $\CanPois$ preserves isomorphisms. To show the reverse direction, suppose that
$\kappa:\CanPois(V,\sigma_V)\to\CanPois(W,\sigma_W)$ is a $\PoisAlg$-isomorphism.
In particular, $\kappa$ is a unital algebra isomorphism $\kappa:S(V)\to S(W)$
between the symmetric tensor algebras of $V$ and $W$. This algebra isomorphism is uniquely specified
by its action on arbitrary $v\in S^1(V)=V$, so let us write $\kappa(v) = \kappa_0(v) + \kappa_1(v) + \kappa_{\geq 2}(v)$, 
where $\kappa_0:V\to \bbR$, $\kappa_1:V\to W$ and $\kappa_{\geq 2}:V\to S^{\geq 2}(W)$
 are the projections of $\kappa\vert_{S^1(V)}$ to the subspaces $S^0(W)$, $S^1(W)$ and $S^{\geq 2}(W):=
\bigoplus_{k=2}^\infty S^k(W)$.
\sk

We now will show that, given any $\PoisAlg$-isomorphism $\kappa$, 
there exists a $\PoisAlg$-isomorphism $\widetilde{\kappa}:\CanPois(V,\sigma_V)
\to \CanPois(W,\sigma_W)$, with $\widetilde{\kappa}_0 =0$ and $\widetilde{\kappa}_1=\kappa_1$.
Consider the $\PoisAlg$-automorphism
$\chi: \CanPois(V,\sigma_V)\to \CanPois(V,\sigma_V)$ defined by, for all $v\in V$,
$\chi(v) = v-\kappa_0(v)$. Define $\widetilde{\kappa} := \kappa\circ \chi : \CanPois(V,\sigma_V)\to 
\CanPois(W,\sigma_W)$, which is as a composition of $\PoisAlg$-isomorphisms again a $\PoisAlg$-isomorphism
and notice that $\widetilde{\kappa}(v) = \kappa_1(v) + \kappa_{\geq 2}(v)$, for any $v\in V$. As 
$\widetilde{\kappa}$ is an algebra homomorphism, it is therefore lower-triangular with 
respect to the gradings of $\CanPois(V,\sigma_V)$ and
$\CanPois(W,\sigma_W)$: the degree-$k$ component of any $\widetilde{\kappa}(a)$
depends only on the components of $a$ with degree $k$ or less. Accordingly, 
$\widetilde{\kappa}^{-1}$ is also lower-triangular, and all diagonal
entries $\pi_{S^k(W)}\circ\widetilde{\kappa}|_{S^k(V)}$ ($k\in\mathbb{N}^0$) are vector space isomorphisms. In particular,
 $\kappa_1:V\to W$ is a vector space isomorphism. The claim that $\kappa_1: (V,\sigma_V)\to (W,\sigma_W)$
is a $\PreSymp$-isomorphism follows by evaluating both sides of the condition, for all $v,v^\prime\in V$,
 $\widetilde{\kappa}\big(\{v,v^\prime\}_{\sigma_V}\big) = \{\widetilde{\kappa}(v),\widetilde{\kappa}(v^\prime)\}_{\sigma_W}$.
\end{proof}
We next show that the covariant functor $\CPA_{p,q}$ defined in (\ref{eqn:PApq}) 
is naturally isomorphic to the covariant functor $\CanPois\circ \PhSp_{p,q}$, where
$\PhSp_{p,q}$ is defined in (\ref{eqn:PhSppq}). This follows from the more general
\begin{lem}\label{lem:naturalisoforpoisson}
The covariant functors $\CanPois\circ \oplus : \PreSymp\times\PreSymp \to \PoisAlg$
and $\otimes\circ (\CanPois\times \CanPois): \PreSymp\times\PreSymp \to \PoisAlg$ are naturally isomorphic. 
Specifically, the $\PoisAlg$-morphisms 
\begin{flalign}
\eta_{((V,\sigma_V),(W,\sigma_W))} : \CanPois(V\oplus W,\sigma_{V\oplus W}) \to
\CanPois(V,\sigma_V)\otimes \CanPois(W,\sigma_W)
\end{flalign}
specified by, for all $(v,w)\in V\oplus W$, $\eta_{((V,\sigma_V),(W,\sigma_W))}(v,w) = v\otimes 1 + 1\otimes w$
define a natural isomorphism $\{\eta_{((V,\sigma_V),(W,\sigma_W))}\} : \CanPois\circ \oplus
\Rightarrow \otimes\circ (\CanPois\times \CanPois)$.
\end{lem}
\begin{proof}
$\eta:=\eta_{((V,\sigma_V),(W,\sigma_W))}$ is clearly a linear map (of $\bbN^0$-degree zero) 
and therefore induces a unital algebra homomorphism
$S(V\oplus W) \to S(V)\otimes S(W)$, which preserves the natural $\bbN^0$-gradings (with a slight abuse of notation
we use for the graded tensor product the same symbol $\otimes$).
A simple computation shows that $\eta$ preserves the Poisson bracket of elements in 
$S^1(V\oplus W)$ and is therefore a Poisson morphism
by \eqref{eqn:poissonbracket}. The inverse to $\eta$
is given by $\eta^{-1}(v\otimes 1) = (v,0)$ and $\eta^{-1}(1\otimes w) = (0,w)$ on the generators $v\otimes 1$ and $1\otimes w$, with 
$v\in V$ and $w\in W$, of $S(V)\otimes S(W)$, and extended as
a unital algebra homomorphism. 
\sk

It remains to show naturality. Let $(L,K): ((V_1,\sigma_{V_1}),(W_1,\sigma_{W_1}))\to  ((V_2,\sigma_{V_2}),(W_2,\sigma_{W_2}))$
be a morphism in $\PreSymp\times \PreSymp$, and denote $\eta_1:=\eta_{((V_1,\sigma_{V_1}),(W_1,\sigma_{W_1}))}$
and $\eta_2:=\eta_{((V_2,\sigma_{V_2}),(W_2,\sigma_{W_2}))}$.
Then, for all $(v,w)\in V_1\oplus W_1$,
\begin{flalign}
\nn \eta_2\big(L\oplus K(v,w)\big) &= \eta_{2}(L(v),K(w)) = L(v)\otimes 1 + 1\otimes K(w) \\
&= (L\otimes K)(v\otimes 1) + (L\otimes K)(1\otimes w) = (L\otimes K)(\eta_1(v,w)),
\end{flalign}
which proves naturality since $(v,w)\in V_1\oplus W_1$ are the  generators of $S(V_1\oplus W_1)$.
\end{proof}
We now can prove the violation of the composition property.
\begin{propo}\label{propo:compviolationPA}
For any $2\leq p\in\bbN$ and $0<q<p$, the covariant functors $\CPA_{p,q}:\LocSrc_p\to \PoisAlg$ 
and $\CPA_p:\LocSrc_p\to \PoisAlg$ are not naturally isomorphic. 
Indeed, there is no object
$(\bfM,\bfJ)$ in $\LocSrc_p$ for which the Poisson algebras $\CPA_{p,q}(\bfM,\bfJ)$
and $\CPA_p(\bfM,\bfJ)$ are isomorphic.
\end{propo}
\begin{proof}
We argue by contradiction: suppose $\CPA_{p,q} (\bfM,\bfJ)$ and 
 $\CPA_{p}(\bfM,\bfJ)$ are $\PoisAlg$-isomorphic for some object $(\bfM,\bfJ)$ in $\LocSrc_p$.  Then $\CanPois(\PhSp_p(\bfM,\bfJ))$ 
and $\CanPois\big(\PhSp_q(\bfM,\bfJ^q)\oplus \PhSp_{p-q}(\bfM,\bfJ^{p-q})\big)$ are also $\PoisAlg$-isomorphic by 
Lemma \ref{lem:naturalisoforpoisson}, and hence by Lemma \ref{lem:presymppoisalgiso} we have that
$\PhSp_p(\bfM,\bfJ)$ and $\PhSp_q(\bfM,\bfJ^q)\oplus \PhSp_{p-q}(\bfM,\bfJ^{p-q})$ are
 $\PreSymp$-isomorphic. But this is excluded
by Proposition \ref{propo:compprop}.
\end{proof}

\subsection{\label{subsec:improvedpoisson}Improved Poisson algebras}
In this subsection we will modify the canonical Poisson algebras constructed in
Subsection \ref{subsec:canpois} in order to address the problems concerning the unexpectedly large automorphism group
and the violation of the composition property. As already mentioned, the key point is to represent the algebras 
given by the functor $\CPA_p$ as functionals on the affine space of solutions to the inhomogeneous Klein--Gordon equation. 
When this is done, a degeneracy becomes apparent which was not visible in the description available in the category 
of presymplectic vector spaces.
However, we will show that the degeneracy may be described and also   removed in the category of 
Poisson algebras, thereby resolving the problems discussed above 
(cf.~Subsections \ref{subsec:poissonautomorphism} and \ref{subsec:poissoncomposition}, 
and Appendix~\ref{sec:pointed} for another approach).
\sk

The abstract Poisson algebras $\CPA_p$ are represented as functionals on the solution spaces
 $\Sol_p: \LocSrc_p \to \Aff$ by extending the pairing
\eqref{eqn:pairing} of $\PhSp_p$ and $\Sol_p$ as follows:
For any object $(\bfM,\bfJ)$ in $\LocSrc_p$
we extend $\sip{\,\cdot\,}{\,\cdot\,}_{(\bfM,\bfJ)}$ to a map $\CPA_p(\bfM,\bfJ)\times \Sol_p(\bfM,\bfJ)\to \bbR$
(denoted with a slight abuse of notation by the same symbol) in such a way that it is a unital algebra homomorphism in the left entry.
Explicitly, we set, for all $\phi\in \Sol_p(\bfM,\bfJ)$ and $\alpha\in \bbR$,
\begin{subequations}
\begin{flalign}\label{eqn:zeroorderpairing}
\sip{\alpha}{\phi}_{(\bfM,\bfJ)} = \alpha~,
\end{flalign}
and, for all $\phi\in \Sol_p(\bfM,\bfJ)$ and $[(\varphi_1,\alpha_1)],\dots ,[(\varphi_k,\alpha_k)] \in \PhSp_p(\bfM,\bfJ)$,
\begin{flalign}\label{eqn:solpairingismultiplicative}
\sip{[(\varphi_1,\alpha_1)]\cdots [(\varphi_k,\alpha_k)]}{\phi}_{(\bfM,\bfJ)} 
= \sip{[(\varphi_1,\alpha_1)]}{\phi}_{(\bfM,\bfJ)}\cdots \sip{[(\varphi_k,\alpha_k)]}{\phi}_{(\bfM,\bfJ)}~.
\end{flalign}
\end{subequations} 
Note that  
the pairing induces a pairing $\sip{\,\cdot\,}{\,\cdot\,}^\mathrm{lin}_{\bfM}:
 \PhSp_p^\mathrm{lin}(\bfM)\times \Sol_p^\mathrm{lin}(\bfM)\to\bbR$ between the linearized (pre)symplectic vector space
  and solution space, which describe a multiplet of $p\in\bbN$ homogeneous Klein--Gordon fields.
Explicitly, we have $\PhSp_p^\mathrm{lin}(\bfM) := \big(C^\infty_0(M,\bbR^p)/\mathrm{KG}_{\bfM}[C^\infty_0(M,\bbR^p)],
\sigma_{\bfM}^\mathrm{lin}\big)$, where, for all $[\varphi]^\mathrm{lin},[\psi]^\mathrm{lin}\in 
C^\infty_0(M,\bbR^p)/\mathrm{KG}_{\bfM}[C^\infty_0(M,\bbR^p)]$, the symplectic structure is given by
$\sigma_{\bfM}^\mathrm{lin}([\varphi]^\mathrm{lin},[\psi]^\mathrm{lin})
 := \int_M \ip{\varphi}{\mathrm{E}_{\bfM}(\psi)} \, \vol_{\bfM}$. The linearized pairing reads as,
 for all $[\varphi]^\mathrm{lin}\in \PhSp_p^\mathrm{lin}(\bfM)$ and $\underline{\phi}\in\Sol_p^\mathrm{lin}(\bfM)$,
 \begin{flalign}
 \sip{[\varphi]^\mathrm{lin}}{\underline{\phi}}_{\bfM}^\mathrm{lin} = \int_M\ip{\varphi}{\underline{\phi}}\, \vol_{\bfM}
 \end{flalign}
 and is related to $\sip{\,\cdot\,}{\,\cdot\,}_{(\bfM,\bfJ)}$ via, for all $[(\varphi,\alpha)]\in \PhSp_p(\bfM,\bfJ)$,
 $\phi\in\Sol_p(\bfM,\bfJ)$ and $\underline{\phi}\in\Sol_p^\mathrm{lin}(\bfM)$,
 \begin{flalign}\label{eqn:pairingrelation}
 \sip{[(\varphi,\alpha)]}{\phi +\underline{\phi}}_{(\bfM,\bfJ)} = \sip{[(\varphi,\alpha)]}{\phi }_{(\bfM,\bfJ)} 
 +  \sip{[\varphi]^\mathrm{lin}}{\underline{\phi}}_{\bfM}^\mathrm{lin}~.
 \end{flalign}
 Notice that for the linearized setting the analog of the diagram in (\ref{eqn:pairingdiagramphsp}) holds true. Moreover, we can extend $\sip{\,\cdot\,}{\,\cdot\,}^\mathrm{lin}_{\bfM}$ to a map 
$\CPA_p^\mathrm{lin}(\bfM)\times \Sol^\mathrm{lin}_p(\bfM)\to \bbR$, where 
 $\CPA_p^\mathrm{lin} := \CanPois\circ \PhSp_p^\mathrm{lin}$. 
These extended pairings are also natural, i.e.~the analog of the diagram in (\ref{eqn:pairingdiagramphsp}) holds true.

\begin{rem}
The pairing $\sip{\,\cdot\,}{\,\cdot\,}_{(\bfM,\bfJ)}$ provides us with a representation of the canonical (abstract)
Poisson algebra $\CPA_p(\bfM,\bfJ)$ as a polynomial algebra of functionals on the 
affine space $\Sol_p(\bfM,\bfJ)$. Analogously, the pairing $\sip{\,\cdot\,}{\,\cdot\,}^\mathrm{lin}_{\bfM}$ 
leads to a representation of $\CPA_p^\mathrm{lin}(\bfM)$ as a polynomial algebra 
of functionals on the vector space $\Sol_p^\mathrm{lin}(\bfM)$.
The Poisson bracket (\ref{eqn:poissonbracket}) can be expressed
in this representation as follows, for all $a,b\in \CPA_p(\bfM,\bfJ)$ and $\phi\in \Sol_p(\bfM,\bfJ)$,
\begin{flalign}\label{eqn:poissonfuncder}
\sip{\big\{a,b\big\}_{\sigma_{(\bfM,\bfJ)}}}{\phi}_{(\bfM,\bfJ)} = \int_M\ip{\sip{a^{(1)}}{\phi}_{(\bfM,\bfJ)}}{\mathrm{E}_{\bfM}\left(\sip{b^{(1)}}{\phi}_{(\bfM,\bfJ)}\right)}\, \vol_{\bfM}~,
\end{flalign}
where $a^{(1)}$ and $b^{(1)}$ are the first functional derivatives of $a$ and $b$, respectively, defined uniquely so that  
\begin{flalign}
\int_M\ip{\sip{a^{(1)}}{\phi}_{(\bfM,\bfJ)}}{\underline{\phi}} \,\vol_{\bfM}:= \frac{d}{d\epsilon} 
\sip{a}{\phi+\epsilon\,\underline{\phi}}_{(\bfM,\bfJ)}\vert^{\,}_{\epsilon=0}~,
\end{flalign}
for all $a\in \CPA_p(\bfM,\bfJ)$, $\phi\in \Sol_p(\bfM,\bfJ)$
and $\underline{\phi} \in\Sol_{p}^\mathrm{lin}(\bfM)$.
\end{rem}
\sk

We notice that the pairing $\sip{\,\cdot\,}{\,\cdot\,}_{(\bfM,\bfJ)}$ is non-degenerate when
acting on $\PhSp_p(\bfM,\bfJ)$. This means that $\sip{[(\varphi,\alpha)]}{\phi}_{(\bfM,\bfJ)} =0$, for 
all $\phi\in \Sol_p(\bfM,\bfJ)$, implies that $[(\varphi,\alpha)]=0$, and, vice versa,
that $\sip{[(\varphi,\alpha)]}{\phi}_{(\bfM,\bfJ)}=\sip{[(\varphi,\alpha)]}{\phi^\prime}_{(\bfM,\bfJ)}$,
for all $[(\varphi,\alpha)]\in \PhSp_p(\bfM,\bfJ)$, implies that $\phi = \phi^\prime$. However,
 the extended pairing on $\CPA_p(\bfM,\bfJ)\times\Sol_p(\bfM,\bfJ)$ turns out to be degenerate in the left entry
and non-degenerate in the right entry. For example,
taking $[(0,\alpha)]\in \PhSp_p(\bfM,\bfJ)$ with $\alpha\in \bbR$ we obtain 
\begin{flalign}\label{eqn:tmpvanishingofsimpleelements}
\sip{[(0,\alpha)] -\alpha }{\phi}_{(\bfM,\bfJ)} = \left(\int_M\ip{0}{\phi} \, \vol_{\bfM}\right) + \alpha -\alpha =0~,
\end{flalign}
for all $\phi\in\Sol_p(\bfM,\bfJ)$. Hence, the extension of $\sip{\,\cdot\,}{\,\cdot\,}$ from $\PhSp_p$
to $\CPA_p$ has introduced a new degeneracy, which can not be seen at the level
of presymplectic vector spaces as it mixes different $\bbN^0$-degrees in $\CPA_p$. This degeneracy 
is removed precisely by  taking the quotient via a suitable Poisson ideal, namely 
the vanishing ideal 
\begin{flalign}\label{eqn:poisideal}
\mathfrak{I}_p(\bfM,\bfJ) := \Big\{a\in \CPA_p(\bfM,\bfJ): \sip{a}{\phi}_{(\bfM,\bfJ)} 
=0~,~~\text{for all }\phi\in\Sol_p(\bfM,\bfJ)\Big\}
\end{flalign}
of the pairing $\sip{\,\cdot\,}{\,\cdot\,}$ (we will check that it is
indeed a Poisson ideal below). The
corresponding theory will turn out to have the correct automorphism group and composition property.
At this point we would like to note that the pairing $\sip{\,\cdot\,}{\,\cdot\,}_{\bfM}^\mathrm{lin}$
is non-degenerate when acting on both $\PhSp_p^\mathrm{lin}(\bfM)$ and $\CPA_p^\mathrm{lin}(\bfM)$.

\begin{lem}\label{lem:vanishingidealproperties}
For any object $(\bfM,\bfJ)$ in $\LocSrc_p$ the vanishing ideal $\mathfrak{I}_p(\bfM,\bfJ)$ is
a proper Poisson ideal of $\CPA_p(\bfM,\bfJ)$. Hence the quotient
$\CPA_p(\bfM,\bfJ)/\mathfrak{I}_p(\bfM,\bfJ)$ is a nontrivial unital Poisson algebra. 
\end{lem}
\begin{proof}
Given any element $a\in\mathfrak{I}_p(\bfM,\bfJ)$ it is easy to see that all 
its functional derivatives vanish, in particular $\sip{a^{(1)}}{\phi}_{(\bfM,\bfJ)}=0$ for all $\phi\in\Sol_p(\bfM,\bfJ)$.
Thus $\{a,b\}_{\sigma_{(\bfM,\bfJ)}} \in\mathfrak{I}_p(\bfM,\bfJ)$ for any $b\in \CPA_p(\bfM,\bfJ)$ by \eqref{eqn:poissonfuncder}. Since
$\mathfrak{I}_p(\bfM,\bfJ)$ is certainly an ideal, it is a Poisson ideal, and
 \eqref{eqn:zeroorderpairing} shows that it is proper. 
\end{proof}

The quotient $\CPA_p(\bfM,\bfJ)/\mathfrak{I}_p(\bfM,\bfJ)$ gives
our improved Poisson algebra for the multiplet of $p\in\bbN$ inhomogeneous Klein--Gordon fields.
 It is of course free of the redundancy discussed above. However, it is sometimes cumbersome to 
 do explicit calculations involving $\mathfrak{I}_p(\bfM,\bfJ)$.
In order to simplify the following constructions, we shall provide an equivalent characterization
of $\mathfrak{I}_p(\bfM,\bfJ)$ in terms of an algebraically generated ideal. Let us define
the following ideal (generated by a set) of $\CPA_p(\bfM,\bfJ)$ 
\begin{flalign}\label{eqn:algebraicPoissonideal}
\widetilde{\mathfrak{I}}_p(\bfM,\bfJ) := \Big\langle\big\{[(0,\alpha)] - \alpha \in \CPA_p(\bfM,\bfJ) : \alpha\in\bbR \big\} \Big\rangle~.
\end{flalign}
Since $\big\{[(0,\alpha)] -\alpha , a\big\}_{\sigma_{(\bfM,\bfJ)}}=0$, for all $\alpha\in\bbR$ and $a\in \CPA_p(\bfM,\bfJ)$,
the ideal $\widetilde{\mathfrak{I}}_p(\bfM,\bfJ) $ is a Poisson ideal of $\CPA_p(\bfM,\bfJ)$.
Furthermore, (\ref{eqn:tmpvanishingofsimpleelements})  and (\ref{eqn:solpairingismultiplicative}) implies that
\begin{flalign}\label{eqn:embedded_ideals}
\widetilde{\mathfrak{I}}_p(\bfM,\bfJ) \subseteq \mathfrak{I}_p(\bfM,\bfJ)~.
\end{flalign}
We now prove that $\widetilde{\mathfrak{I}}_p(\bfM,\bfJ) =\mathfrak{I}_p(\bfM,\bfJ)$, which will allow
 us to work with the easy-to-use algebraically generated ideal  $\widetilde{\mathfrak{I}}_p(\bfM,\bfJ)$
 whenever it is suitable.
\begin{lem}\label{lem:Poissonidealsarethesame}
Let $(\bfM,\bfJ)$  be any object in $\LocSrc_p$. 
\begin{itemize}
\item[a)] The Poisson algebra $\CPA_p(\bfM,\bfJ)/\widetilde{\mathfrak{I}}_p(\bfM,\bfJ)$ is 
(noncanonically) isomorphic to $\CPA_p^{\mathrm{lin}}(\bfM)$.
\item[b)] $\widetilde{\mathfrak{I}}_p(\bfM,\bfJ) = \mathfrak{I}_p(\bfM,\bfJ)$.
\end{itemize} 
\end{lem}
\begin{proof}
Proof of a):  Let us define a
Poisson algebra homomorphism $\kappa : \CPA_p(\bfM,\bfJ) \to \CPA_p^\mathrm{lin}(\bfM)$
by setting, for all $[(\varphi,\alpha)]\in \CPA_p(\bfM,\bfJ) $,
\begin{flalign}\label{eqn:noncanonicalkappa}
\kappa\big([(\varphi,\alpha)]\big) = \sip{[(\varphi,\alpha)]}{\phi_0}_{(\bfM,\bfJ)} + [\varphi]^\mathrm{lin}~,
\end{flalign}
where $\phi_0\in\Sol_p(\bfM,\bfJ)$ is any fixed solution. As $\widetilde{\mathfrak{I}}_p(\bfM,\bfJ)$ clearly lies in the kernel
of $\kappa$, we can induce a Poisson algebra homomorphism (denoted by the same symbol) 
$\kappa : \CPA_p(\bfM,\bfJ)/\widetilde{\mathfrak{I}}_p(\bfM,\bfJ) \to \CPA_p^\mathrm{lin}(\bfM)$.
To show that the induced $\kappa$ is a $\PoisAlg$-isomorphism, we notice that
setting, for all $[\varphi]^\mathrm{lin}\in \CPA_p^\mathrm{lin}(\bfM)$,
\begin{flalign}
\kappa^{-1}\big([\varphi]^\mathrm{lin}\big) := \big[ [(\varphi,0)] -\sip{[(\varphi,0)]}{\phi_0}_{(\bfM,\bfJ)}  \big]\in \CPA_p(\bfM,\bfJ)/\widetilde{\mathfrak{I}}_p(\bfM,\bfJ) 
\end{flalign} 
is well-defined and defines the inverse of $\kappa$.
\sk

Proof of b): By a), $\CPA_p(\bfM,\bfJ)/\widetilde{\mathfrak{I}}_p(\bfM,\bfJ)$ is a simple (and nontrivial) 
Poisson algebra. Hence $\widetilde{\mathfrak{I}}_p(\bfM,\bfJ)$ is a maximal proper ideal.
In view of \eqref{eqn:embedded_ideals}, this shows that $\widetilde{\mathfrak{I}}_p(\bfM,\bfJ) = \mathfrak{I}_p(\bfM,\bfJ)$.
\end{proof}

These results now allow us to construct our improved functor for the classical theory
of a multiplet of $p\in\bbN$ inhomogeneous Klein--Gordon fields.
\begin{propo}\label{propo:quotientfunctor}
The following defines a covariant functor $\PA_p : \LocSrc_p\to \PoisAlg $:
To any object $(\bfM,\bfJ)$ in $\LocSrc_p$ we associate the
Poisson algebra
$\PA_p (\bfM,\bfJ) := \CPA_p(\bfM,\bfJ)/\mathfrak{I}_p(\bfM,\bfJ)$. 
To any morphism $f:(\bfM_1,\bfJ_1)\to (\bfM_2,\bfJ_2)$ in $\LocSrc_p$ we associate the map
$\PA_p(f): \PA_p(\bfM_1,\bfJ_1) \to \PA_p(\bfM_2,\bfJ_2)$
that is canonically induced from $\CPA_p(f): \CPA_p(\bfM_1,\bfJ_1) \to \CPA_p(\bfM_2,\bfJ_2)$.
\end{propo}
\begin{proof}
Lemma \ref{lem:vanishingidealproperties} has established that
the quotients are nontrivial unital Poisson algebras. Next, let $f:(\bfM_1,\bfJ_1)\to (\bfM_2,\bfJ_2)$ be
 any morphism in $\LocSrc_p$. Then $\CPA_p(f)$  induces a Poisson algebra homomorphism 
 $\PA_p(f):\PA_p (\bfM_1,\bfJ_1)  \to \PA_p (\bfM_2,\bfJ_2)$
because it restricts to a map 
 $\CPA_p(f): \mathfrak{I}_p(\bfM_1,\bfJ_1) \to \mathfrak{I}_p(\bfM_2,\bfJ_2)$, as is obvious from
  Lemma \ref{lem:Poissonidealsarethesame} b), the explicit characterization of the
 algebraic Poisson ideal (\ref{eqn:algebraicPoissonideal}) and the fact that
 $\CPA_p(f)([(0,\alpha)]-\alpha) = [(0,\alpha)]-\alpha$, for any $\alpha\in\bbR$. It is clear that $\PA_p(f)$ is unit-preserving, because
$\CPA_p(f)$ is; moreover, it is injective, since $\PA_p(\bfM_1,\bfJ_1)$ is simple by~Lemma \ref{lem:Poissonidealsarethesame} a),
 and hence it does not have any nontrivial proper Poisson ideals (and
$\PA_p(f)$ is a unit-preserving map to a nontrivial unital
Poisson algebra, so it is not the zero map). 
The composition and identity properties of $\PA_p$ are inherited from $\CPA_p$, hence
 $\PA_p:\LocSrc_p\to \PoisAlg$ is a covariant functor.
\end{proof}

\begin{propo}\label{propo:improvedpoisalgfunctorisLCFT}
The covariant functor $\PA_p: \LocSrc_p\to\PoisAlg$ satisfies the causality property and the time-slice axiom, Hence,
it is a locally covariant classical field theory.
\end{propo}
\begin{proof}
By Proposition \ref{propo:PApproperties} the covariant functor $\CPA_p$ satisfies these properties.
The quotients by Poisson ideals used in the definition of the functor $\PA_p$
preserve these properties due to the following arguments: For the time-slice axiom we
just have to notice that any $\PoisAlg$-isomorphism which preserves the Poisson ideals 
induces a $\PoisAlg$-isomorphism on the quotients (simply induce the morphism and its inverse to the quotients). Causality holds
because if two subalgebras of a Poisson algebra $\mathcal{A}$
Poisson-commute, then the same is true of the corresponding
subalgebras of any quotient of $\mathcal{A}$ by a Poisson ideal.
\end{proof}

\begin{rem}\label{rem:classical_fields}
The theory $\PA_p$ admits a locally covariant field, in the following sense. 
 For each $(\bfM,\bfJ)$, let $\Phi_{(\bfM,\bfJ)}:C_0^\infty(M, \bbR^p)
\to \PA_p(\bfM,\bfJ)$ be the linear map
\begin{flalign}
\Phi_{(\bfM,\bfJ)}(\varphi) =\big[ [(\varphi,0)]\big]  \,,
\end{flalign}
where the outer brackets denote the equivalence relation used in defining
$\PA_p$ and the inner square brackets that used in defining $\PhSp_p$. 
The significance of this field is seen via its action on solutions,  
\begin{flalign}
\sip{\Phi_{(\bfM,\bfJ)}(\varphi)}{\phi}_{(\bfM,\bfJ)} = 
 \int_M\ip{\varphi}{\phi}\,\vol_{\bfM} ~,
\end{flalign}
where we abuse notation slightly, by extending the notation for
the pairing of $\CPA_p$ with $\Sol_p$ to encompass the pairing
of $\PA_p$ with $\Sol_p$ obtained by quotienting. 
Furthermore, if $f:(\bfM_1,\bfJ_1)\to (\bfM_2,\bfJ_2)$ is a morphism in $\LocSrc_{p}$
then the fields are related by 
\begin{flalign}
\QA_p(f)\big(\Phi_{(\bfM_1,\bfJ_1)}(\varphi)\big)
=\Phi_{(\bfM_2,\bfJ_2)}\big(f_\ast(\varphi)\big) ~,
\end{flalign}
for all $\varphi\in C_0^\infty(M_1,\bbR^p)$. Formalizing
the test function spaces via 
a functor $\mathfrak{D}_p:\LocSrc_{p}\to \Vec$ defined by
$\mathfrak{D}_p(\bfM,\bfJ)=
C_0^\infty(M,\bbR^p)$, $\mathfrak{D}_p(f) =f_\ast$, we
see that the maps $\Phi_{(\bfM,\bfJ)}$ provide the components
of a natural transformation $\Phi:\mathfrak{D}_p\Rightarrow\PA_p$ (we suppress a forgetful functor from $\PoisAlg$ to $\Vec$). 
Thus $\Phi$ conforms to the understanding of fields
as natural transformations, set out in \cite{Brunetti:2001dx} for the case of quantum fields.
Moreover,  $\Phi_{(\bfM,\bfJ)}$ obeys  
the field equation in the form
\begin{flalign}\label{eq:cfeq}  
\nn \Phi_{(\bfM,\bfJ)}\big(\mathrm{KG}_{\bfM}(\varphi) \big)
+\int_M \ip{\varphi}{\bfJ}\, \vol_{\bfM}&= \Big[ \big[\big(\mathrm{KG}_{\bfM}(\varphi) ,0\big)\big]\Big]  + \Big[ \big[\big(0 ,\int_M \ip{\varphi}{\bfJ}\, \vol_{\bfM}\big)\big]\Big] \\
&=\Big[ \big[\big(\mathrm{KG}_{\bfM}(\varphi) ,\int_M\ip{\varphi}{\bfJ}\, \vol_{\bfM}\big)\big]\Big] =0~,
\end{flalign}
for all $\varphi\in C_0^\infty(M,\bbR^p)$, 
where we have used \eqref{eqn:algebraicPoissonideal} in the first equality and \eqref{eqn:equivalence} in the last one.
\end{rem}

\begin{rem}\label{rem:splitting}
The isomorphism of Lemma \ref{lem:Poissonidealsarethesame} a) merits further consideration. 
For $\lambda\in\bbR$, define a covariant functor $\Z_\lambda:\LocSrc_p\to\LocSrc_p$ so that
$\Z_\lambda(\bfM,\bfJ)=(\bfM,\lambda\bfJ)$ and so that, if $f:(\bfM_1,\bfJ_1)\to (\bfM_2,\bfJ_2)$, then
 $\Z_\lambda(f):(\bfM_1,\lambda\bfJ_1)\to (\bfM_2,\lambda\bfJ_2)$ has the same underlying map as $f$. 
 (Notice that $\Z_{\lambda}$ is the identity functor for $\lambda=1$.)
 Then the covariant functor $\PA_p\circ\Z_\lambda : \LocSrc_p\to \PoisAlg$ 
 describes the theory of $p$ inhomogeneous Klein--Gordon fields with a coupling
constant $\lambda$ to the external source. The corresponding equation of motion  on an object $(\bfM,\bfJ)$ in $\LocSrc_p$ is
$\mathrm{KG}_{\bfM}(\phi) + \lambda\, \bfJ=0$,
which reduces to the homogeneous Klein--Gordon equation for $\lambda=0$.
Lemma \ref{lem:Poissonidealsarethesame} a) asserts that, fixing any object $(\bfM,\bfJ)$ in $\LocSrc_p$, all Poisson 
algebras $\PA_p\circ\Z_\lambda(\bfM,\bfJ) = \PA_p(\bfM,\lambda\,\bfJ)$ ($\lambda\in \bbR$)
 are isomorphic to $\CPA_p^{\mathrm{lin}}(\bfM)$, albeit in a noncanonical fashion. 
The physical interpretation of 
this result is that the theories $\PA_p\circ\Z_\lambda$ ($\lambda\in\bbR$) have the same observables, 
which moreover satisfy identical algebraic relations, for any fixed object $(\bfM,\bfJ)$ in $\LocSrc_p$. 
In particular, this means that there exists a (by no means canonical) linear map
 $\Psi^{\mu,\lambda}_{(\bfM,\bfJ)}:C_0^\infty(\bfM;\bbR)\to\PA_p\circ\Z_\lambda(\bfM,\bfJ)$
  for each $\lambda,\mu\in\bbR$ such that
\begin{flalign}
\Psi^{\mu,\lambda}_{(\bfM,\bfJ)}\big(\mathrm{KG}_{\bfM}(\varphi) \big)
+\mu \int_M \ip{\varphi}{\bfJ}\, \vol_{\bfM} &= 0~,
\end{flalign}
for all $\varphi\in C_0^\infty(M,\bbR^p)$. In other words, restricting attention
to a single object $(\bfM,\bfJ)$, the theory for coupling constant $\lambda$ 
admits a field obeying the field equation for coupling constant $\mu$!

These observations {\em do not} assert that the theories with different coupling constants are equivalent.
Rather, they show that the distinction between the theories $\PA_p\circ\Z_\lambda$ and $\PA_p\circ\Z_\mu$ 
for $\lambda\neq\mu$ is {\em only} visible when one considers how the Poisson algebras
for two objects $(\bfM_1,\bfJ_1)$ and $(\bfM_2,\bfJ_2)$ are
embedded if there is a morphism $f:(\bfM_1,\bfJ_1)\to(\bfM_2,\bfJ_2)$. For instance,
the field assignments $\Psi^{\mu,\lambda}_{(\bfM,\bfJ)}$ just mentioned cannot be chosen
so as to constitute a locally covariant field $\Psi^{\mu,\lambda}:\mathfrak{D}_p\Rightarrow\PA_p\circ \Z_\lambda$;
if they could, the difference $\Psi^{\mu,\lambda}-\Phi^{\lambda}$ (with $\Phi^\lambda$ being the locally covariant field
for the theory $\PA_p\circ \Z_{\lambda}$, cf.\ Remark \ref{rem:classical_fields}) would be a multiple of the unit 
yielding a naturally assigned solution to the inhomogeneous field equation for coupling $\mu-\lambda$, which is impossible
unless $\lambda=\mu$. This is direct evidence for the necessity of our locally covariant perspective on inhomogeneous field theories.

At the root of this discussion is the fact, evident from 
the proof of Lemma \ref{lem:Poissonidealsarethesame} a), that
the identification between $\PA_p(\bfM,\bfJ)$ and $\CPA_p^{\mathrm{lin}}(\bfM)$ 
requires a (necessarily noncanonical) splitting of an inhomogeneous solution $\phi = \phi_0+\underline{\phi}$
into a fixed ``background solution'' $\phi_0$ and a homogeneous solution $\underline{\phi}$. A similar result holds true
after quantization, see Lemma \ref{lem:quantidealproperties} below, where the choice of a ``background solution'' 
$\phi_0$ is replaced by the choice of a ``background quantum state'' $\omega_0$ that satisfies appropriate 
inhomogeneous properties.
In textbook treatments of inhomogeneous classical and quantum field theories, see e.g.\ \cite[Chapter 4]{IZbook},
such a splitting is typically assumed from the beginning in order to reduce the construction
of inhomogeneous theories to homogeneous ones. As we have explained above, 
the splitting approach is highly noncanonical and can easily obscure the (physical) interpretation of the theory, 
which, however, is faithfully reflected in the functorial structure.
It is also worth mentioning that, in Appendix~\ref{sec:fedosov}, we discuss the Fedosov quantization method, 
which addresses the noncanonical nature of such splittings by treating {\em all} of them on an equal footing.

\end{rem}


\subsection{\label{subsec:poissonrce}Relative Cauchy evolution of the functor $\PA_p$}
The relative Cauchy evolution of the functor $\PhSp_p:\LocSrc_p\to \PreSymp$ induces 
that of the functor $\PA_p: \LocSrc_p\to\PoisAlg$ as follows:
Let $(\bfM,\bfJ)$ be any object in $\LocSrc_p$ and let $(\h,\j)\in H(\bfM,\bfJ)$ be any globally hyperbolic
perturbation. From the explicit expression for $\rce_{(\bfM,\bfJ)}^{(\PhSp_p)}[\h,\j] \in\Aut(\PhSp_p(\bfM,\bfJ))$
given in (\ref{eq:rce_reform}) we observe that the relative Cauchy evolution $\rce_{(\bfM,\bfJ)}^{(\PA_p)}[\h,\j]\in\Aut(
\PA_p(\bfM,\bfJ))$ of $\PA_p$ is uniquely specified by, for
all $[(\varphi,\alpha)]\in \PA_p(\bfM,\bfJ)$,
\begin{multline}\label{eqn:rcePA}
\rce_{(\bfM,\bfJ)}^{(\PA_p)}[\h,\j]\big([(\varphi,\alpha)]\big) = \left[\left(\varphi + (\mathrm{KG}_{\bfM} - \mathrm{KG}_{\bfM[\h]})\big(\mathrm{E}_{\bfM[\h]}(\varphi)\big),0\right)\right] + \alpha \\
+  \int_M \left(\ip{-\j}{\mathrm{E}_{\bfM[\h]}(\varphi)} + \ip{(1-\rho_{\h})\,(\bfJ+\j)}{\mathrm{E}_{\bfM[\h]}(\varphi)}\right)\,\vol_{\bfM}~,
\end{multline}
where on the right hand side we have used the equivalence relation entering the definition
of $\PA_p(\bfM,\bfJ)$ (cf.\ Proposition \ref{propo:quotientfunctor}) and we have chosen
as in (\ref{eq:rce_reform})  a representative $\varphi$ with compact support in $M^+$.
With the techniques presented in Appendix \ref{app:Cauchy} one can differentiate this expression to yield 
\begin{flalign}
\nonumber \frac{d}{ds} \rce_{(\bfM,\bfJ)}^{(\PA_p)}[s\h,s\j]\big([(\varphi,\alpha)]\big)\Big\vert_{s=0} &= - \left[\left( \mathrm{KG}_{\bfM[\h]}^\prime \big(\mathrm{E}_{\bfM}(\varphi)\big),0\right)\right] \\
\nonumber & \qquad ~\qquad - \int_M\ip{\frac{1}{2} g^{ab}\,h_{ab}\,\bfJ + \j}{\mathrm{E}_{\bfM}(\varphi)}\,\vol_{\bfM}\\
& = -\left\{ \frac{1}{2}  T_{(\bfM,\bfJ)}(\h)+ [(\j,0)]   ,   [(\varphi,\alpha)]\right\}_{\sigma_{(\bfM,\bfJ)}}~,\label{eq:PArcederiv}
\end{flalign}
where $T_{(\bfM,\bfJ)}(\h)$ is the functional on $\Sol_p(\bfM,\bfJ)$ given by
$\phi \mapsto \int_M h_{ab}\,T^{ab}_{(\bfM,\bfJ)}[\phi]\,\vol_{\bfM} $ and
 $T^{ab}_{(\bfM,\bfJ)}[\phi]$ is the stress-energy tensor \eqref{eqn:fullSET}.
Although \eqref{eqn:rcePA} was derived
 under an assumption on the support of the representative $\varphi$, 
 the formulae in \eqref{eq:PArcederiv} do not require to choose a suitable representative as they depend only on the equivalence
 class of $(\varphi,\alpha)$.

\subsection{\label{subsec:poissonautomorphism}Automorphism group of the functor $\PA_p$}
We study the automorphism group of the covariant functor $\PA_p:\LocSrc_p\to\PoisAlg$ defined
in Proposition \ref{propo:quotientfunctor}. We shall obtain that it is, as expected, the trivial group
for $m\neq 0$ and isomorphic to $\bbR^p$ for $m=0$.
\sk

We first show that for $m=0$ the automorphism group
of $\PA_p$ contains  $\bbR^p$ as a subgroup.
\begin{propo}\label{propo:R}
If $m=0$ there exists a faithful homomorphism $\eta: \bbR^p \to \Aut(\PA_p)$
induced by the one in Proposition \ref{propo:Z2xR} restricted to $\{+1\} \times \bbR^p\subseteq \bbZ_2\times \bbR^p$.
Explicitly, for any object $(\bfM,\bfJ)$ in $\LocSrc_p$ the automorphism
$\eta(\mu)_{(\bfM,\bfJ)}$ is specified by, for all $[(\varphi,\alpha)]\in \PA_p(\bfM,\bfJ)$,
\begin{flalign}\label{eqn:etamu}
\eta(\mu)_{(\bfM,\bfJ)}\big([(\varphi,\alpha)]\big) = \Big[\Big(\varphi, \alpha + \int_M \ip{\varphi}{\mu}\, \vol_{\bfM} \Big)\Big]~.
\end{flalign}
\end{propo}
\begin{proof}
Applying the functor $\CanPois$, the automorphism $\eta(\sigma,\mu)\in \Aut(\PhSp_p)$ of Proposition \ref{propo:Z2xR} 
induces an element in $\Aut(\CPA_p)$, which we denote with a slight abuse of notation by the same symbol
$\eta(\sigma,\mu)$. For $\sigma= -1$ this automorphism
does not preserve the Poisson ideals $\mathfrak{I}_p(\bfM,\bfJ)$:
Indeed, for $[(0,\alpha)]-\alpha \in\mathfrak{I}_p(\bfM,\bfJ)$ we find
$\eta(-1,\mu)_{(\bfM,\bfJ)}([(0,\alpha)] -\alpha) =[(0,-\alpha)] -\alpha\not\in \mathfrak{I}_p(\bfM,\bfJ)$.
For $\sigma=+1$ and $\mu\in\bbR^p$ arbitrary the Poisson ideals are preserved,
hence $\eta(+1,\mu)$ induces the automorphism $\eta(\mu)\in\Aut(\PA_p)$ given by \eqref{eqn:etamu}. 
 The group law $\eta(\mu)\circ \eta(\mu^\prime) = \eta(\mu+\mu^\prime)$
is a consequence of the group law for $\eta(\sigma,\mu)$, cf.~Proposition \ref{propo:Z2xR}.
\end{proof}
\begin{rem}
The above argument shows that the nontrivial $\bbZ_2$-automorphism
in the massless case (cf.~Proposition \ref{propo:Z2xR}) does not induce an automorphism
of $\PA_p$. The same holds true for the nontrivial $\bbZ_2$-automorphism in the massive case
(cf.~Proposition \ref{propo:Z2flip}).
\end{rem}

We now prove that the automorphisms found in Proposition \ref{propo:R}
exhaust the automorphism group of the covariant functor $\PA_p$. For the proof we
require the analog of Theorem \ref{theo:endobycomponent}, stating that an endomorphism is uniquely
determined by its component on one object, for the functor $\PA_p:\LocSrc_p\to\PoisAlg$.
This follows by a similar proof as for Theorem \ref{theo:endobycomponent} and we omit the details. 
\begin{theo}\label{theo:goodautgroupclassical}
Every endomorphism of the functor $\PA_p$ is an automorphism and
\begin{flalign}
\End(\PA_p) = \Aut(\PA_p) \simeq 
\begin{cases}
\{ \id_{\PA_p} \}& ~,~~\text{for } m\neq 0~,\\
\bbR^p & ~,~~\text{for } m=0~, 
\end{cases}
\end{flalign}
where the action for $m=0$ is given by Proposition \ref{propo:R}.
\end{theo}
\begin{proof}
Let $\eta\in\End(\PA_p)$ and let us consider its component $\eta_{(\bfM_0,0)}\in\End(\PA_p(\bfM_0,0))$ on the Minkowski spacetime
$\bfM_0$ with $\bfJ_0=0$.  Now, $\eta_{(\bfM_0,0)}$ must commute with the relative Cauchy 
evolution (\ref{eqn:rcePA}) and its derivative (\ref{eq:PArcederiv});
considering the $\h=0$ case of (\ref{eq:PArcederiv}), we obtain the
requirement
\begin{flalign} \label{eqn:etaM00}
\eta_{(\bfM_0,0)}\left(\Big\{[(\j,0)],
 [(\varphi,0)]\Big\}_{\sigma_{(\bfM_0,0)}}\right) = \Big\{[(\j,0)], \eta_{(\bfM_0,0)}\big([(\varphi,0)]\big)\Big\}_{\sigma_{(\bfM_0,0)}}
\end{flalign}
for all $\j,\varphi\in C^\infty_0(M_0,\bbR^p)$. Next, 
we exploit the fact (cf.\ Lemma~\ref{lem:Poissonidealsarethesame}) that there is a preferred isomorphism
 $\kappa:\PA_p(\bfM_0,0)\to\CPA^\mathrm{lin}_p(\bfM_0)$ given by $\kappa\big([(\varphi,\alpha)]\big) = 
 \alpha + [\varphi]^{\mathrm{lin}}$, i.e.\ the $\phi_0 =0$ case of \eqref{eqn:noncanonicalkappa}, which intertwines 
 the natural actions of the Poincar{\'e} transformations
on $\PA_p(\bfM_0,0)$ and $\CPA^\mathrm{lin}_p(\bfM_0)$. Then
the induced endomorphism $\widetilde{\eta}=\kappa\circ\eta_{(\bfM_0,0)}
\circ\kappa^{-1}$ of $\CPA^\mathrm{lin}_p(\bfM_0)$ must commute
with all Poincar\'e transformations, because $\eta_{(\bfM_0,0)}$ does.
Owing to \eqref{eqn:etaM00}, $\widetilde{\eta}$ also satisfies, for all 
generators $[\varphi]^\mathrm{lin}\in  \CPA_p^\mathrm{lin}(\bfM_0)$
and all $\j\in C^\infty_0(M_0,\bbR^p)$,
\begin{flalign} \label{eqn:etatilde}
\widetilde{\eta}\left(\Big\{[\j]^\mathrm{lin}, [\varphi]^\mathrm{lin}\Big\}_{\sigma^\mathrm{lin}_{\bfM_0}}\right) = \Big\{[\j]^\mathrm{lin}, \widetilde{\eta}\big([\varphi]^\mathrm{lin}\big)\Big\}_{\sigma^\mathrm{lin}_{\bfM_0}}~.
\end{flalign}
The left hand side of this equation is simply $\sigma^\mathrm{lin}_{\bfM_0}([\j]^\mathrm{lin},[\varphi]^\mathrm{lin})$
and the right hand side can be simplified as follows: We write $\widetilde{\eta}([\varphi]^\mathrm{lin}) = 
\widetilde{\eta}_0([\varphi]^\mathrm{lin}) + \widetilde{\eta}_1([\varphi]^\mathrm{lin}) +
 \widetilde{\eta}_{\geq 2}([\varphi]^\mathrm{lin})$, where the index labels the $\bbN^0$-degree 
 of $\widetilde{\eta}([\varphi]^\mathrm{lin})$ in $\CPA_p^\mathrm{lin}(\bfM_0)$.
 This yields the condition, for all $[\varphi]^\mathrm{lin}\in \PhSp^\mathrm{lin}_p(\bfM_0)$
and all $\j\in C^\infty_0(M_0,\bbR^p)$,
 \begin{flalign}
 \sigma^\mathrm{lin}_{\bfM_0}([\j]^\mathrm{lin},[\varphi]^\mathrm{lin}) = \sigma^\mathrm{lin}_{\bfM_0}\big([\j]^\mathrm{lin},\widetilde{\eta}_1\big([\varphi]^\mathrm{lin}\big)\big) + \big\{[\j]^\mathrm{lin}, \widetilde{\eta}_{\geq 2}\big([\varphi]^\mathrm{lin}\big)\big\}_{\sigma^\mathrm{lin}_{\bfM_0}}~.
 \end{flalign}
Counting the $\bbN^0$-degree of the individual terms and using the fact that
the Poisson bracket in $\CPA_p^\mathrm{lin}(\bfM_0)$ is non-degenerate we obtain
that $\widetilde{\eta}_{\geq 2} =0$ and $\widetilde{\eta}_1 = \id_{\CPA_p^\mathrm{lin}(\bfM_0)}$.
Hence, $\widetilde{\eta}([\varphi]^\mathrm{lin}) = \widetilde{\eta}_0([\varphi]^\mathrm{lin}) + [\varphi]^\mathrm{lin}$,
for all $[\varphi]^\mathrm{lin}\in\PhSp_p^\mathrm{lin}(\bfM_0)$, and the remaining
freedom in $\widetilde{\eta}$ is a linear map $\widetilde{\eta}_0 : \PhSp_p^\mathrm{lin}(\bfM_0)\to \bbR$,
which also has to be Poincar{\'e} invariant.
By Lemma \ref{lem:transinvfnl}, $\widetilde{\eta}_0 \equiv 0$ in the case of $m\neq 0$ and $\widetilde{\eta}_0([\varphi]^\mathrm{lin})
= \int_{M_0}\ip{\varphi}{\mu}\, \vol_{\bfM_0}$ for some $\mu\in\bbR^p$ in the massless case.
Hence, there are no nontrivial endomorphisms of $\PA_p(\bfM_0,0)$
in the massive case. For $m=0$ the endomorphisms of  $\PA_p(\bfM_0,0)$
coincide with the Minkowski space components of the functor automorphisms
found in Proposition \ref{propo:R}. Since any endomorphism $\eta\in\End(\PA_p)$ is uniquely determined by 
its component on one object, this proves our claim.
\end{proof}

\subsection{\label{subsec:poissoncomposition}Composition property of the functor $\PA_p$}
It remains to prove the validity of the composition property of the covariant functor
$\PA_p:\LocSrc_p \to \PoisAlg$. Explicitly, we define for $p\geq 2$ and $0<q<p$ the covariant functor
\begin{flalign}
\PA_{p,q}:= \otimes \circ\big(\PA_q\times\PA_{p-q}\big)\circ \mathfrak{Split}_{p,q}:\LocSrc_p\to\PoisAlg
\end{flalign}
and we will prove that $\PA_{p,q}$ and $\PA_p$ are naturally isomorphic.
\begin{theo}\label{theo:classicalcompoA}
For any $2\leq p\in\bbN$ and $0<q<p$, the covariant functors $\PA_{p,q}:\LocSrc_p\to \PoisAlg$
and  $\PA_p : \LocSrc_p\to \PoisAlg$ are naturally isomorphic. The  natural isomorphism
$\eta : \PA_{p,q}\Rightarrow \PA_p$ is specified by,
for all $[(\varphi,\alpha)]\in \PA_q(\bfM,\bfJ^q)$ and $[(\psi,\beta)]\in \PA_{p-q}(\bfM,\bfJ^{p-q})$,
\begin{flalign}\label{eqn:compoiso}
\eta_{(\bfM,\bfJ)}\big([(\varphi,\alpha)]\otimes 1\big) = 
\big[(\varphi ,\alpha)\big]~,~~\eta_{(\bfM,\bfJ)}\big(1\otimes [(\psi,\beta)]\big) 
= \big[(\psi,\beta)\big]~,
\end{flalign}
where on the right hand sides we have identified $\varphi\in C^\infty_0(M,\bbR^q)$ and $\psi \in C^\infty_0(M,\bbR^{{p-q}})$
as elements in $C_0^\infty(M,\bbR^p)$ ($\varphi$ is placed in the first $q$  and 
$\psi$ in the last $p-q$ components of $\bbR^p$) .
\end{theo}
\begin{proof}
We first notice that (\ref{eqn:compoiso}) actually defines a Poisson algebra homomorphism
$\CPA_q(\bfM,\bfJ^q)\otimes \CPA_{p-q}(\bfM,\bfJ^{p-q}) \to\CPA_p(\bfM,\bfJ)$.
It induces a unital Poisson algebra homomorphism between the
quotients, 
$\PA_q(\bfM,\bfJ^q)\otimes \PA_{p-q}(\bfM,\bfJ^{p-q}) \to\PA_p(\bfM,\bfJ)$, 
since, for all $\alpha\in\bbR$,
\begin{flalign}
\eta_{(\bfM,\bfJ)}\big( \big([(0,\alpha)]-\alpha\big)\otimes 1\big) =\eta_{(\bfM,\bfJ)}\big(  1\otimes \big([(0,\alpha)]-\alpha\big)\big)=[(0,\alpha)]-\alpha ~.
\end{flalign}

Naturality of the $\eta_{(\bfM,\bfJ)}$ is also a straightforward check. 
We next show that $\eta_{(\bfM,\bfJ)}$ is invertible, hence a $\PoisAlg$-isomorphism.
Notice that setting, for any $[(\varphi,\alpha)]\in \PA_p(\bfM,\bfJ)$,
\begin{flalign}
\eta_{(\bfM,\bfJ)}^{-1}\big([(\varphi,\alpha)]\big) = \Big[ [(\varphi^q,\alpha)]\otimes 1 + 1 \otimes [(\varphi^{p-q},0)]\Big]\in 
\PA_q(\bfM,\bfJ^q)\otimes \PA_{p-q}(\bfM,\bfJ^{p-q}) ~,
\end{flalign}
where $\varphi = \varphi^q + \varphi^{p-q}$ is the split of $\varphi$ into the first $q$ and last $p-q$ components,
is well-defined and extends to a unital Poisson algebra homomorphism $\eta_{(\bfM,\bfJ)}^{-1}: \PA_p(\bfM,\bfJ)\to
\PA_q(\bfM,\bfJ^q)\otimes \PA_{p-q}(\bfM,\bfJ^{p-q})$.
One checks directly that $\eta_{(\bfM,\bfJ)}^{-1}$ is the inverse of $\eta_{(\bfM,\bfJ)}$.
\end{proof}

\section{\label{sec:quantization}Quantization}
We shall now turn to the quantization of our model. As a first step, we are going to use the
CCR-functor (in polynomial form) in order to construct a covariant functor
$\CQA_p := \CCR \circ \PhSp_p:\LocSrc_p\to \astAlg$, where $\astAlg$ is the category
of unital $\ast$-algebras over $\bbC$ with injective unital $\ast$-algebra homomorphisms as morphisms.
As $\CQA_p$ is a deformation quantization 
of the Poisson algebra functor $\CPA_p$, it is not surprising to find that its automorphism group is 
too large and that it violates the composition property. We then improve this functor following 
a strategy similar to that of Subsection \ref{subsec:improvedpoisson} for the Poisson algebras.
The essential step is to specify a suitable state space for $\CQA_p$. The kernel
corresponding to this state space forms a two-sided $\ast$-ideal in the algebras described
by $\CQA_p$, which when quotiented out leads to a covariant functor 
$\QA_p :\LocSrc_p\to \astAlg$ that has the correct automorphism group and 
satisfies the composition property. Accordingly, we find that $\QA_p$ is the correct
description of the quantum field theory of a multiplet of $p\in\bbN$ inhomogeneous Klein--Gordon fields
and not the functor $\CQA_p$, which was used in  \cite{Benini:2012vi}. 
As we shall explain in Remark \ref{rem:HW}, our improved
functor is naturally isomorphic to a functor obtained by the construction of Hollands and Wald \cite{Hollands:2004yh},
which is inspired by the Borchers--Uhlmann algebra \cite{Borchers,Uhlmann}.
Further equivalent constructions of  the improved theory $\QA_p$ are via the quantization of pointed presymplectic
spaces (see Appendix \ref{sec:pointed}), via deformation quantization of the improved classical theory 
(see Appendix \ref{sec:fedosov}) and via the `Fedosov inspired' approach of \cite{Sanders:2012sf} (see Appendix \ref{sec:fedosov}
for some comments).

\subsection{\label{subsec:canonical}Canonical algebras}
We briefly review the CCR-functor $\CCR: \PreSymp\to \astAlg$ in polynomial form, following the slightly non-standard approach
 taken in \cite{Baez:1992tj} and \cite[\S 5]{Fewster:2011pn}, which is equivalent to the
 standard presentation in terms of generators and relations.
To any object $(V,\sigma_V)$ in $\PreSymp$ we associate the following unital $\ast$-algebra $\CCR(V,\sigma_V)$:
The vector space underlying $\CCR(V,\sigma_V)$ is the complexification of the vector space 
underlying the symmetric tensor algebra $S(V):=\bigoplus_{k=0}^\infty S^k(V)$.
The involution $\ast$ is defined by $\bbC$-antilinearity and $(v_1\,\cdots\,v_k)^\ast = v_1\,\cdots\,v_k $, for all
$v_1,\dots,v_k\in V$, where juxtaposition denotes the symmetric product. 
The product ${\, \star \,}$ in $\CCR(V,\sigma_V)$ is specified (uniquely) by demanding, for all $v_1,v_2\in V$ and $n,m\in\bbN_0$,
\begin{flalign}\label{eqn:starproduct}
v_1^m\star v_2^n = \sum\limits_{r=0}^{\text{min}\{m,n\}} \left(\frac{i\sigma_V(v_1,v_2)}{2}\right)^r\,\frac{m!\,n!}{r!\,(m-r)!\,(n-r)!}\,
v_1^{m-r}\, v_2^{n-r}~.
\end{flalign}
To any  morphism $L: (V,\sigma_V)\to (W,\sigma_W)$ in $\PreSymp$ we associate the injective unital $\ast$-algebra 
homomorphism  $\CCR(L): \CCR(V,\sigma_V) \to \CCR(W,\sigma_W)$, which is specified by 
  $\CCR(L)(v_1\, v_2\, \cdots\, v_k) = L(v_1)\, L(v_2)\, \cdots\, L(v_k)$, for all $v_1,\dots,v_k\in V$, and $\bbC$-linearity. 
\sk

Composing the covariant functor $\PhSp_p : \LocSrc_p\to \PreSymp$ with $\CCR$ yields
the covariant functor $\CQA_p := \CCR \circ \PhSp_p:\LocSrc_p\to \astAlg$. It is standard that 
$\CCR$ preserves the time-slice axiom and the causality property.
In \cite{Benini:2012vi} $\CQA_p$ was taken to describe the quantized polynomial algebras of a multiplet of
$p\in\bbN$ inhomogeneous Klein--Gordon fields.
\sk

It is easy to see that the automorphism group $\Aut(\CQA_p)$ contains
a $\bbZ_2$-subgroup for the massive case and a $\bbZ_2\times \bbR^p$-subgroup
for $m=0$. This is an immediate consequence of Theorem \ref{theo:automorphismgroup}
and the fact that any automorphism $\eta=\{\eta_{(\bfM,\bfJ)}\}$ of $\PhSp_p$ lifts
to an automorphism of $\CQA_p$ with components $\{\CCR(\eta_{(\bfM,\bfJ)})\}$. 
\sk

To show that $\CQA_p$ violates the composition property, we define, for all $2\leq p\in\bbN$ and $0<q<p$, the covariant functor
\begin{flalign}
\CQA_{p,q} := \otimes \circ (\CQA_q\times \CQA_{p-q})\circ \mathfrak{Split}_{p,q} : \LocSrc_p\to\astAlg~,
\end{flalign}
where $\otimes : \astAlg \times \astAlg \to \astAlg$ now denotes  the covariant functor that takes the algebraic tensor
product of unital $\ast$-algebras. 
Adapting the proof of Lemma \ref{lem:naturalisoforpoisson},
one observes that the two covariant functors
$\CCR\circ \oplus : \PreSymp\times \PreSymp\to\astAlg$  and
$\otimes\circ (\CCR\times \CCR) : \PreSymp\times \PreSymp\to\astAlg$ are naturally isomorphic. 
Furthermore,  the result of Lemma \ref{lem:presymppoisalgiso} also extends to our present setting:   
two unital $\ast$-algebras $\CCR(V,\sigma_V)$ and $\CCR(W,\sigma_W)$
are isomorphic if and only if $(V,\sigma_V)$ and $(W,\sigma_W)$ are isomorphic as presymplectic vector spaces.
Then by an argument similar to that of  Proposition \ref{propo:compviolationPA} we obtain
\begin{propo}
For any $2\leq p \in\bbN$ and $0<q<p$, the covariant functors
$\CQA_{p,q}:\LocSrc_p\to\astAlg$ and $\CQA_p:\LocSrc_p\to \astAlg$
are not naturally isomorphic. Indeed, there is no object $(\bfM,\bfJ)$ in $\LocSrc_p$
for which the unital $\ast$-algebras $\CQA_{p,q}(\bfM,\bfJ)$ and $\CQA_p(\bfM,\bfJ)$
are isomorphic.
\end{propo}

\subsection{\label{subsec:improved}Improved algebras}
Employing a strategy similar to the one in Subsection \ref{subsec:improvedpoisson}, we now modify
the canonical algebras of Subsection \ref{subsec:canonical} in order to obtain the correct automorphism
group and satisfy the composition property. The essential idea is again to make our theory remember
that it came from affine functionals acting on the affine space of solutions of the inhomogeneous Klein--Gordon equation. 
\sk

We implement this idea mathematically by introducing suitable state spaces. Recall that a state space
 for a unital $\ast$-algebra $A$ is a subset $S$ of the set of normalized and positive linear 
 functionals on $A$ that is closed under convex linear combinations and operations induced by $A$. 
 The latter property means that,
given any state $\omega \in S$ and $b\in A$ such that $\omega(b^\ast b)>0$,
then the state $\omega_b(a ) := \omega(b^\ast a b)/\omega(b^\ast b)$, for all $a\in A$, is also contained in $S$.
To promote the concept of state spaces to the categorical setting, 
we define the category $\mathsf{State}$ as follows: The objects in $\mathsf{State}$ 
are all possible state spaces for objects in $\astAlg$, with affine maps as morphisms.   
A state space for our covariant functor $\CQA_p:\LocSrc_p\to\astAlg$ 
is a contravariant functor $\mathfrak{S}_p: \LocSrc_p\to\mathsf{State}$,
such that $\mathfrak{S}_p(\bfM,\bfJ)$ is a state space for $\CQA_p(\bfM,\bfJ)$ for each object
$(\bfM,\bfJ)$ and $\mathfrak{S}_p(f)=\CQA_p(f)^*|_{\mathfrak{S}_p(\bfM_2,\bfJ_2)}$ for every morphism 
$f:(\bfM_1,\bfJ_1)\to(\bfM_2,\bfJ_2)$ in $\LocSrc_p$ (it is a necessary condition
that $\CQA_p(f)^*[\mathfrak{S}_p(\bfM_2,\bfJ_2)]\subseteq \mathfrak{S}_p(\bfM_1,\bfJ_1)$). 
\begin{defi}
An {\bf admissible state space} $\mathfrak{S}_p: \LocSrc_p\to \mathsf{State}$
for the covariant functor $\CQA_p :\LocSrc_p\to\astAlg$ is a state space $\mathfrak{S}_p$, such that
for each object $(\bfM,\bfJ)$ in $\LocSrc_p$, and for all $\omega \in \mathfrak{S}_p(\bfM,\bfJ)$ and
$[(0,\alpha)],[(0,\beta)]\in \CQA_p(\bfM,\bfJ)$,
\begin{flalign}\label{eqn:admissiblestate}
\omega\big([(0,\alpha)]\big)=\alpha \quad,\qquad \omega\big([(0,\alpha)]\star [(0,\beta)]\big) = \alpha\,\beta~.
\end{flalign}
\end{defi}
\begin{rem}
The first condition in (\ref{eqn:admissiblestate}) demands that the expectation values of the quantum observables
corresponding to $[(0,\alpha)]$ agree with the classical result (\ref{eqn:pairing}). The second condition
in (\ref{eqn:admissiblestate}) sets the fluctuations around this classical result to zero, cf.~the lemma below. 
This behavior of states for the quantum theory is motivated by the fact that $[(0,\alpha)]$ 
corresponds in the classical theory to a constant functional.
\end{rem}
\begin{lem}
Let $\mathfrak{S}_p$ be any admissible state space for $\CQA_p$. Then for any object
$(\bfM,\bfJ)$ in $\LocSrc_p$, and for all $\omega\in \mathfrak{S}_p(\bfM,\bfJ)$ and
$[(0,\alpha_1)],\dots,[(0,\alpha_n)]\in \CQA_p(\bfM,\bfJ)$,
\begin{flalign}
\omega\big([(0,\alpha_1)]\star \cdots\star [(0,\alpha_n)]\big) = \alpha_1\,\cdots\,\alpha_n~.
\end{flalign}
\end{lem}
\begin{proof}
This is a straightforward consequence of the Cauchy-Schwarz inequality and a simple proof by induction.
Using the short notation $\widehat{\alpha} :=[(0,\alpha)]$ we obtain
\begin{flalign}
\nn \big\vert \omega\big(\widehat{\alpha_1}\star \cdots \star \widehat{\alpha_n}\big)-\alpha_1\cdots\alpha_n\big\vert^2
&=\big\vert \omega\big(\widehat{\alpha_1}\star \cdots \star \widehat{\alpha_n}\big)-\alpha_1\,\omega\big(\widehat{\alpha_2}\star\cdots\star\widehat{\alpha_n}\big)\big\vert^2\\
\nn &=\big\vert \omega\big((\widehat{\alpha_1}-\alpha_1)\star\widehat{\alpha_2}\star\cdots\star\widehat{\alpha_n}\big)\big\vert^2\\
&\leq \omega\big((\widehat{\alpha_1}-\alpha_1)^2\big)~\omega\big((\widehat{\alpha_2}\star\cdots\star\widehat{\alpha_n})^2\big)=0~,
\end{flalign}
where the last equality follows from the admissibility condition (\ref{eqn:admissiblestate}).
\end{proof}
\begin{lem}
There exists a non-empty admissible state space $\mathfrak{S}_p$  for $\CQA_p$, i.e.\
$\mathfrak{S}_p(\bfM,\bfJ)$ is non-empty for all objects $(\bfM,\bfJ)$ in $\LocSrc_p$.
\end{lem}
\begin{proof}
Let $\mathfrak{S}^\mathrm{max}_p(\bfM,\bfJ)$ be the set of all states 
 on $\CQA_p(\bfM,\bfJ)$ satisfying (\ref{eqn:admissiblestate}).
This set is non-empty, since it was shown in \cite[\S 8]{Benini:2012vi} 
that any state of the homogeneous Klein--Gordon theory induces a state in $\mathfrak{S}^\mathrm{max}_p(\bfM,\bfJ)$.
The admissibility condition of states in $\mathfrak{S}^\mathrm{max}_p(\bfM,\bfJ)$ is met by construction 
and it is preserved under convex linear combinations and operations induced by $\CQA_p(\bfM,\bfJ)$ 
(to prove the latter statement, use the Cauchy-Schwarz inequality 
  and the fact that $[(0,\alpha)]-\alpha$ lies in the center of $\CQA_p(\bfM,\bfJ)$).
 Thus, it remains to show that
 \begin{flalign}
 \mathfrak{S}^\mathrm{max}_p(f): \mathfrak{S}_p^\mathrm{max}(\bfM_2,\bfJ_2)\to \mathfrak{S}_p^\mathrm{max}(\bfM_1,\bfJ_1)~,~~\omega \mapsto \mathfrak{S}_p^\mathrm{max}(f)(\omega) = 
 \omega\circ \CQA_p(f)
 \end{flalign}
 is a morphism in $\mathsf{State}$, i.e.~that
  $\mathfrak{S}_p^\mathrm{max}(f)(\omega)\in\mathfrak{S}_p^\mathrm{max}(\bfM_1,\bfJ_1)$, for all 
  $\omega\in\mathfrak{S}_p^\mathrm{max}(\bfM_2,\bfJ_2)$. 
This holds because  $\mathfrak{S}_p^\mathrm{max}(f)(\omega)$
  is clearly a state, and obeys  (\ref{eqn:admissiblestate}) because $\CQA_p(f)\big([(0,\alpha)]\big) = [(0,\alpha)]$.
\end{proof}

Given any non-empty admissible state space $\mathfrak{S}_p$ for $\CQA_p$, we define 
for every object $(\bfM,\bfJ)$ in $\LocSrc_p$ 
\begin{flalign}\label{eqn:quantumideal}
\mathfrak{J}^{\mathfrak{S}_p}(\bfM,\bfJ) := \bigcap_{\omega\in \mathfrak{S}_p(\bfM,\bfJ)} \ker (\pi_\omega)\subseteq \CQA_p(\bfM,\bfJ)~,
\end{flalign}
where $\pi_\omega$ denotes the GNS-representation of $\CQA_p(\bfM,\bfJ)$ induced by
 $\omega \in\mathfrak{S}_p(\bfM,\bfJ)$. The subset (\ref{eqn:quantumideal}) of $\CQA_p(\bfM,\bfJ)$
 is clearly a two-sided $\ast$-ideal, and it must be proper because
kernels of unital algebra homomorphisms necessarily exclude the unit. Hence 
 $\CQA_p(\bfM,\bfJ)/\mathfrak{J}^{\mathfrak{S}_p}(\bfM,\bfJ)$ is a nontrivial unital $\ast$-algebra.
 \sk
 
 It will again be convenient to express $\mathfrak{J}^{\mathfrak{S}_p}(\bfM,\bfJ)$ in terms of an algebraically generated
 ideal. Let us consider the following two-sided $\ast$-ideal (generated by a set) of $\CQA_p(\bfM,\bfJ)$
 \begin{flalign}\label{eqn:quantumidealalgebraic}
 \widetilde{\mathfrak{J}}_p(\bfM,\bfJ) := \Big\langle \big\{[(0,\alpha)]-\alpha \in \CQA_p(\bfM,\bfJ):\alpha\in \bbR\big\}\Big\rangle~.
 \end{flalign}
 It is easy to see that  $\widetilde{\mathfrak{J}}_p(\bfM,\bfJ) \subseteq \mathfrak{J}^{\mathfrak{S}_p}(\bfM,\bfJ)$:
 Let $\omega \in \mathfrak{S}_p(\bfM,\bfJ)$ be arbitrary. Then, for all $b,c\in\CQA_p(\bfM,\bfJ)$ and all $\alpha\in\bbR$,
 \begin{flalign}
 \nn \Big\vert\omega\big(b\star \big([(0,\alpha)]-\alpha\big)\star c\big) \Big\vert^2&= 
 \Big\vert\omega\big(b\star c\star \big([(0,\alpha)]-\alpha\big)\big) \Big\vert^2 \\
 &\leq \omega\big(b\star c \star (b\star c)^\ast\big)\,\omega\big(\big([(0,\alpha)]-\alpha\big)^2\big) =0~,
 \end{flalign}
 where in the first step we have used that $[(0,\alpha)]-\alpha$ lies in the center of $\CQA_p(\bfM,\bfJ)$,
 in the second step the Cauchy-Schwarz inequality and in the last one the admissibility condition (\ref{eqn:admissiblestate}).
 Hence, $\widetilde{\mathfrak{J}}_p(\bfM,\bfJ) \subseteq \ker(\pi_\omega)$ and since $\omega$ was arbitrary we have 
 $\widetilde{\mathfrak{J}}_p(\bfM,\bfJ) \subseteq \mathfrak{J}^{\mathfrak{S}_p}(\bfM,\bfJ)$, for any
  non-empty admissible state space $\mathfrak{S}_p$.
 \begin{lem}\label{lem:quantidealproperties}
 Let $(\bfM,\bfJ)$ be any object in $\LocSrc_p$.
 \begin{itemize}
 \item[a)] The unital $\ast$-algebra $\CQA_p(\bfM,\bfJ)/\widetilde{\mathfrak{J}}_p(\bfM,\bfJ)$
 is (noncanonically) isomorphic to $\CQA^\mathrm{lin}_p(\bfM) := \CCR\big(\PhSp^\mathrm{lin}_p(\bfM)\big)$.
 \item[b)] $\widetilde{\mathfrak{J}}_p(\bfM,\bfJ) = \mathfrak{J}^{\mathfrak{S}_p}(\bfM,\bfJ)$ whenever
 $\mathfrak{S}_p$ is a non-empty admissible state space.
 \end{itemize}
 \end{lem}
 \begin{proof}
Proof of a): We define a unital $\ast$-algebra homomorphism
 $\kappa:\CQA_p(\bfM,\bfJ)\to\CQA_p^{\mathrm{lin}}(\bfM)$ by setting,
 for all $[(\varphi,\alpha)]\in\CQA_p(\bfM,\bfJ)$,
 \begin{flalign}
 \kappa\big([(\varphi,\alpha)]\big) = \omega_0\big([(\varphi,\alpha)]\big) + [\varphi]^\mathrm{lin}~,
 \end{flalign}
 where $\omega_0$ is any choice of admissible state. 
 As $\widetilde{\mathfrak{J}}_p(\bfM,\bfJ)$ clearly lies in the kernel of $\kappa$, we
 can induce a unital $\ast$-algebra homomorphism $\kappa : \CQA_p(\bfM,\bfJ)/\widetilde{\mathfrak{J}}_p(\bfM,\bfJ)
 \to\CQA_p^{\mathrm{lin}}(\bfM)$. To show that the induced $\kappa$ is a $\astAlg$-isomorphism
  we notice that setting, for all $[\varphi]^\mathrm{lin}\in\CQA_p^\mathrm{lin}(\bfM)$,
  \begin{flalign}
  \kappa^{-1}\big([\varphi]^\mathrm{lin}\big):= \Big[ [(\varphi,0)] - \omega_0\big([\varphi,0]\big)\Big]\in\CQA_p(\bfM,\bfJ)/\widetilde{\mathfrak{J}}_p(\bfM,\bfJ)
  \end{flalign}
is well-defined and defines the inverse of $\kappa$.
\sk

Proof of b): By a), $\CQA_p(\bfM,\bfJ)/\widetilde{\mathfrak{J}}_p(\bfM,\bfJ)$ is a simple nontrivial 
unital $\ast$-algebra. Hence $\widetilde{\mathfrak{I}}_p(\bfM,\bfJ)$ is a maximal proper ideal.
But $\widetilde{\mathfrak{J}}_p(\bfM,\bfJ) \subseteq \mathfrak{J}^{\mathfrak{S}_p}(\bfM,\bfJ)$ 
and $\mathfrak{J}^{\mathfrak{S}_p}(\bfM,\bfJ)$ is proper
so the ideals are equal. 
 \end{proof}
 \begin{rem}
 As a consequence of this lemma, the two-sided $\ast$-ideals
 $\mathfrak{J}^{\mathfrak{S}_p}(\bfM,\bfJ)$ do not depend on which (non-empty)
 admissible state space $\mathfrak{S}_p$ for $\CQA_p$ we use in the construction. 
We therefore introduce a simpler notation and set for any object $(\bfM,\bfJ)$ in
$\LocSrc_p$
\begin{flalign}
\mathfrak{J}_p(\bfM,\bfJ):= \mathfrak{J}^{\mathfrak{S}_p}(\bfM,\bfJ) = \widetilde{\mathfrak{J}}_p(\bfM,\bfJ)~.
\end{flalign}
  \end{rem}
  \sk
  
These studies now allow us to construct our improved functor 
for the quantum theory of a multiplet of $p\in\bbN$ inhomogeneous Klein--Gordon fields.
  \begin{propo}\label{propo:quantfunctorgood}
  The following rules define a covariant functor $\QA_p: \LocSrc_p\to\astAlg$:
  To any object $(\bfM,\bfJ)$ in $\LocSrc_p$ we associate 
  $\QA_p(\bfM,\bfJ):= \CQA_p(\bfM,\bfJ)/\mathfrak{J}_p(\bfM,\bfJ)$.
  To any morphism $f:(\bfM_1,\bfJ_1)\to(\bfM_2,\bfJ_2)$ in $\LocSrc_p$ we associate
  the map $\QA_p(f): \QA_p(\bfM_1,\bfJ_1)\to \QA_p(\bfM_2,\bfJ_2)$
  that is canonically induced from $\CQA_p(f):\CQA_p(\bfM_1,\bfJ_1)\to \CQA_p(\bfM_2,\bfJ_2)$.
  \end{propo}
  \begin{proof}
Lemma \ref{lem:quantidealproperties} has established that
the quotients are nontrivial unital $\ast$-algebras.
Next, let $f:(\bfM_1,\bfJ_1)\to (\bfM_2,\bfJ_2)$ be any morphism in $\LocSrc_p$. Then $\CQA_p(f)$ 
induces a unital $*$-homomorphism $\QA_p(f):\QA_p (\bfM_1,\bfJ_1)  \to \QA_p (\bfM_2,\bfJ_2)$ because    
it restricts to a map $\CQA_p(f):\mathfrak{J}_p(\bfM_1,\bfJ_1) \to \mathfrak{J}_p(\bfM_2,\bfJ_2) $.
  This is clear from the fact that $\CQA_p(f)\big( [(0,\alpha)]-\alpha \big) = [(0,\alpha)]-\alpha$, for any $\alpha\in\bbR$.
  The induced unital $\ast$-algebra homomorphism 
  $\QA_p(f): \QA_p(\bfM_1,\bfJ_1)\to\QA_p(\bfM_2,\bfJ_2)$ is injective (i.e.\ a morphism in $\astAlg$),
   since $\QA_p(\bfM_1,\bfJ_1)$ is simple, cf.\ Lemma \ref{lem:quantidealproperties}. 
The composition and identity properties
  of the association $\QA_p$ are consequences of the same properties
  of $\CQA_p$, hence $\QA_p: \LocSrc_p \to \astAlg$ is a covariant functor.
  \end{proof}
The following statement may be proved in complete analogy with Proposition \ref{propo:improvedpoisalgfunctorisLCFT}:
\begin{propo}
The covariant functor $\QA_p:\LocSrc_p\to\astAlg$ satisfies the causality property
and the time-slice axiom and is therefore a locally covariant quantum field theory.
\end{propo}
\begin{rem}\label{rem:fields}
The quantum field itself may now be introduced, following a
similar pattern to the classical fields of Remark~\ref{rem:classical_fields}. 
For each $(\bfM,\bfJ)$, let $\Phi_{(\bfM,\bfJ)}:C_0^\infty(M, \bbC^p)
\to \QA_p(\bfM,\bfJ)$ be the complex linear map
\begin{flalign}
\Phi_{(\bfM,\bfJ)}(\varphi) =\big[ [(\mathrm{Re}\,\varphi,0)]\big] + i \big[[(\mathrm{Im}\, \varphi,0)]\big]\,,
\end{flalign}
where the outer brackets denote the equivalence relation used in defining
$\QA_p$ and the inner square brackets that used in defining $\PhSp_p$. 
Then $\Phi_{(\bfM,\bfJ)}$ obeys hermiticity, i.e., $\Phi_{(\bfM,\bfJ)}(\overline{\varphi})= \Phi_{(\bfM,\bfJ)}(\varphi)^*$ for all
$\varphi\in C_0^\infty(M,\bbC^p)$, obeys 
the field equation in the form
\begin{flalign}\label{eq:feq}
\Phi_{(\bfM,\bfJ)}\big(\mathrm{KG}_{\bfM}(\varphi)\big) + 
\int_{M} \ip{\varphi}{\bfJ} \,\vol_{\bfM} =0~,
\end{flalign}
for all $\varphi\in C_0^\infty(M,\bbC^p)$, 
and obeys the canonical commutation relations
\begin{flalign}
\big[\Phi_{(\bfM,\bfJ)}(\varphi) \stackrel{\star}{,} \Phi_{(\bfM,\bfJ)}(\psi)\big]
= i \int_{M} \ip{\varphi}{\mathrm{E}_{\bfM}(\psi)} \,\vol_{\bfM}~,
\end{flalign}
for all $\varphi,\psi\in C_0^\infty(M,\bbC^p)$,
where we suppress explicit identity operators. 
Furthermore, if $f:(\bfM_1,\bfJ_1)\to (\bfM_2,\bfJ_2)$ is a morphism in $\LocSrc_{p}$
then the quantum fields are related by 
\begin{flalign}
\QA_p(f)\big(\Phi_{(\bfM_1,\bfJ_1)}(\varphi)\big)
=\Phi_{(\bfM_2,\bfJ_2)}\big(f_\ast(\varphi)\big) ~,
\end{flalign}
for all $\varphi\in C_0^\infty(M_1,\bbC^p)$. 
Defining a functor $\mathfrak{D}_p:\LocSrc_{p}\to \Vec_{\bbC}$ by  $\mathfrak{D}_p(\bfM,\bfJ)=
C_0^\infty(M,\bbC^p)$, $\mathfrak{D}_p(f) =f_\ast$, we
see that the maps $\Phi_{(\bfM,\bfJ)}$ provide the components
of a natural transformation $\Phi:\mathfrak{D}_p\Rightarrow\QA_p$ (suppressing a forgetful functor from $\astAlg$ to $\Vec_{\bbC}$)
and thus a locally covariant field.
All the above properties follow straightforwardly from our discussion
and we omit most of the details, save to mention that the field
equation~\eqref{eq:feq} is proved by a calculation
similar to \eqref{eq:cfeq}, but using  \eqref{eqn:quantumidealalgebraic}
in place of \eqref{eqn:algebraicPoissonideal}. 
\end{rem}
\begin{rem}\label{rem:HW}
The fields just introduced allow us to clarify the relation between our improved quantum algebras and
the algebras for the inhomogeneous Klein--Gordon theory used by Hollands and Wald \cite{Hollands:2004yh}, which  
are inspired by the Borchers--Uhlmann algebra \cite{Borchers,Uhlmann}.
In our notation, what Hollands and Wald propose is the following construction: Consider
the {\it off-shell} presymplectic vector space for a multiplet of $p\in\bbN$ Klein--Gordon fields
$\PhSp_p^\mathrm{HW}(\bfM,\bfJ) := \big(C^\infty_0(M,\bbR^p),\sigma_{\bfM}\big)$, 
where for all $\varphi,\psi\in C^\infty_0(M,\bbR^p) $,
$\sigma_{\bfM}(\varphi,\psi) = \int_M\ip{\varphi}{\mathrm{E}_{\bfM}(\psi)}\, \vol_{\bfM}$. Apply the CCR-functor
$\CCR(\PhSp_p^\mathrm{HW}(\bfM,\bfJ))$ and consider the two-sided $\ast$-ideal
\begin{flalign}\label{eqn:HWideal}
\mathfrak{I}_p^\mathrm{HW}(\bfM,\bfJ) := \left\langle \Big\{ \mathrm{KG}_{\bfM}(\varphi) +\int_M \ip{\varphi}{\bfJ} \, \vol_{\bfM} : 
\varphi \in C^\infty_0(M,\bbR^p)\Big\}\right\rangle~,
\end{flalign}
which is supposed to describe the inhomogeneous Klein--Gordon equation. The algebras of Hollands and Wald
are then defined by the quotient $\QA^{\mathrm{HW}}_p(\bfM,\bfJ):=
 \CCR(\PhSp_p^\mathrm{HW}(\bfM,\bfJ))/\mathfrak{I}_p^{\mathrm{HW}}(\bfM,\bfJ)$ 
 and it is easy to see that they are functorial, i.e.\ that we have a covariant 
 functor $\QA^\mathrm{HW}_p :\LocSrc_p\to \astAlg$. The covariant functor
 $\QA^\mathrm{HW}_p$ turns out to be naturally isomorphic to our functor $\QA_p$ given 
 in Proposition \ref{propo:quantfunctorgood}. Explicitly, the natural isomorphism $\kappa : \QA^\mathrm{HW}_p\Rightarrow \QA_p$
is given by setting, for all $\varphi\in C^\infty_0(M,\bbR^p)$,
\begin{flalign}
\kappa_{(\bfM,\bfJ)}(\varphi) =\Phi_{(\bfM,\bfJ)}(\varphi)~,
\end{flalign}
and extending as a $\astAlg$-morphism, which is possible
as a result of the properties of $\Phi$ set out in Remark~\ref{rem:fields}. 
\end{rem}

 
 \subsection{\label{subsec:quantrce}Relative Cauchy evolution of the functor $\QA_p$}
 The relative Cauchy evolution of the functor $\PhSp_p:\LocSrc_p\to \PreSymp$ induces that
of the functor $\QA_p: \LocSrc_p\to\astAlg$ as follows:
 Let $(\bfM,\bfJ)$ be any object in $\LocSrc_p$ and let $(\h,\j)\in H(\bfM,\bfJ)$ be any globally hyperbolic
 perturbation. From the explicit expression for $\rce_{(\bfM,\bfJ)}^{(\PhSp_p)}[\h,\j] \in\Aut(\PhSp_p(\bfM,\bfJ))$  
 given in (\ref{eq:rce_reform}) we observe that the relative Cauchy evolution $\rce_{(\bfM,\bfJ)}^{(\QA_p)}[\h,\j]\in\Aut(
 \QA_p(\bfM,\bfJ))$ of $\QA_p$ is uniquely specified by, for
 all $[(\varphi,\alpha)]\in \QA_p(\bfM,\bfJ)$,
 \begin{multline}\label{eqn:rceQA}
 \rce_{(\bfM,\bfJ)}^{(\QA_p)}[\h,\j]\big([(\varphi,\alpha)]\big) = \left[\left(\varphi + (\mathrm{KG}_{\bfM} - \mathrm{KG}_{\bfM[\h]})\big(\mathrm{E}_{\bfM[\h]}(\varphi)\big),0\right)\right] + \alpha \\
 +  \int_M \left(\ip{-\j}{\mathrm{E}_{\bfM[\h]}(\varphi)} + \ip{(1-\rho_{\h})\,(\bfJ+\j)}{\mathrm{E}_{\bfM[\h]}(\varphi)}\right)\,\vol_{\bfM}~,
 \end{multline}
 where on the right hand side we have used the equivalence relation entering the definition
 of $\QA_p(\bfM,\bfJ)$ (cf.\ Proposition \ref{propo:quantfunctorgood}) and we have chosen
 as in (\ref{eq:rce_reform})  a representative $\varphi$ with compact support in $M^+$.
In sufficiently regular representations of the algebra $\QA_p(\bfM,\bfJ)$ one can differentiate this
expression, yielding
\begin{flalign}
\nonumber \frac{d}{ds} \rce_{(\bfM,\bfJ)}^{(\QA_p)}[s\h,s\j]\big([(\varphi,\alpha)]\big)\Big\vert_{s=0} &= - \left[\left( \mathrm{KG}_{\bfM[\h]}^\prime \big(\mathrm{E}_{\bfM}(\varphi)\big),0\right)\right] \\
\nonumber & \qquad ~\qquad - \int_M\ip{\frac{1}{2} g^{ab}\,h_{ab}\,\bfJ + \j}{\mathrm{E}_{\bfM}(\varphi)}\,\vol_{\bfM}\\
& = i\,\left[\frac{1}{2}  T_{(\bfM,\bfJ)}(\h)+ [(\j,0)]\stackrel{\star}{,} [(\varphi,\alpha)]\right]~,\label{eq:QArcederiv}
\end{flalign}
where $T_{(\bfM,\bfJ)}(\h) = \int_M h_{ab}\,T^{ab}_{(\bfM,\bfJ)}\,\vol_{\bfM}$ is the smearing with $h_{ab}$ 
of the quantization of the stress-energy tensor 
 \eqref{eqn:fullSET}, with regularization by point-splitting (as emphasized in \cite{Brunetti:2001dx}, 
 the precise nature of the $c$-number subtraction
 is irrelevant owing to the commutator). Although \eqref{eqn:rceQA} was derived
 under an assumption on the support of the representative $\varphi$, 
 the formulae in \eqref{eq:QArcederiv} are valid for any
 representative $(\varphi,\alpha)$ of its equivalence class. Of course, \eqref{eq:QArcederiv} is the Dirac quantization of \eqref{eq:PArcederiv}. 
Finally, we note the special case of (\ref{eqn:rceQA}) for vanishing metric perturbation $\h=0$, namely
 \begin{flalign}\label{eqn:easyrceQAj}
 \rce_{(\bfM,\bfJ)}^{(\QA_p)}[0,\j]\big([(\varphi,\alpha)]\big) = [(\varphi,\alpha)] -\int_M\ip{\j}{\mathrm{E}_{\bfM}(\varphi)}
 = [(\varphi,\alpha)]  + i\,\Big[[(\j,0)]\stackrel{\star}{,}[(\varphi,\alpha)]\Big]~.
 \end{flalign}

 \subsection{\label{subsec:quantautomorphism}Automorphism group of the functor $\QA_p$}
 We study the automorphism group of the covariant functor $\QA_p:\LocSrc_p\to \astAlg$ defined
 in Proposition \ref{propo:quantfunctorgood}. For this we first notice that in the massless case $m=0$
 the automorphism group contains a $\bbR^p$ subgroup.
 \begin{propo}\label{propo:quantautogood}
 If $m=0$ there exists a faithful homomorphism $\eta:\bbR^p\to\Aut(\QA_p)$
 induced by the one in Proposition \ref{propo:Z2xR} restricted to $\{+1\}\times\bbR^p\subseteq \bbZ_2\times\bbR^p$.
 Explicitly, for any object $(\bfM,\bfJ)$ in $\LocSrc_p$ the automorphism $\eta(\mu)_{(\bfM,\bfJ)}$
 is specified by, for all $[(\varphi,\alpha)]\in \QA_p(\bfM,\bfJ)$,
 \begin{flalign}
 \eta(\mu)_{(\bfM,\bfJ)}\big([(\varphi,\alpha)]\big) = \Big[\Big(\varphi, \alpha + \int_M\ip{\varphi}{\mu} \, \vol_{\bfM}\Big) \Big]~.
 \end{flalign}
 \end{propo}
 \begin{proof}
 Applying the functor $\CCR$, the automorphism $\eta(\sigma,\mu)\in \Aut(\PhSp_p)$ of Proposition \ref{propo:Z2xR} induces an element in $\Aut(\CQA_p)$ (denoted with a slight abuse of notation by the same symbol).
 For $\sigma=-1$ this automorphism does not preserve the two-sided $\ast$-ideals $\mathfrak{J}_p(\bfM,\bfJ)$, since
 $\eta(-1,\mu)_{(\bfM,\bfJ)}\big([(0,\alpha)]-\alpha\big) = [(0,-\alpha)]-\alpha\not\in \mathfrak{J}_p(\bfM,\bfJ)$. 
 For $\sigma =+1$ and $\mu\in \bbR^p$
 arbitrary the two-sided $\ast$-ideals are preserved, hence $\eta(+1,\mu)$ induces the automorphism
 $\eta(\mu)\in\Aut(\QA_p)$ which is claimed in this proposition. The group law is an obvious consequence
of the group law of the automorphisms $\eta(\sigma,\mu)$ of Proposition \ref{propo:Z2xR}.
 \end{proof}
 \begin{rem} In the same way, one may also show for $m\neq 0$ that the nontrivial $\bbZ_2$-automorphism of $\PhSp_p$  
  does not lift to an automorphism of $\QA_p$.
 \end{rem}
 We may now prove that the automorphisms found in Proposition \ref{propo:quantautogood} exhaust $\Aut(\QA_p)$.
 We require the analog of Theorem \ref{theo:endobycomponent} for the functor 
 $\QA_p : \LocSrc_p\to\astAlg$, which can be obtained by a similar proof as in 
 Theorem \ref{theo:endobycomponent} and hence can be omitted. 
 \begin{theo}\label{theo:goodautgroupquant}
 Every endomorphism of the functor $\QA_p$ is an automorphism and
 \begin{flalign}
 \End(\QA_p) = \Aut(\QA_p) \simeq \begin{cases}
 \{\id_{\QA_p}\} & ~,~~\text{for } m\neq 0~,\\
 \bbR^p & ~,~~\text{for } m=0~,
 \end{cases}
 \end{flalign}
 where the action for $m=0$ is given by Proposition \ref{propo:quantautogood}.
 \end{theo}
 \begin{proof}
 The steps in this proof are similar to the ones in Theorem \ref{theo:goodautgroupclassical}.
 Let $\eta\in\End(\QA_p)$ be any endomorphism and let us consider its component  $\eta_{(\bfM_0,0)}$,
 where $\bfM_0$ is Minkowski spacetime. For this particular object, 
 the $\astAlg$-isomorphism $\kappa:\QA_p(\bfM_0,0)\to\CQA_p^\mathrm{lin}(\bfM_0)$ defined by 
 $\kappa\big([(\varphi,\alpha)]\big) = \alpha + [\varphi]^{\mathrm{lin}}$ intertwines the natural action of 
 the Poincar{\'e} transformations on $\QA_p(\bfM_0,0)$ and $\CQA_p^\mathrm{lin}(\bfM_0)$. 
Consequently, the endomorphism $\widetilde{\eta} := \kappa \circ \eta_{(\bfM_0,0)} \circ\kappa^{-1}$ 
of $\CQA^\mathrm{lin}_p(\bfM_0)$
 has to commute with all Poincar{\'e} transformations. 
 Furthermore, because $\eta_{(\bfM_0,0)}$ commutes with (derivatives of)
 the relative Cauchy evolution on $\QA_p(\bfM_0,0)$ -- in particular those
 with $\h=0$ -- we obtain the condition, for all $\j\in C^\infty_0(M_0,\bbR^p)$ and
   $[\varphi]^\mathrm{lin}\in\CQA_p^\mathrm{lin}(\bfM_0)$,
   \begin{flalign}\label{tmp:quantconditioncommutator}
   \widetilde{\eta}\Big(\Big[[\j]^\mathrm{lin}\stackrel{\star}{,} [\varphi]^\mathrm{lin}\Big]\Big) = 
   \Big[[\j]^\mathrm{lin}\stackrel{\star}{,} \widetilde{\eta}\big([\varphi]^\mathrm{lin}\big)\Big]~
   \end{flalign}
on $\widetilde{\eta}$, where we have used (\ref{eq:QArcederiv}) with $\h=0$. 
The left hand side of \eqref{tmp:quantconditioncommutator}, which is
analogous to \eqref{eqn:etatilde} in the proof of Theorem \ref{theo:goodautgroupclassical}, is simply $i\,\sigma_{\bfM_0}^\mathrm{lin}\big([\j]^\mathrm{lin},[\varphi]^\mathrm{lin}\big)$.
Using the explicit expression (\ref{eqn:starproduct}) for the $\star$-product in $\CQA_p^\mathrm{lin}(\bfM_0)$,
we find that the right hand side of this equality is equal to the Poisson bracket 
$i\, \big\{[\j]^\mathrm{lin},\widetilde{\eta}\big([\varphi]^{\mathrm{lin}}\big)\big\}_{\sigma_{\bfM_0}^\mathrm{lin}}$.
The remainder of the proof runs in complete analogy with that of Theorem \ref{theo:goodautgroupclassical}.
\end{proof}

\subsection{\label{subsec:quantcomposition}Composition property of the functor $\QA_p$}
It remains to prove that the covariant functor $\QA_p:\LocSrc_p\to \astAlg$
satisfies the composition property. We define for $p\geq 2$ and $0<q<p$
the covariant functor
\begin{flalign}
\QA_{p,q}:= \otimes \circ \big(\QA_q\times\QA_{p-q}\big) \circ \mathfrak{Split}_{p,q}:\LocSrc_p\to \astAlg~ 
\end{flalign}
and we obtain the following, in complete analogy with the proof of Theorem \ref{theo:classicalcompoA}
\begin{theo}
For any $p\geq 2$ and $0<q<p$, the covariant functors $\QA_{p,q}: \LocSrc_p\to\astAlg$
and $\QA_p:\LocSrc_p\to\astAlg$ are naturally isomorphic. The natural isomorphism $\eta = \{\eta_{(\bfM,\bfJ)}\}:
\QA_{p,q}\Rightarrow \QA_p$ is specified by, for all $[(\varphi,\alpha)]\in\QA_q(\bfM,\bfJ^q)$
and $[(\psi,\beta)]\in\QA_{p-q}(\bfM,\bfJ^{p-q})$,
\begin{flalign}
\eta_{(\bfM,\bfJ)}\big([(\varphi,\alpha)]\otimes 1\big) = \big[(\varphi,\alpha)\big]~,~~
\eta_{(\bfM,\bfJ)}\big(1\otimes [(\psi,\beta)]\big) = \big[(\psi,\beta)\big]~,
\end{flalign}
where on the right hand sides we have identified $\varphi\in C^\infty_0(M,\bbR^q)$ and $\psi\in C^\infty_0(M,\bbR^{p-q})$
as elements in $C^\infty_0(M,\bbR^p)$ ($\varphi$ is placed in the first $q$ and $\psi$ in the last $p-q$ components of $\bbR^p$).
\end{theo}

\subsection{\label{subsec:dynloc}Dynamical locality}
To conclude, we shall study whether or not our improved functor
$\QA_p:\LocSrc_p\to \astAlg$ satisfies the dynamical locality property,
which was introduced in \cite{Fewster:2011pe} as part
of an investigation into the question of what it means for a 
theory to describe the same physics in all spacetimes (SPASs).
The dynamical locality property has been proven previously for the homogeneous Klein--Gordon theory 
with non-vanishing mass $m\neq 0$ in \cite{Fewster:2011pn} and for
extended algebras of Wick polynomials in \cite{Ferguson:2014}.
\sk

We start by formulating the content of the dynamical locality property, 
essentially following \cite{Fewster:2011pe,Fewster:2011pn}, suitably
adapted to theories on the category $\LocSrc_p$. 
Let $(\bfM,\bfJ)$ be any object in $\LocSrc_p$. As above, we shall denote by $\mathscr{O}(\bfM)$
the set of all causally compatible, open and globally hyperbolic subsets of $\bfM$
with finitely many connected components all of which are mutually causally disjoint.
To each non-empty $O\in \mathscr{O}(\bfM)$, there is an object $(\bfM,\bfJ)\vert_O$
in $\LocSrc_p$ obtained by restricting all the geometric data (including the source $\bfJ$)
to the subset $O$ of $\bfM$. Moreover, there is a canonical inclusion $\iota_{(\bfM,\bfJ);O}: (\bfM,\bfJ)\vert_O\to (\bfM,\bfJ)$
which is a morphism in $\LocSrc_p$. Adapting an idea from \cite{Brunetti:2001dx}, we may construct
from $\QA_p$ a net of unital $\ast$-algebras as follows:
Given any non-empty $O\in\mathscr{O}(\bfM)$, we denote by $\QA_p^\mathrm{kin}((\bfM,\bfJ);O)$
the image of $\QA_p((\bfM,\bfJ)\vert_O)$ in $\QA_p(\bfM,\bfJ)$ under the $\astAlg$-morphism
$\QA_p(\iota_{(\bfM,\bfJ);O})$.
The assignment
\begin{flalign}
O \mapsto \QA_p^\mathrm{kin}((\bfM,\bfJ);O)\subseteq \QA_p(\bfM,\bfJ)
\end{flalign}
is called the {\it kinematic net}, and is one way of describing the local physics of the theory $\QA_p$ in
a region $O$ in the spacetime $\bfM$  underlying the object $(\bfM,\bfJ)$.
It is easily seen that $\QA_p^\mathrm{kin}((\bfM,\bfJ);O)$ is generated by the unit together with
all $[(\varphi,0)]\in\QA_p(\bfM,\bfJ)$ such that $\supp(\varphi)\subseteq O$.
\sk

Another description of the local physics of the theory $\QA_p$ in a region $O$ in the spacetime
$\bfM$  underlying the object $(\bfM,\bfJ)$ can be obtained by using the relative Cauchy evolution
 and was introduced in \cite{Fewster:2011pe}. For $K\subseteq M$ compact, let  us denote by
 $H((\bfM,\bfJ);K^\perp)$ the set of all globally hyperbolic perturbations 
 $(\h,\j)$ of $(\bfM,\bfJ)$,
 such that $\supp(\h)\cup\supp(\j)\subseteq K^\perp$, with $K^\perp := M\setminus J_{\bfM}(K)$ the causal complement
 of $K$. We define $\QA_p^\bullet((\bfM,\bfJ);K)$ to be the subalgebra of $\QA_p(\bfM,\bfJ)$
consisting of fixed points under arbitrary relative Cauchy evolutions $\rce_{(\bfM,\bfJ)}^{(\QA_p)}[\h,\j]$ 
with $(\h,\j)\in H((\bfM,\bfJ);K^\perp)$. 
The idea behind this definition is that the elements in $\QA_p^\bullet((\bfM,\bfJ);K)$
can be regarded as localized in $K$ because they are insensitive to perturbations $(\h,\j)$ of the background localized in the causal complement
$K^\perp$. By taking the subalgebra of $\QA_p(\bfM,\bfJ)$ that is generated by the $\QA_p^\bullet((\bfM,\bfJ);K)$
as $K$ ranges over suitable compact subsets of $O\in\mathscr{O}(\bfM)$ we obtain the {\em dynamical
algebras} $\QA_p^\mathrm{dyn}((\bfM,\bfJ);O)$, which can be compared with 
the kinematic ones $\QA_p^\mathrm{kin}((\bfM,\bfJ);O)$. 
More precisely, let us denote by $\mathscr{K}_b(\bfM;O)$ the set of finite unions of 
causally disjoint subsets of $O\in\mathscr{O}(\bfM)$, each of which is
the closure of a Cauchy ball $B$ with a relatively compact Cauchy development $D_{\bfM}(B)$.
Here,  a Cauchy ball $B$ is a subset of a Cauchy surface, for which there is a chart containing 
the closure of $B$, and in which $B$ is a non-empty open ball. 
With these definitions in place, we set for any non-empty $O\in\mathscr{O}(\bfM)$
\begin{flalign}
\QA_p^\mathrm{dyn}((\bfM,\bfJ);O) := \bigvee\limits_{K\in\mathscr{K}_b(\bfM;O)} \QA_p^\bullet((\bfM,\bfJ);K)\subseteq \QA_p(\bfM,\bfJ)~.
\end{flalign}
(The same algebra $\QA_p^\mathrm{dyn}((\bfM,\bfJ);O)$ is
obtained if $K$ runs over all compact subsets of $O$ possessing a
neighborhood of the form $\cup_{i=1}^k D_{\bfM}(B_i)$ where
each $B_i$ is a Cauchy ball contained in $O$ -- see 
\cite[Lemma 5.3.]{Fewster:2011pe}.)
\begin{defi}
The functor $\QA_p:\LocSrc_p\to \astAlg$ satisfies the {\it dynamical locality property} if, for all
objects $(\bfM,\bfJ)$ in $\LocSrc_p$ and all non-empty $O\in\mathscr{O}(\bfM)$, we
have 
\begin{flalign}\label{eqn:dynlocdef}
\QA_p^\mathrm{kin}((\bfM,\bfJ);O) = \QA_p^\mathrm{dyn}((\bfM,\bfJ);O)~.
\end{flalign}
\end{defi}
\begin{rem}
In its original formulation \cite{Fewster:2011pe}, dynamical locality was
defined using relative Cauchy evolution induced by metric perturbations,
because only theories defined on $\Loc$ were considered. In generalizing
to theories on $\LocSrc_p$, one has a choice as to whether to consider
perturbations in both the metric and the external source, or just the metric,
or potentially something intermediate. We have adopted the first of
these possibilities as being the most natural -- it would indeed appear
strange to regard as local an observable that was sensitive to perturbations
in the external source located in the causal complement of the localization
region. However, our consideration of massless inhomogeneous theories
will suggest a more nuanced view, which will be discussed below.
\end{rem}

Using the relative Cauchy evolution of the functor $\QA_p$ derived
in Subsection \ref{subsec:quantrce}, we can characterize
 the fixed point subalgebras $\QA_p^\bullet((\bfM,\bfJ);K)$ of 
 $\rce_{(\bfM,\bfJ)}^{(\QA_p)}[\h,\j]$ with $(\h,\j)\in H((\bfM,\bfJ);K^\perp)$.
 \begin{lem}\label{lem:fixedpointsubalgebras}
 Let $(\bfM,\bfJ)$ be any object in $\LocSrc_p$ and let $K$ be any compact subset of $\bfM$.
 Then $\QA_p^\bullet((\bfM,\bfJ);K)$ is the subalgebra of $\QA_p(\bfM,\bfJ)$ generated by the unit
 together with all $[(\varphi,0)]\in \QA_p(\bfM,\bfJ)$ such that  $\supp \big(\mathrm{E}_{\bfM}(\varphi)\big)\subseteq J_{\bfM}(K)$.
 \end{lem}
 \begin{proof}
The stated subalgebra of $\QA_p(\bfM,\bfJ)$ 
  is a subalgebra of $\QA_p^\bullet((\bfM,\bfJ);K)$ for the following reason: If  
   $[(\varphi,0)]$ obeys $\supp \big(\mathrm{E}_{\bfM}(\varphi)\big)\subseteq J_{\bfM}(K)$ then,
   for any $(\h,\j)\in H((\bfM,\bfJ);K^\perp)$, we have $\mathrm{E}_{\bfM[\h]}(\varphi) = \mathrm{E}_{\bfM}(\varphi)$ and hence
   $\rce_{(\bfM,\bfJ)}^{(\QA_p)}[\h,\j]\big([(\varphi,0)]\big) = [(\varphi,0)]$ by
    (\ref{eqn:rceQA}). Thus  $\QA_p^\bullet((\bfM,\bfJ);K)$ contains
    the subalgebra generated by (finite sums of finite products of) 
such elements and the unit.
\sk

To show the reverse inclusion, let us take any element 
$a\in \QA_p^\bullet((\bfM,\bfJ);K)$. In particular,
using (\ref{eqn:easyrceQAj}), we find the condition that, for all $\j\in C_0^\infty(K^\perp,\bbR^p)$,
$\big[[(\j,0)]\stackrel{\star}{,}a\big]=0$. Evaluating the $\star$-product (\ref{eqn:starproduct})
in the commutator, this condition reduces to the vanishing
Poisson bracket condition $\big\{[(\j,0)],a\big\}_{\sigma_{(\bfM,\bfJ)}}=0$, for all $\j\in C_0^\infty(K^\perp,\bbR^p)$.
We can now express $a$ as a finite sum of finite {\it symmetric} products of the unit and the 
elements $[(\varphi,0)]$ with $\varphi\in C^\infty_0(M,\bbR^p)$. 
Notice that if $a$ is one of the generators $[(\varphi,0)]$, with $\varphi\in C^\infty_0(M,\bbR^p)$,
then the vanishing Poisson bracket condition implies that $\supp\big(\mathrm{E}_{\bfM}(\varphi)\big)\subseteq J_{\bfM}(K)$.
For generic $a\in \QA_p^\bullet((\bfM,\bfJ);K)$ we follow the strategy in \cite[Lemma 5.2.\ and Appendix A]{Fewster:2011pn}
and associate to $a$ its {\it support subspace} $Y_a$, which is a finite dimensional vector subspace of the complex vector space
spanned by the $[(\varphi,0)]$, $\varphi\in C^\infty_0(M,\bbR^p)$,
such that the element $a$ lies in the subalgebra generated by $Y_a$ together with the unit. If $a\in \QA_p^\bullet((\bfM,\bfJ);K)$ then
$Y_a$ is invariant under the relative Cauchy evolution corresponding
to perturbations supported in $K^\perp$; considering relative
Cauchy evolutions of the form (\ref{eqn:easyrceQAj}), we see that
all $[(\varphi,0)]$ in the support subspace must satisfy $\supp\big(\mathrm{E}_{\bfM}(\varphi) \big)
\subseteq J_{\bfM}(K)$.  Hence, $a$ is generated only by (finite sums of finite {\it symmetric} products of) 
the unit and those generators $[(\varphi,0)]$ satisfying $\supp\big(\mathrm{E}_{\bfM}(\varphi)\big) \subseteq J_{\bfM}(K)$. 
As one can invert the formula (\ref{eqn:starproduct}) for the $\star$-product (leading to an expression for 
the symmetric product in terms of $\star$-products)
this implies that $a$ is also generated by finite sums of finite $\star$-products of the unit 
and the  elements $[(\varphi,0)]$ with $\supp\big(\mathrm{E}_{\bfM}(\varphi) \big)\subseteq J_{\bfM}(K)$.
 \end{proof}
 
With this preparation we can prove the main statement of this subsection.
\begin{theo}\label{theo:dynloc}
The functor $\QA_p: \LocSrc_p\to \astAlg$ satisfies the dynamical locality property.
\end{theo}
\begin{proof}
We must show that \eqref{eqn:dynlocdef} holds for all objects $(\bfM,\bfJ)$ in $\LocSrc_p$ and 
all non-empty $O\in\mathscr{O}(\bfM)$.
Notice that the unit is contained in both $\QA_p^\mathrm{kin}((\bfM,\bfJ);O)$ and $\QA_p^\mathrm{dyn}((\bfM,\bfJ);O)$.
To show the inclusion ``$\subseteq$'', note that any $a\in \QA_p^\mathrm{kin}((\bfM,\bfJ);O)$
is generated by finite sums of finite products of the unit and the elements $[(\varphi,0)]$
 with $\supp(\varphi)\subseteq O$, all of which may be 
shown to lie in $\QA_p^\mathrm{dyn}((\bfM,\bfJ);O)$ 
 by an argument similar to \cite[Lemma 3.3.]{Fewster:2011pn}: We can decompose $\varphi \in C^\infty_0(O,\bbR^p)$
 into a finite sum $\varphi = \sum_{i=1}^n\varphi_i$, such that $\supp(\varphi_i)\subset D_{\bfM}(B_i)$ for Cauchy ball $B_i$ with $\text{cl}\, B_i\in \mathscr{K}_b(\bfM;O)$.
 (Take for example an open cover of $\supp(\varphi)$ by diamonds, pass to a finite subcover 
 and then use a partition of unity.) For each $\varphi_i$ we have  $\supp\big(\mathrm{E}_{\bfM}(\varphi_i)\big)
 \subseteq J_{\bfM}(\supp(\varphi_i))\subseteq J_{\bfM}(B_i)$, which shows that 
$[(\varphi_i,0)] \in \QA_p^\bullet((\bfM,\bfJ);\text{cl}\,B_i)$ and hence
 $[(\varphi,0)] = \sum_{i=1}^n [(\varphi_i,0)] \in \QA_p^\mathrm{dyn}((\bfM,\bfJ);O)$.
 \sk
 
 To show the other inclusion ``$\supseteq$'', it is by Lemma \ref{lem:fixedpointsubalgebras}
 sufficient to prove that, for any $K\in\mathscr{K}_b(\bfM;O)$, all elements
 $[(\varphi,0)]\in\QA_p(\bfM,\bfJ)$ with $\supp\big(\mathrm{E}_{\bfM}(\varphi)\big)\subseteq J_{\bfM}(K)$ 
 are contained in $\QA_p^\mathrm{kin}((\bfM,\bfJ);O)$.
This is a simple consequence of the following argument:
Since $\mathrm{E}_{\bfM}(\varphi)$ has support in $J_{\bfM}(K)$ and $K\subseteq O$ is a compact subset,
there exists a $\varphi^\prime\in C^\infty_0(O,\bbR^p)$, such that $\mathrm{E}_{\bfM}(\varphi^\prime)
=\mathrm{E}_{\bfM}(\varphi)$, see e.g.\ \cite[Lemma 3.1.\ (i)]{Fewster:2011pn}. As the Klein--Gordon operator is normally hyperbolic,
we have $\varphi^\prime = \varphi + \mathrm{KG}_{\bfM}(h)$, for some $h\in C^\infty_0(M,\bbR^p)$.
Thus, $[(\varphi,0)] = [(\varphi + \mathrm{KG}_{\bfM}(h),\int_M\ip{\bfJ}{h}\,\vol_{\bfM})] 
= [(\varphi^\prime,0)]+ \int_M\ip{\bfJ}{h}\,\vol_{\bfM}$ lies in $\QA_p^\mathrm{kin}((\bfM,\bfJ);O)$.
\end{proof}
\begin{rem}\label{rem:dynloc}
The proofs of Lemma \ref{lem:fixedpointsubalgebras} and 
Theorem \ref{theo:dynloc} do not distinguish 
between the massless and the massive case. In contrast, this distinction
was essential 
for the {\it homogeneous} Klein--Gordon theory studied in \cite{Fewster:2011pn}; indeed only 
the massive homogeneous Klein--Gordon field
satisfies the dynamical locality property. At first sight this looks like a discrepancy, because
the homogeneous Klein--Gordon theory seems to be contained as a special case of our
inhomogeneous model by setting all source terms to zero. 
However, the inhomogeneous theory is formulated as a functor $\QA_p$ 
from the category $\LocSrc_p$ to $\astAlg$ and the relative Cauchy evolution $\rce_{(\bfM,\bfJ)}^{(\QA_p)}[\h,\j]$ depends
on both a metric perturbation $\h$
and a source term perturbation $\j$.  Even restricting to (the full subcategory of) objects with zero source term $\bfJ=0$, 
we still can study the response (via the relative Cauchy evolution) 
of the restricted theory to source term perturbations $\j$, as well as the response to metric perturbations $\h$, thus obtaining stronger restrictions on the fixed point subalgebras $\QA_p^\bullet((\bfM,\bfJ);K)$
than those arising from metric perturbations $\h$ alone as in \cite{Fewster:2011pn}.  To conclude, we point out that if we {\it forbade} nontrivial source term perturbations $\j$ in our proofs above, i.e.\ making
only use of the relative Cauchy evolutions $\rce_{(\bfM,\bfJ)}^{(\QA_p)}[\h,0]$ depending on $\h$, 
we would obtain as in \cite{Fewster:2011pn} (and by similar arguments) that the massive theory satisfies (this restricted form of) the dynamical
locality property and that the massless theory does not. 
We finally remark that if we were to restrict to coexact source perturbations $\j=\delta\alpha$
 (for compactly supported one-forms $\alpha$) 
 and traceless metric perturbations $\h$ we would also lose dynamical locality in the massless case.
 This will be discussed further below. 
\end{rem}


\section{\label{sec:massless}Gauge theory interpretation in the massless case}
In this section we shall briefly point out and discuss some features of the classical and quantum theory
of a massless multiplet of $p\in\bbN$ inhomogeneous Klein--Gordon fields.
\sk

Let us start with the automorphisms of this theory. As was shown in Theorem \ref{theo:goodautgroupclassical}
for the classical and in Theorem \ref{theo:goodautgroupquant} for the quantized case,
this theory has a nontrivial automorphism group isomorphic to $\bbR^p$. 
These symmetries can also be understood from the Lagrangian of this model (see equation 
(\ref{eqn:actionintro}) with $m=0$) as they correspond to shifts of the classical field 
$\phi$, i.e.\ transformations $\phi \mapsto \phi + \mu$ with $\mu\in\bbR^p$.
According to \cite{Fewster:2012yc}  one should regard the massless Klein--Gordon theory 
as a gauge theory of the first kind with $\phi$ playing the role of a zero-form gauge field.
This is supported by the fact that the Lagrangian can also be written as
\begin{flalign}\label{eqn:lagrangian}
\mathcal{L} = \frac{1}{2} \ip{\mathrm{d}\phi}{\ast \mathrm{d}\phi} -\ip{\phi}{\ast\bfJ}~,
\end{flalign}
where $\ast$ denotes the Hodge operator corresponding to $\bfM$. The differentials
$\mathrm{d}\phi$ play the same role as the field strength $F = \mathrm{d} A$ in electromagnetism, just one
differential form degree lower. Under gauge transformations $\phi \mapsto \phi^\prime = \phi + \mu$, $\mu\in \bbR^p$,
 the Lagrangian transforms as
\begin{flalign}\label{eqn:tmplangrangiantransformation}
\mathcal{L} \mapsto \mathcal{L}^\prime = \mathcal{L} - \ip{\mu}{\ast\bfJ}~,
\end{flalign}
thus it is gauge invariant up to a $\phi$-independent term $- \ip{\mu}{\ast\bfJ}$,
which however depends on the metric via the Hodge operator.
In particular, the gauge transformations map the solution space
of the inhomogeneous massless Klein--Gordon equation to itself.
This global gauge invariance is exactly the one described by the automorphism groups characterized
in Theorem \ref{theo:goodautgroupclassical} and Theorem \ref{theo:goodautgroupquant}.
With this interpretation in mind, the observables of the theory should be identified
with those elements of the Poisson or quantized algebras that are fixed under the
action of the automorphism group. As described in \cite[\S 3.3]{Fewster:2012yc}  
this would lead in a natural way to subfunctors of $\PA_p$ and $\QA_p$ that
can be interpreted as the `theories of observables'. This strategy was implemented
for the massless homogeneous Klein--Gordon theory in  \cite[\S 5.3]{Fewster:2012yc}  and, 
while we have not worked through the analogue
for the present models, our expectation is that it would result in the 
theories obtained by the following construction: For any object 
$(\bfM,\bfJ)$ in $\LocSrc_p$ we take the vector subspace $\PhSp_p^\mathrm{inv}(\bfM,\bfJ)\subseteq \PhSp_p(\bfM,\bfJ)$
consisting of all $[(\varphi,\alpha)]$, such that $\int_M \ip{\varphi}{\mu}\,\vol_{\bfM}=0$ for all $\mu\in \bbR^p$.
It is easy to see that $\PhSp_p^\mathrm{inv} : \LocSrc_p\to\PreSymp$ is a subfunctor
of $\PhSp_p$ and that, by the same construction as in Section \ref{sec:poisson} and Section \ref{sec:quantization},
we arrive at subfunctors $\PA_p^\mathrm{inv}:\LocSrc_p\to\PoisAlg$ and $\QA_p^\mathrm{inv}: \LocSrc_p\to \astAlg$
of, respectively, $\PA_p$ and $\QA_p$, which are gauge-invariant. (The
remaining issue is whether they coincide with the fixed-point subtheories of
$\PA_p$ and $\QA_p$, but this is our expectation.) 
\sk

The role of this gauge invariance is obscured when we study globally hyperbolic perturbations $(\h,\j)$ 
of the background $(\bfM,\bfJ)$ via the relative Cauchy evolution. The stress-energy tensor
(see (\ref{eqn:fullSET}) and set $m=0$) obtained by the $\h$-derivative of the relative Cauchy evolution is
{\it not} gauge invariant under $\phi\mapsto \phi^\prime = \phi + \mu$, $\mu\in\bbR^p$; 
it transforms as
\begin{flalign}
T_{(\bfM,\bfJ)}^{ab}[\phi] \mapsto T_{(\bfM,\bfJ)}^{ab}[\phi^\prime] = T_{(\bfM,\bfJ)}^{ab}[\phi] + g^{ab}\ip{\mu}{\bfJ}~,
\end{flalign}
and therefore is not an observable according to our discussion above.\footnote{The stress-energy tensor does not
actually belong to the algebras we have considered; here we have in mind an 
extended (quantum) algebra containing (Wick) products.} 
This feature becomes again clear by looking at the transformation property of the
Lagrangian (\ref{eqn:tmplangrangiantransformation}): In fact, the stress-energy tensor
derived from the transformed Lagrangian $\mathcal{L}^\prime$ via taking the
functional derivative along $g_{ab}$ does not coincide with the one obtained from 
the untransformed Lagrangian $\mathcal{L}$ due to the metric-dependent extra term in  
the transformation law (\ref{eqn:tmplangrangiantransformation}). 
Note, however, 
that smearings of the stress-energy tensor by {\em traceless} tensor fields are gauge-invariant 
and thus qualify as observables.  Likewise, smearings of the
field against test functions that are derivatives of compactly supported $1$-forms also give observables. 
This gives an interesting perspective on some
of the points made in Remark~\ref{rem:dynloc}: the massless inhomogeneous theory
fails to be dynamically local if one restricts to variations of background
structures with relative Cauchy evolution generated by observable fields,
but is dynamically local if one also allows variations generated by
unobservable fields. 
\sk 

In order to obtain a gauge invariant
stress-energy tensor we might proceed as follows: If we replace
the source terms $\bfJ\in C^\infty(M,\bbR^p)$ by source terms that are top-form valued
$\bfJt\in\Omega^{\dim(M)}(M,\bbR^p)$, the Lagrangian (\ref{eqn:lagrangian}) can be written as
\begin{flalign}
\widetilde{\mathcal{L}} = \frac{1}{2} \ip{\mathrm{d}\phi}{\ast\mathrm{d}\phi} - \ip{\phi}{\bfJt}~.
\end{flalign}
Under gauge transformations $\phi \mapsto \phi + \mu$, $\mu\in\bbR^p$, the Lagrangian transforms as
\begin{flalign}
\widetilde{\mathcal{L}} \mapsto \widetilde{\mathcal{L}}^\prime = \widetilde{\mathcal{L}} - \ip{\mu}{\bfJt}~,
\end{flalign}
where now the $\phi$-independent additional term does not depend on the metric. 
As a consequence, the stress-energy tensor obtained by the functional derivative of the Lagrangian
along $g_{ab}$ reads
\begin{flalign}\label{eqn:tildedstress}
\widetilde{T}_{(\bfM,\bfJt)}^{ab}[\phi] =
 \ip{\nabla^a \phi }{\nabla^b \phi} - \frac{1}{2}g^{ab} \ip{\nabla_c \phi }{\nabla^c \phi}
\end{flalign}
and is gauge invariant. The functorial theory with top-form valued source terms can be obtained from
our functors $\PA_p:\LocSrc_p\to \PoisAlg$ and $\QA_p:\LocSrc_p\to\astAlg$ by noticing the
following equivalence of categories: Let us define in analogy to $\LocSrc_p$ (see Definition \ref{defi:LocSrcp})
the category $\LocTop_p$, where objects are tuples $(\bfM,\bfJt)$ with $\bfJt\in\Omega^{\dim(M)}(M,\bbR^p)$
a top-form source term. The categories $\LocSrc_p$ and $\LocTop_p$ are equivalent
via the Hodge operator. Explicitly, we define the covariant functor 
$\mathfrak{Hodge} : \LocTop_p \to \LocSrc_p$ on objects by $\mathfrak{Hodge}(\bfM,\bfJt) = (\bfM,\ast\bfJt)$
and on morphisms by $\mathfrak{Hodge}(f) = f$ (with a slight abuse of notation we denote both a morphism
and its underlying smooth map by the same symbol). The inverse Hodge operator provides us with 
the inverse functor $\mathfrak{Hodge}^{-1} : \LocSrc_p\to\LocTop_p$.
Hence, $\LocSrc_p$ and $\LocTop_p$ are equivalent categories.
Composing the functor $\mathfrak{Hodge}$ with our functors $\PA_p$ and $\QA_p$ we obtain
the covariant functors
\begin{subequations}
\begin{flalign}
\widetilde{\PA}_p&:= \PA_p\circ \mathfrak{Hodge} : \LocTop_p\to\PoisAlg~,\\
\widetilde{\QA}_p &:= \QA_p\circ \mathfrak{Hodge}: \LocTop_p\to \astAlg~,
\end{flalign} 
\end{subequations}
which are respectively a locally covariant classical and quantum field theory. The endomorphisms
of $\widetilde{\PA}_p$ and $\widetilde{\QA}_p$ of course coincide with the ones of $\PA_p$ and $\QA_p$.
However, as the functor $\mathfrak{Hodge}$ mixes between the metric and the external source terms,
the relative Cauchy evolution of the new theories differs from the that of the original theories. Indeed,
from (\ref{eqn:tmpcauchy}) and (\ref{eq:rce_reform}) one easily observes that,
for all $[(\varphi,\alpha)]\in \widetilde{\PA}_p(\bfM,\bfJt)$,
\begin{flalign}
\rce_{(\bfM,\bfJt)}^{(\widetilde{\PA}_p)}[\h,\widetilde{\j}]\big([(\varphi,\alpha)]\big) = \left[\left(\varphi + (\mathrm{KG}_{\bfM} - \mathrm{KG}_{\bfM[\h]})\big(\mathrm{E}_{\bfM[\h]}(\varphi)\big),\alpha -\int_M\ip{\widetilde{\j}}{\mathrm{E}_{\bfM[\h]}(\varphi)}\vol_{\bfM}\right)\right] ~,
\end{flalign}
where now $\widetilde{\j}\in\Omega^{\dim(M)}_0(M,\bbR^p)$ is a compactly supported perturbation of
the top-form source term $\bfJt$. A similar formula holds true for the relative Cauchy evolution
of $\widetilde{\QA}_p$.
Following the same steps as in Appendix \ref{app:Cauchy},
we can extract the stress-energy tensor
(up to a constant functional) from the derivative of this relative Cauchy evolution along $\h$.
 In the massless case, we find exactly the one obtained from the Lagrangian, see (\ref{eqn:tildedstress}).
As already mentioned, this stress-energy tensor is gauge invariant under the gauge transformations
$\phi \mapsto \phi + \mu$, $\mu\in\bbR^p$.


\section{\label{sec:conclusion}Concluding remarks}
Our original aim was to understand how the methods of \cite{Fewster:2012yc} could be extended to the setting of 
locally covariant theories with external sources in order to compute the automorphism group (which should be the global gauge group) of the inhomogeneous Klein--Gordon theory.  
This has brought to light various shortcomings in the
formulation of the theory according to the prescription of \cite{Benini:2012vi}: 
it has automorphisms that are not gauge symmetries of the
original theory (cf.\ Theorem \ref{theo:automorphismgroup});
furthermore, it violates a natural composition property  that expresses the lack of interaction between fields in the multiplet (cf.\ Proposition \ref{propo:compprop}). To remedy these problems, we have
proposed an improved formulation of such theories at the classical and quantum level. 
We have traced the source of the pathological behavior to a failure of
\cite{Benini:2012vi} to adequately capture 
the interplay between the observables spaces (described by the presymplectic vector spaces) with the solution spaces. 
We have reintroduced this information by
studying the representation of the abstract Poisson algebras derived from these presymplectic vector spaces
on the solution spaces. These representations of the Poisson algebras have a kernel, 
which has no corresponding analog in the category of presymplectic vector spaces.
Performing the quotient by these kernels, we have obtained a functor to the category of Poisson algebras
which gives an appropriate description of the classical theory of a multiplet of inhomogeneous Klein--Gordon
fields. We have substantiated this claim by proving that the theory has the correct automorphism group and satisfies the composition property. In the quantized setting, we have replaced the pairing 
between observables and solutions with carefully defined state spaces on the CCR-algebras
derived from our presymplectic vector spaces. The GNS representation of our CCR-algebras in these
state spaces has a kernel, and we have shown that the quantum algebras obtained by quotienting out
these kernels are given functorially. Again, we have justified our constructions by showing that
the automorphism group of this improved functor is the correct one and that the composition property holds true. 
\sk

In this paper we have restricted ourselves to the simplest case given by a multiplet of inhomogeneous
Klein--Gordon fields, as this choice made it possible to characterize explicitly the relative Cauchy evolution
and the automorphism groups, which were important tools in unraveling the pathological features
of the earlier approach to affine field theories \cite{Benini:2012vi}. However, our insights concerning the improved functors
describing the classical and quantum theory of this model remain valid for generic affine field theories
as described in \cite{Benini:2012vi}. In particular, the general presymplectic vector 
space functor in \cite{Benini:2012vi} can be promoted via $\CanPois$ to a covariant functor with values in the
category of Poisson algebras. This functor can be paired with the solution space functor corresponding
to the equation of motion operators, which are part of the source category in  \cite{Benini:2012vi},
and the corresponding kernel forms a Poisson ideal. The improved classical functor for generic
affine field theories is then given by taking the quotient of the canonical Poisson algebras by these Poisson ideals.
In the quantized setting one proceeds analogously to Section \ref{sec:quantization}, i.e.\ one defines suitable state spaces
and studies the kernels of the corresponding GNS representation. The 
same techniques apply to Abelian gauge theories
\cite{Benini:2013tra,Benini:2013ita}, where however the following remark is in order: The kernel of the presymplectic
spaces in \cite{Benini:2013tra,Benini:2013ita} does {\it not} only consist of constant affine functionals,
but also topological observables (`electric charges') depending on the topology of spacetime. This implies that quotienting out 
the kernel of the pairing between the abstract Poisson algebras and the solution spaces does not generally yield simple Poisson algebras. The same holds true for quotienting out the kernel
of the GNS representation of the CCR-algebras given by suitable state spaces. This additional
degeneracy in the improved Poisson algebras (or the center in the improved quantum algebras) 
is the reason that Abelian gauge theories violate the injectivity axiom of locally covariant quantum field theory. 
As is clear from the general no-go theorem in \cite{Benini:2013ita}, 
our approach does not resolve this issue and hence 
reconciling gauge theory with locally covariant quantum field theory remains an open problem for future work.
\sk

Finally, we have shown that our improved quantum theory
is equivalent to that used by Hollands and Wald \cite{Hollands:2004yh} for 
the inhomogeneous Klein--Gordon theory -- see Remark \ref{rem:HW}.  
However, we would like to stress that our techniques 
are more general and flexible; in fact, in cases where the 
theory is already affine at the kinematical level (as e.g.\ in Abelian gauge theory),
the Borchers--Uhlmann algebra of the {\it affine dual} of the configuration space 
has to be considered, which will have a similar pathological behavior as our canonical algebras
in Subsection \ref{subsec:canonical} and hence has to be modified by our methods developed in Subsection \ref{subsec:improved}.
Moreover, using our techniques, we have determined a number of detailed properties of these inhomogeneous field theories
and exemplified the opportunities for analyzing and distinguishing locally covariant theories
 by functorial invariants opened up by~\cite{Fewster:2012yc}.


\section*{Acknowledgements}
AS would like to thank the Department of Mathematics of the University of York for hospitality
during a visit in September 2013. CJF and AS thank the Mathematisches Forschungsinstitut 
Oberwolfach for support and hospitality during the 
Mini-Workshop ``New Crossroads between Mathematics and Field Theory'', at which this work commenced, and also thank the
workshop organizers. CJF thanks Atsushi Higuchi, Daniel Siemssen and 
Katarzyna Rejzner for useful discussions.
We also thank Ko Sanders for prompting us to clarify the relationship between some
of our results and those of reference \cite{Sanders:2012sf}.
Furthermore, we thank the referees for their careful evaluation of our manuscript and their useful comments and suggestions.


\appendix
\section{\label{app:Cauchy}Differentiation of the relative Cauchy evolution and stress-energy tensor}
 
We discuss the sense in which the relative Cauchy evolution of
$\PhSp_p$ can be differentiated, and show that it has the
derivative stated in the text. Our method follows that of
\cite[Appendix A]{Fewster:2011pn}; however, additional care
must be taken when defining a suitable topology. In \cite{Fewster:2011pn}, the relative Cauchy evolution for the 
(homogeneous) real scalar field was differentiated using the 
weak symplectic topology on the spacelike compact solution space (describing the linear observables of the theory)
induced by seminorms of the form $|\sigma_{\bfM}(\,\cdot\,,\phi)|$, where
$\sigma_{\bfM}$ is the symplectic structure and $\phi$ ranges over the 
symplectic vector space. An obvious generalization to our present 
context is to induce a topology from the presymplectic structure
on $\PhSp_p(\bfM,\bfJ)$ in a similar way. However, the resulting
topology does not separate points, because the presymplectic structure is
degenerate. In particular, it is clear from \eqref{eqn:tmpcauchy} that the presymplectic pairing between any element with 
an element $\mathrm{rce}_{(\bfM,\bfJ)}^{(\PhSp_p)}[\h,\j]\big([(\varphi,\alpha)]\big)$
is independent of $\j$ -- it would therefore be
impossible to obtain the `obviously correct' derivative given in
\eqref{eqn:Jdefinition}.
\sk

The solution to this problem is, as at various other points in this paper, 
to recall that $\PhSp_p(\bfM,\bfJ)$ acts as a space of functionals 
on the affine solution space $\Sol_p(\bfM,\bfJ)$ of the inhomogeneous
theory. Dually, therefore, each solution $\phi\in \Sol_p(\bfM,\bfJ)$ 
defines a seminorm $|\!\sip{\,\cdot\,}{\phi}_{(\bfM,\bfJ)}\!|$, where
the pairing $\sip{\,\cdot\,}{\,\cdot\,}_{(\bfM,\bfJ)}$ was defined in \eqref{eqn:pairing}. The 
resulting weak-$*$ topology does separate points and is the appropriate generalization of the weak symplectic
topology used in~\cite{Fewster:2011pn}.
With the topology fixed, differentiability of the relative Cauchy evolution may be 
established as follows. {}From \eqref{eq:rce_reform} we have
\begin{multline}
 \sip{\Big(\mathrm{rce}_{(\bfM,\bfJ)}^{(\PhSp_p)}[\h,\j]-\id_{\PhSp_p(\bfM,\bfJ)}\Big)\big([(\varphi,\alpha)]\big) }{\phi}_{(\bfM,\bfJ)} \\\quad = 
\int_M   \Big(\ip{(\mathrm{KG}_{\bfM}-\mathrm{KG}_{\bfM[\h]})\big(\mathrm{E}_{\bfM[\h]}(\varphi)\big)}{\phi} 
-\ip{\j}{\mathrm{E}_{\bfM[\h]}(\varphi)} 
+  (1-\rho_{\h}) \ip{\bfJ+\j}{\mathrm{E}_{\bfM[\h]}(\varphi)}\Big)\, \vol_{\bfM}~,
\end{multline}
and the integration region may be restricted, without loss, to any strip\footnote{Identifying $\bfM$ diffeomorphically with a product manifold of form
$\bbR\times\Sigma$, so that $\{t\}\times\Sigma$ is spacelike for
all $t$, a {\em strip} of $\bfM$ is any subset of the form $I\times\Sigma$
where $I\subset \bbR$ is a closed interval.} $S$ of $\bfM$ containing the support of $\j$ and $\h$. As in \cite[Appendix B]{Fewster:2011pn}, energy 
estimates entail that $s\mapsto \mathrm{E}_{\bfM[s\h]}(\varphi)$
is differentiable in $L^2(S,\vol_{\bfM})\otimes \bbC^p$; moreover $s\mapsto 1-\rho_{s\h}$ is smooth in $L^2(S,\vol_{\bfM})$ near $s=0$, with first derivative $-\frac{1}{2}g^{ab}h_{ab}$ at $s=0$. It follows that
\begin{subequations}
\begin{flalign}
 \int_M  \ip{s \j}{\mathrm{E}_{\bfM[s \h]}(\varphi)}\, \vol_{\bfM}  = s\, 
\int_M  \ip{\j}{\mathrm{E}_{\bfM}(\varphi)}\, \vol_{\bfM} + O(s^2)
\end{flalign}
and 
\begin{flalign}
\int_M   (1-\rho_{s \h}) \, \ip{\bfJ+s \j}{\mathrm{E}_{\bfM[s \h]}(\varphi)}\, \vol_{\bfM}
= -\frac{s}{2} \int_M g^{ab}h_{ab}\, \ip{\bfJ}{\mathrm{E}_{\bfM}(\varphi)} \, \vol_{\bfM}+ O(s^2)~.
\end{flalign}
\end{subequations}
Moreover, the formula
\begin{flalign}
\int_M  \ip{\big(\mathrm{KG}_{\bfM}-\mathrm{KG}_{\bfM[s\h]}\big)\big(\mathrm{E}_{\bfM[s\h]}(\varphi)\big)}{\phi}\, \vol_{\bfM} 
= -s\int_M  \ip{\mathrm{KG}'_{\bfM[\h]} \big(\mathrm{E}_{\bfM}(\varphi)\big)}{\phi}\, \vol_{\bfM} +O(s^2)
\end{flalign}
was established in \cite[Appendix B]{Fewster:2011pn}. Note that
here $\phi$ solves the inhomogeneous equation, while its analogue
in the cited reference solved the homogeneous equation. However the
difference is inessential, because the only property of $\phi$ used
is that it is square-integrable on the strip $S$. Assembling these observations, 
\begin{flalign}\label{eq:rce_deriv1}
\left.\frac{d}{ds}\sip{\mathrm{rce}_{(\bfM,\bfJ)}^{(\PhSp_p)}[s\h,s\j] \big([(\varphi,\alpha)]\big) }{\phi}_{(\bfM,\bfJ)} \right|_{s=0} &= -
\sip{\left(\mathcal{T}_{(\bfM,\bfJ)}[\h]+\mathcal{J}_{(\bfM,\bfJ)}[\j]\right)\big([(\varphi,\alpha)]\big)}{\phi}_{(\bfM,\bfJ)}
\end{flalign}
for every $\phi\in\Sol_p(\bfM,\bfJ)$, where $\mathcal{T}_{(\bfM,\bfJ)}[\h]$ and $\mathcal{J}_{(\bfM,\bfJ)}[\j]$ were given in
\eqref{eqn:Tdefinition} and \eqref{eqn:Jdefinition}. Accordingly, 
\eqref{eq:rce_deriv} holds in the weak-$*$ topology on $\PhSp_p(\bfM,\bfJ)$. 
\sk

To relate these
formulae with the classical stress-energy tensor and action, we
note that
\begin{flalign}
\mathrm{KG}'_{\bfM[\h]} =\frac{d}{ds}\mathrm{KG}_{\bfM[s\h]}\Big\vert_{s=0} =
-\nabla_a  h^{ab}\nabla_b  + \frac{1}{2}\left(\nabla^a h^b_{\phantom{b}b}\right)\nabla_a 
\end{flalign}
(unfortunately, a sign error appears in the analogous step in 
\cite{Fewster:2011pn}: see the second line of the central displayed formula
on p.1706 of that reference). Inserting this formula and integrating by parts we obtain
\begin{flalign}
\nn \sip{\mathcal{T}_{(\bfM,\bfJ)}[\h]\big([(\varphi,\alpha)]\big)}{\phi}_{(\bfM,\bfJ)} &=
\int_M \left( h_{ab} \ip{\nabla^b \mathrm{E}_{\bfM}(\varphi)}{\nabla^a \phi}
 - \frac{1}{2}h^b_{\phantom{b}b}\nabla^a \ip{\nabla_a \mathrm{E}_{\bfM}(\varphi)}{\phi} 
\right) \, \vol_{\bfM}\notag\\
&\qquad\qquad + \int_M\frac{1}{2}\,g^{ab}\,h_{ab}\,\ip{\bfJ}{\mathrm{E}_{\bfM}(\varphi)}\, \vol_{\bfM} \notag \\
&=\int_M h_{ab}\left( \ip{\nabla^b \mathrm{E}_{\bfM}(\varphi)}{\nabla^a \phi} 
- \frac{1}{2}g^{ab}\,\ip{\nabla_c \mathrm{E}_{\bfM}(\varphi)}{\nabla^c \phi}
\right. \notag\\
&\qquad\qquad  + \frac{1}{2}m^2 g^{ab}\,\ip{ \mathrm{E}_{\bfM}(\varphi)}{\phi}+ \left.\frac{1}{2}\,g^{ab}\,\ip{\bfJ}{\mathrm{E}_{\bfM}(\varphi)}\right) \, \vol_{\bfM}\notag \\
&=
\frac{1}{2}\left.\frac{d}{ds} \int_{M}h_{ab} \,T_{(\bfM,\bfJ)}^{ab}[\phi+s \,\mathrm{E}_{\bfM}(\varphi)]\, \vol_{\bfM} \right|_{s=0}  ~,\label{eq:tmp_stresstensorfuncderiv}
\end{flalign}
thus establishing \eqref{eq:tmp_stresstensorfuncderiv_short}, where the stress-energy tensor is given by \eqref{eqn:fullSET} 
and we have used the fact that 
\begin{flalign}
\nn \nabla^a \ip{\nabla_a \mathrm{E}_{\bfM}(\varphi)}{\phi} &=
\ip{\Box_{\bfM} \big(\mathrm{E}_{\bfM}(\varphi)\big)}{\phi} + \ip{\nabla_a \mathrm{E}_{\bfM}(\varphi)}{\nabla^a\phi} 
\\
&= -m^2 \ip{\mathrm{E}_{\bfM}(\varphi)}{\phi} + \ip{\nabla_a \mathrm{E}_{\bfM}(\varphi)}{\nabla^a\phi} ~.
\end{flalign}

To conclude, we will express the derivative of the relative Cauchy evolution
using the Poisson bracket. Identifying the space of solutions $\Sol_p(\bfM,\bfJ)$ with 
phase space, and identifying an element $[(\varphi,\alpha)]\in\PhSp_p(\bfM)$ with the functional
\begin{flalign}
\phi\mapsto \sip{[(\varphi,\alpha)]}{\phi}_{(\bfM,\bfJ)} =  \left(\int_{M} 
\ip{\phi}{\varphi} \, \vol_{\bfM} \right) +\alpha
\end{flalign}
on phase space (cf.\ \eqref{eqn:pairing}), we already have the Poisson bracket (cf.\ \eqref{eqn:PhSp_PS})
\begin{flalign}
\big\{ [(\varphi,\alpha)],[(\varphi',\alpha')]\big\}_{\sigma_{(\bfM,\bfJ)}}(\phi) :=\sip{\big\{ [(\varphi,\alpha)],[(\varphi',\alpha')]\big\}_{\sigma_{(\bfM,\bfJ)}}}{\phi}_{(\bfM,\bfJ)} =  \mathrm{E}_{\bfM}(\varphi,\varphi')~. 
\end{flalign}
Unsmearing the second slot, this gives
\begin{flalign}
\big\{ [(\varphi,\alpha)],\phi(x)\big\}_{\sigma_{(\bfM,\bfJ)}}(\phi) =  -\big(\mathrm{E}_{\bfM}(\varphi)\big)(x) ~,
\end{flalign}
where, in an obvious way, $\phi(x)$, $x\in M$, stands for the functional $\phi\mapsto\phi(x)$ on phase space. Thus, we also have
for the functional $\phi\mapsto \int_M f \ip{\phi}{\phi}\, \vol_{\bfM}$
\begin{flalign}
\left\{ [(\varphi,\alpha)],\phi\mapsto\int_M f \ip{\phi}{\phi}\, \vol_{\bfM} \right\}_{\sigma_{(\bfM,\bfJ)}}(\phi)= 
 2\,\mathrm{E}_{\bfM}(\varphi, f\phi) ~.
\end{flalign}
Proceeding in this way, one easily obtains the formula
\begin{multline}\label{eq:rce_deriv3}
\left.\frac{d}{ds}\sip{\mathrm{rce}_{(\bfM,\bfJ)}^{(\PhSp_p)}[s\h,s\j] \big([(\varphi,\alpha)]\big) }{\phi}_{(\bfM,\bfJ)} \right|_{s=0} \\
= 
 \left\{ [(\varphi,\alpha)],\phi\mapsto\int_M \left(
\frac{1}{2} h_{ab}\, T^{ab}_{(\bfM,\bfJ)}[\phi] + \ip{\j}{\phi}\right)\, \vol_{\bfM}\right\}_{\sigma_{(\bfM,\bfJ)}}(\phi) ~,
\end{multline}
which can also be written in the form \eqref{eq:PArcederiv}. 
In this form, it is clear that, just as the relative Cauchy evolution correctly identifies  the current coupling to the metric as
the stress-energy tensor, so it also correctly identifies the `current' coupling to the external source $\bfJ$ as the field $\phi$ itself. 

\section{\label{sec:pointed}Pointed presymplectic spaces}
We briefly discuss an alternative method for obtaining
a good classical and quantum theory of the inhomogeneous models. 
The idea is to modify the functor $\PhSp_p$, so that it
takes values in the category of pointed presymplectic spaces
$\pPreSymp$ defined as follows: Objects of $\pPreSymp$ are
pairs $((V,\sigma_V),1_V)$, where $(V,\sigma_V)$ is an object in $\PreSymp$ (so that
$\sigma_V$ has nontrivial null space) and $1_V$ is a distinguished nonzero vector in the null space of $\sigma_V$. 
A morphism $L:((V,\sigma_V),1_V)\to
((W,\sigma_W),1_W)$ is a $\PreSymp$ morphism $L : (V,\sigma_V)\to (W,\sigma_W)$,
such that $L(1_V)=1_W$. For example, in the one-dimensional vector space with trivial presymplectic
structure $\mathcal{I}=(\bbR,0)$ we may single out the unit of $\bbR$ to 
obtain a pointed presymplectic space $(\mathcal{I},1)$, which is then an initial
object of $\pPreSymp$ -- in fact, we may also regard $\pPreSymp$ as the 
category of arrows in $\PreSymp$ with domain $\mathcal{I}$.
\sk

Our modified functor $\pPhSp_p: \LocSrc_p\to \pPreSymp$ is defined on objects by
\begin{flalign}
\pPhSp_p(\bfM,\bfJ)=
\big(\PhSp_p(\bfM,\bfJ),[(0,1)]\big)
\end{flalign}
and on morphisms by $\pPhSp_p(f) = \PhSp_p(f)$, noting that the
latter indeed preserve the distinguished elements $[(0,1)]\in \PhSp_p(\bfM,\bfJ)$, which are naturally distinguished
by their action on solutions: $\sip{[(0,1)]}{\phi}_{(\bfM,\bfJ)}= 1$
for all $\phi\in\Sol_p(\bfM,\bfJ)$. 
\sk

The resulting theory has the expected symmetries. 
\begin{theo} \label{theo:pPhSp}
Every endomorphism of the covariant functor $\pPhSp_p:\LocSrc_p\to \pPreSymp$ is an automorphism and
\begin{flalign}
\End(\pPhSp_p) = \Aut(\pPhSp_p)\simeq 
\begin{cases} \{\id_{\pPhSp_p}\} & ~,~~\text{for }m\neq 0~,\\  \bbR^p &~,~~\text{for } m= 0~,\end{cases}
\end{flalign}
where for $m=0$ the action is given by, for all $[(\varphi,\alpha)]\in \pPhSp_p(\bfM,\bfJ) $ and $\mu\in\bbR^p$,
\begin{flalign}
\bullet\eta(\mu)_{(\bfM,\bfJ)}\big([(\varphi,\alpha)] \big) =\left[\left(\varphi,\alpha+ \int_{M} \ip{\varphi}{\mu}\, \vol_{\bfM} \right)\right]~. 
\end{flalign}
\end{theo}
\begin{proof}
The forgetful functor $\pPreSymp\to\PreSymp$ induces a
faithful homomorphism $\End(\pPhSp_p)\to \End(\PhSp_p)$ of monoids. 
Thus, it suffices to determine which endomorphisms of $\PhSp_p$
lift to $\pPhSp_p$. By Proposition \ref{propo:Z2flip} and Theorem \ref{theo:automorphismgroup},
 we see that $\eta(-1)_{(\bfM,\bfJ)}([(0,1)]) = [(0,-1)]$
so $\eta(-1)$ does not lift, leaving only the trivial group for $m\neq 0$. Similarly, by Proposition \ref{propo:Z2xR}
and Theorem \ref{theo:automorphismgroup},  we see in the massless case 
that $\eta(\sigma,\mu)_{(\bfM,\bfJ)}([(0,1)]) = [(0,\sigma)]$ so $\eta(\sigma,\mu)$ lifts if and only if $\sigma=1$. 
\end{proof}

The second problem identified with $\PhSp_p$ was its failure to 
behave correctly with respect to the composition of systems, see Section \ref{sec:compproperty}. In our
present context this can be remedied in the following way: The natural
composition of pointed presymplectic spaces $((V,\sigma_V),1_V)$ and
$((W,\sigma_W),1_W)$ is not the direct sum, but
rather the direct sum with amalgamation of distinguished points
\begin{flalign}
((V,\sigma_V),1_V) \obulplus ((W,\sigma_W),1_W) = \big( ((V\oplus W)/\sim,\sigma_{V\oplus W}),\llbracket(1_V,0)\rrbracket\big)~,
\end{flalign}
where the equivalence relation implements a quotient by the subspace $\{(\lambda 1_V, -\lambda 1_W ):\lambda\in\bbR\}$.
Note this is well-defined due to our assumption that the distinguished elements lie in the null space of the 
presymplectic structures. The operation $\obulplus$  gives a 
monoidal structure on $\pPreSymp$, with monoidal unit $(\mathcal{I},1)$,  just as the direct sum $\oplus$ 
does for $\PreSymp$ (with the zero-dimensional
presymplectic vector space as the monoidal unit).

\begin{theo} \label{theo:pPhSp_comp}
The theory $\pPhSp_{p}$ obeys the composition property, i.e., there is a natural isomorphism
\begin{flalign}
\pPhSp_p\cong \pPhSp_{p,q} := 
 \obulplus\circ \big(\pPhSp_q\times\pPhSp_{p-q}\big)\circ \mathfrak{Split}_{p,q}~.
\end{flalign}
\end{theo}
\begin{proof} 
For any object $(\bfM,\bfJ)$ in $\LocSrc_p$, define $(\eta_{p,q})_{(\bfM,\bfJ)}:\pPhSp_{p,q}(\bfM,\bfJ)\to\pPhSp_p(\bfM,\bfJ)$ by,
for all $[(\varphi,\alpha)]\in \pPhSp_q(\bfM,\bfJ^q)$ and $[(\psi,\beta)]\in\pPhSp_{p-q}(\bfM,\bfJ^{p-q})$,
\begin{flalign}\label{eq:eta}
(\eta_{p,q})_{(\bfM,\bfJ)}\left(
\left\llbracket \big( [(\varphi,\alpha)] , [(\psi,\beta)] \big) \right\rrbracket\right) = 
\big[ (\varphi +\psi ,\alpha + \beta)\big] ~,
\end{flalign}
where on the right hand side we have identified $\varphi\in C^\infty_0(M,\bbR^q)$ and $\psi\in C^\infty_0(M,\bbR^{p-q})$
as elements in $C^\infty_0(M,\bbR^p)$ ($\varphi$ is placed in the first $q$ and $\psi$ in the last $p-q$ components of
$\bbR^p$).
This map is well-defined (because $\alpha$ and $\beta$ are summed on the right-hand side of \eqref{eq:eta}),
linear and it preserves the distinguished element and the presymplectic structure. Furthermore, it is invertible via
$\pPhSp_p(\bfM,\bfJ)\ni [(\varphi,\alpha)]\mapsto \big\llbracket \big( [(\varphi^q,\alpha)], [(\varphi^{p-q},0)]\big)\big\rrbracket$,
where $\varphi = \varphi^q + \varphi^{p-q}$ denotes the split of $\varphi\in C^\infty_0(M,\bbR^p)$ into
the first $q$ and last $p-q$ components.
Hence, \eqref{eq:eta}  is an isomorphism in $\pPreSymp$.
\sk

To establish naturality, consider any morphism $f:(\bfM_1,\bfJ_1)\to(\bfM_2,\bfJ_2)$ in $\LocSrc_p$.
It is straightforward to check commutativity of
\begin{flalign}
\xymatrix{
\ar@{|->}[d]_-{\pPhSp_{p,q}(f)} \left\llbracket \big([(\varphi,\alpha)] ,[(\psi,\beta)] \big)\right\rrbracket  \ar@{|->}[rrr]^-{(\eta_{p,q})_{(\bfM_1,\bfJ_1)}} &&& \big[(\varphi+\psi,\alpha+\beta)\big] \ar@{|->}[d]^-{\pPhSp_p(f)} 
\\
\left\llbracket\big( [(f_*(\varphi),\alpha)] , [(f_*(\psi),\beta)] \big)\right\rrbracket 
\ar@{|->}[rrr]^-{(\eta_{p,q})_{(\bfM_2,\bfJ_2)}} &&& \big[(f_*(\varphi  + \psi),\alpha+ \beta)\big]
}
\end{flalign}
for all $[(\varphi,\alpha)]\in \pPhSp_q(\bfM_1,\bfJ_1^q)$ and $[(\psi,\beta)]\in\pPhSp_{p-q}(\bfM_1,\bfJ_1^{p-q})$.
Hence, the $(\eta_{p,q})_{(\bfM,\bfJ)}$ form the components of a natural isomorphism. 
\end{proof}

Finally, we consider the quantization of pointed presymplectic spaces
via a suitable covariant functor $\pCCR : \pPreSymp\to \astAlg$. On objects we set $\pCCR((V,\sigma_V),1_V)=
\CCR(V,\sigma_V)/\mathfrak{I}(V,\sigma_V)$, where $\mathfrak{I}(V,\sigma_V)$ is the two-sided $\ast$-ideal
generated by $1_V - 1$. 
On morphisms $L:((V,\sigma_V),1_V)\to ((W,\sigma_W),1_W)$ we have
$\CCR(L)(1_V) = 1_W$ and hence there is a uniquely
defined injective $\ast$-algebra homomorphism $\pCCR(L):\pCCR((V,\sigma_V),1_V)\to \pCCR((W,\sigma_W),1_W)$
 making the following diagram commute
\begin{flalign}
\xymatrix{
\ar[d]  \CCR(V,\sigma_V)  \ar[rr]^-{\CCR(L)} && \CCR(W,\sigma_W) \ar[d]
\\
\pCCR((V,\sigma_V),1_V) \ar[rr]^-{\pCCR(L)}&& \pCCR((W,\sigma_W),1_W)
}
\end{flalign}
where the vertical morphisms are the quotient maps (and the diagram 
is drawn in the category of unital $*$-algebras without the requirement
that morphisms be monic).  
Our last result is simply a restatement of the definition
of the improved functor $\QA_p:\LocSrc_p\to\astAlg$ in Subsection \ref{subsec:improved}.
\begin{theo} 
$\QA_p = \pCCR\circ\pPhSp_p$.
\end{theo}


\section{\label{sec:fedosov}Deformation quantization}
As a further alternative construction of the improved functor $\QA_p: \LocSrc_p\to \astAlg$
we focus on deformation quantization. We show that,
starting from the {\em improved} Poisson algebra functor $\PA_p:\LocSrc_p\to\PoisAlg$,
we can obtain for each object $(\bfM,\bfJ)$ in $\LocSrc_p$ the unital $\ast$-algebra
$\QA_p(\bfM,\bfJ)$ by deformation quantization of (the complexification of) the Poisson algebra $\PA_p(\bfM,\bfJ)$.
After this, we make some remarks on the application of Fedosov quantization \cite{Fedosov:1994zz},
which has been studied recently in \cite{Sanders:2012sf} (although not adhering strictly to \cite{Fedosov:1994zz}, as we will explain) 
for the inhomogeneous Maxwell field.
\sk

Let $(\bfM,\bfJ)$ be any object in $\LocSrc_p$ and consider the Poisson algebra
$\PA_p(\bfM,\bfJ)$ constructed in Subsection \ref{subsec:improvedpoisson}, which we shall denote 
also by $A := \PA_p(\bfM,\bfJ)$ in order to simplify the notation.
We define a {\it differential graded algebra} over the unital algebra $A$
as follows: Consider the graded commutative algebra 
$\Omega^\bullet :=  \PA_p(\bfM,\bfJ) \otimes\bigwedge^\bullet \PhSp_p^\mathrm{lin}(\bfM)$
with product defined by linearity and, for all $a\otimes \lambda, a^\prime\otimes \lambda^\prime\in \Omega^\bullet$,
\begin{flalign}
(a\otimes \lambda)\,(a^\prime\otimes \lambda^\prime) := (a\,a^\prime) \otimes (\lambda\wedge\lambda^\prime)~.
\end{flalign}
We define a differential $\dd: \Omega^\bullet \to \Omega^{\bullet+1}$ by linearity, 
the graded Leibniz rule and setting, for all $1\otimes\lambda\in \Omega^\bullet$ and 
$[(\varphi,\alpha)]\otimes \lambda \in \Omega^\bullet$,
\begin{flalign}
\dd\big(1\otimes\lambda\big) = 0\quad,\qquad \dd\big([(\varphi,\alpha)]\otimes \lambda \big) = 1\otimes \big([\varphi]^\mathrm{lin}\wedge \lambda\big)~.
\end{flalign}
Using the differential graded algebra $(\Omega^\bullet,\dd)$ over $A$,
the Poisson bracket in $\PA_p(\bfM,\bfJ)$ can be reformulated as,
for all $a,a^\prime\in A$,
\begin{flalign}
\{a,a^\prime\}_{\sigma_{(\bfM,\bfJ)}} = \Pi\big(\dd a, \dd a^\prime\big)~,
\end{flalign}
where  the {\it Poisson tensor} $\Pi: \Omega^1 \times \Omega^1 \to A$
is the $A$-bilinear map that is defined by the following extension of $\sigma_{\bfM}^\mathrm{lin}:
 \PhSp_p^\mathrm{lin}(\bfM)\times\PhSp_p^\mathrm{lin}(\bfM)\to \bbR$, 
for all $a\otimes \lambda,a^\prime\otimes\lambda^\prime\in \Omega^1$,
\begin{flalign}
\Pi\big(a\otimes\lambda,a^\prime\otimes\lambda^\prime\big) :=
 \sigma_{\bfM}^\mathrm{lin}(\lambda,\lambda^\prime)\, a\,a^\prime~.
\end{flalign}

On the $A$-module $\Omega^1$ there is a {\it canonical connection}
$\nabla: \Omega^1 \to \Omega^1 \otimes_A \Omega^1$ specified
by linearity and, for all $a\otimes \lambda\in \Omega^1$, 
\begin{flalign}
\nabla(a\otimes \lambda) := (1\otimes \lambda)\otimes_A (\dd a)~.
\end{flalign}
The Leibniz rule $\nabla(\omega\,a) = \nabla(\omega)\,a + \omega\otimes_A\dd a$,
for all $\omega\in\Omega^1$ and $a\in A$, is obviously satisfied.
The {\it torsion} of $\nabla$ is the $A$-module homomorphism
$\mathrm{T}: \Omega^1 \to \Omega^2$ defined by, for all $\omega\in\Omega^1$,
\begin{flalign}
\mathrm{T}(\omega) := \wedge\big(\nabla(\omega)\big) + \dd \omega~.
\end{flalign}
One easily checks that the canonical connection $\nabla$ is torsion free.
The {\it curvature} of $\nabla$ is the $A$-module homomorphism
$\mathrm{R}: \Omega^1\to \Omega^1 \otimes_A\Omega^2$ defined by, 
for all $\omega\in\Omega^1$,
\begin{flalign}
\mathrm{R}(\omega) := \nabla^2(\omega) := \nabla\nabla(\omega)~.
\end{flalign}
When applying $\nabla$ the second time, the usual extension 
$\nabla: \Omega^1\otimes_A\Omega^\bullet \to 
\Omega^1\otimes_A \Omega^{\bullet+1}$ 
defined by linearity and, for all $\omega\otimes_A\omega^\prime\in\Omega^1\otimes_A\Omega^\bullet $,
\begin{flalign}\label{eqn:usualextension}
\nabla(\omega\otimes_A\omega^\prime) := (\id_{\Omega^1}\otimes \wedge)\big(\nabla(\omega)\otimes_A \omega^\prime\big) 
+ \omega\otimes_A\dd\omega^\prime
\end{flalign}
is implicitly understood. One easily checks that the canonical connection $\nabla$  is flat, i.e., $\mathrm{R}=0$. 
As a last property, notice the canonical connection preserves the Poisson
tensor $\Pi$, i.e.\ the $A$-bilinear map
 $\mathrm{Q} : \Omega^1\times\Omega^1\to \Omega^1$
defined by, for all $\omega,\omega^\prime\in\Omega^1$,
\begin{flalign}\label{eqn:poiscompcondition}
\mathrm{Q}(\omega,\omega^\prime) := \dd \big(\Pi\big(\omega,\omega^\prime\big) \big)- 
\Pi\big(\nabla(\omega),\omega^\prime\big)
- \Pi\big(\omega,\nabla(\omega^\prime)\big) 
\end{flalign}
vanishes.\footnote{
In (\ref{eqn:poiscompcondition}) we have extended
the Poisson tensor $\Pi:\Omega^1\times\Omega^1\to A $ to $A$-bilinear maps 
 $(\Omega^1\otimes_A\Omega^1)\times \Omega^1\to \Omega^1$ and $\Omega^1\times (\Omega^1\otimes_A\Omega^1)\to \Omega^1$
 (denoted by the same symbol) by setting $\Pi(\omega\otimes_A\omega^\prime , \omega^{\prime\prime}) 
 :=  \Pi(\omega,\omega^{\prime\prime})\,\omega^\prime$ and 
 $\Pi(\omega,\omega^\prime\otimes_A\omega^{\prime\prime}) := 
 \Pi(\omega,\omega^\prime)\,\omega^{\prime\prime}$, for all $\omega,\omega^\prime,\omega^{\prime\prime}\in\Omega^1$.
} To sum up, we have shown that the canonical connection $\nabla:\Omega^1\to\Omega^1\otimes_A\Omega^1$
is a flat and torsion free Poisson connection.
\sk

As the next step, we consider the tensor module $\mathcal{E} := \bigoplus_{n=0}^\infty \left(\Omega^{1}\right)^{\otimes_A^n}$,
which is an $\bbN_0$-graded $A$-module that describes {\it tensor fields} on $A$. The connection 
$\nabla$ on $\Omega^1$ lifts to a connection on
$\mathcal{E} = \bigoplus_{n=0}^\infty\mathcal{E}^n$
via a recursive construction: On $\mathcal{E}^0\simeq A$ we choose the connection
$\nabla_0: A\to A\otimes_A \Omega^1\simeq \Omega^1 $ given by the differential $\dd$. On 
$\mathcal{E}^1= \Omega^1$ we take the canonical connection $\nabla_1 = \nabla: \Omega^1\to \Omega^1\otimes_A\Omega^1$.
Given a connection $\nabla_n$ on $\mathcal{E}^n = \Omega^1\otimes_A\cdots\otimes_A\Omega^1$ ($n$-times) we construct 
a connection $\nabla_{n+1}: \mathcal{E}^{n+1}\to \mathcal{E}^{n+1}\otimes_A \Omega^1$ on 
$\mathcal{E}^{n+1} = \mathcal{E}^{n}\otimes_A \Omega^1$ by linearity and setting, for all
 $\omega\otimes_A\omega^\prime \in \mathcal{E}^{n}\otimes_A \Omega^1$,
 \begin{flalign}\label{eqn:tensorconnection}
 \nabla_{n+1}(\omega\otimes_A\omega^\prime) := 
 (\id_{\mathcal{E}^n}\otimes \tau)\big(\nabla_{n}(\omega) \otimes_A \omega^\prime\big) + \omega\otimes_A \nabla_1(\omega^\prime)~,
 \end{flalign}
where $\tau:\Omega^1\otimes_A\Omega^1\to\Omega^1\otimes_A\Omega^1$ is the flip map,
for all $\omega\otimes_A\omega^\prime$, $\tau(\omega\otimes_A\omega^\prime) = \omega^\prime \otimes_A\omega$.
To simplify the notation, we shall denote the resulting connection on $\mathcal{E}$ by $\nabla_{\mathcal{E}}$
and notice that we can regard it as a linear map (of degree 1) $\nabla_{\mathcal{E}}:\mathcal{E}\to\mathcal{E}$. 
\sk

Since $\nabla$ is a flat and torsion free Poisson connection, we can define an associative $\star$-product 
on the complexification of $A$ (denoted with abuse of notation also by $A$),
for all $a,a^\prime\in A$,
\begin{flalign}\label{eqn:starproductfedo}
a\star_{(\Pi,\nabla)} a^\prime := \sum_{n=0}^\infty\frac{1}{n!} \left(\frac{i}{2}\right)^n\Pi^n\big(\nabla_\mathcal{E}^n(a),\nabla_\mathcal{E}^n(a^\prime)\big)~.
\end{flalign}
Here $\nabla_\mathcal{E}^n $ denotes the $n$-times iterated application of $\nabla_\mathcal{E}$ and
 $\Pi^n:\mathcal{E}^n\times\mathcal{E}^n\to A$ is the $A$-bilinear map specified by, for
all $\omega_1,\dots,\omega_n,\omega_1^\prime,\dots,\omega^\prime_n\in\Omega^1$,
\begin{flalign}
\Pi^n\big(\omega_1\otimes_A\cdots\otimes_A \omega_n, \omega_1^\prime\otimes_A\cdots,\otimes_A \omega_n^\prime\big)
:=\Pi(\omega_1,\omega_1^\prime) \cdots \Pi(\omega_n,\omega_n^\prime)~
\end{flalign}
and $\Pi^0(a,a^\prime) = a\,a^\prime$, for all $a,a^\prime\in A$.
Notice that the sum in (\ref{eqn:starproductfedo}) terminates, since $a,a^\prime\in A$ 
are polynomials, hence $\nabla_{\mathcal{E}}^m(a)=0$ and $\nabla_{\mathcal{E}}^{m^\prime}(a^\prime)=0$
for sufficiently large  $m,m^\prime\in \bbN$ (see also (\ref{eqn:derivativefinite}) below).
Furthermore, the $\star$-product $\star_{(\Pi,\nabla)}$ is hermitian if we equip $A$ with the involution
$\ast$ defined by $\big([(\varphi_1,\alpha_1)]\,\cdots\,[(\varphi_n,\alpha_n)]\big)^\ast =
 [(\varphi_1,\alpha_1)]\,\cdots\,[(\varphi_n,\alpha_n)]$ and $\bbC$-antilinear extension.
 \sk
 
 It remains to show that the $\star$-product (\ref{eqn:starproductfedo}) 
 coincides with the product in $\QA_p(\bfM,\bfJ)$, which is defined in (\ref{eqn:starproduct}).
 For the elements $[(\varphi,\alpha)]^m \in A$ we find, for all $k\leq m$,
 \begin{flalign}\label{eqn:derivativefinite}
 \nabla^k_{\mathcal{E}}\big([(\varphi,\alpha)]^m \big) = \frac{m!}{(m-k)!}\, [(\varphi,\alpha)]^{m-k}\otimes {[\varphi]^{\mathrm{lin}}}^{\otimes k}~,
 \end{flalign}
and for $k>m$, $\nabla^k_{\mathcal{E}}\big([(\varphi,\alpha)]^m \big) =0$.
Plugging this into (\ref{eqn:starproductfedo})  we obtain, 
 for all $[(\varphi,\alpha)]^m,[(\varphi^\prime,\alpha^\prime)]^n\in A$,
 \begin{multline}
 [(\varphi,\alpha)]^m \star_{(\Pi,\nabla)}[(\varphi^\prime,\alpha^\prime)]^n \\
 = 
  \sum_{k=0}^{\min(m,n)} \left(\frac{i\sigma_{\bfM}^\mathrm{lin}([\varphi]^\mathrm{lin},[\varphi^\prime]^\mathrm{lin})}{2}\right)^k  \frac{m!\,n!}{k!\,(m-k)! \,(n-k)!} [(\varphi,\alpha)]^{m-k} \,[(\varphi^\prime,\alpha^\prime)]^{n-k}~.
 \end{multline}
 This shows that the products (\ref{eqn:starproductfedo}) and  (\ref{eqn:starproduct}) coincide.
\sk

In view of this direct construction of the $\star$-product (\ref{eqn:starproductfedo}) given above,
the application of full-fledged Fedosov quantization \cite{Fedosov:1994zz} to our model is not required.
However, we shall now focus on this quantization method using our algebraic approach developed in this appendix,
as this will clarify certain issues in an earlier treatment of this subject \cite{Sanders:2012sf}.
The basic structure entering Fedosov's approach is a bundle of CCR-algebras over a symplectic (or regular Poisson)
manifold. In our algebraic approach this bundle is given by the  $A$-module $\mathcal{W} := S^{\otimes_A}(\Omega^1)$,
where $S^{\otimes_A}(\Omega^1)$ is the (complexified) symmetric tensor algebra with respect to the tensor product $\otimes_A$
of the $A$-module of one-forms $\Omega^1$ on $A$. Using that $\Omega^1 = \PA_p(\bfM,\bfJ) \otimes \PhSp_p^{\mathrm{lin}}(\bfM)$,
$\mathcal{W}$ is isomorphic to $\PA_p(\bfM,\bfJ)\otimes S\big(\PhSp_p^{\mathrm{lin}}(\bfM)\big)$. We shall 
suppress this isomorphism. Notice that  $S\big(\PhSp_p^{\mathrm{lin}}(\bfM)\big)$ is the vector space underlying
the CCR-algebra $B:=\CQA_p^{\mathrm{lin}}(\bfM) = \CCR\big(\PhSp_p^\mathrm{lin}(\bfM)\big)$ -- 
the quantum algebra of observables of the homogeneous theory -- hence, we can equip the $A$-module
$\mathcal{W} = A\otimes B $ with a (noncommutative) product, for all
$a\otimes b,a^\prime\otimes b^\prime\in A\otimes B$,
\begin{flalign}
(a\otimes b)\,(a^\prime\otimes b^\prime) := (a\,a^\prime) \otimes (b\star b^\prime)~.
\end{flalign}
Notice that $A$ can be identified with a commutative subalgebra of $\mathcal{W}$ via $A\to \mathcal{W}\,,~a\mapsto a\otimes 1$.
Fedosov's idea is to characterize a subalgebra $A_\star$ of $\mathcal{W}$ in such a way that
$A_\star$ is isomorphic to $A$ as a vector space and that the induced product on $A_\star$ is a deformation
quantization of the Poisson structure on $A$. To achieve this goal, the first step is to construct a suitable connection on $\mathcal{W}$
as follows: Given the canonical connection $\nabla$ on $\Omega^1$, we can induce a tensor product connection $\nabla_{\mathcal{W}}$
on $\mathcal{W}$ via the prescription outlined in (\ref{eqn:tensorconnection}). This connection extends
analogously to (\ref{eqn:usualextension}) to a linear map $\nabla_{\mathcal{W}} : 
\mathcal{W}\otimes_A\Omega^\bullet \to \mathcal{W}\otimes_A\Omega^{\bullet +1}$. Notice that
$\mathcal{W}\otimes_A \Omega^\bullet$ is an $\bbN_0$-graded algebra (the grading is inherited from that of differential forms
$\Omega^\bullet$) by setting, for all
$\mathrm{w}\otimes_A \omega,\mathrm{w}^\prime\otimes_A\omega^\prime\in \mathcal{W}\otimes_A \Omega^\bullet$,
$(\mathrm{w}\otimes_A\omega)\,(\mathrm{w}^\prime\otimes_A\omega^\prime) 
= (\mathrm{w}\,\mathrm{w}^\prime)\otimes_A(\omega\,\omega^\prime)$. It is easy to check that $\nabla_{\mathcal{W}}$
is flat, i.e.\ $\mathrm{R}_{\mathcal{W}} = \nabla_{\mathcal{W}}^2 =0$,
since it is the tensor product connection of our flat canonical connection $\nabla$ on $\Omega^1$.
Furthermore, $\nabla_{\mathcal{W}}$ satisfies the graded Leibniz rule on the $\bbN_0$-graded algebra,
for all homogeneous elements
$\mathrm{w}^\bullet,\mathrm{w}^{\bullet\prime}\in \mathcal{W}\otimes_A \Omega^\bullet$,
\begin{flalign}
\nabla_{\mathcal{W}} \big(\mathrm{w}^\bullet\,\mathrm{w}^{\bullet\prime} \big)
=\big(\nabla_{\mathcal{W}}(\mathrm{w}^\bullet)\big)\,\mathrm{w}^{\bullet\prime} + (-1)^{\vert\mathrm{w}^\bullet\vert}\,
\mathrm{w}^\bullet\,\nabla_{\mathcal{W}}(\mathrm{w}^{\bullet\prime})~. 
\end{flalign}
Thus, $\nabla_{\mathcal{W}}$ structures $\mathcal{W}\otimes_A\Omega^\bullet$ as a differential graded algebra.
Fedosov's idea \cite{Fedosov:1994zz}  is now to modify $\nabla_{\mathcal{W}}$ into a differential 
$\mathsf{D}:  \mathcal{W}\otimes_A\Omega^\bullet\to \mathcal{W}\otimes_A\Omega^{\bullet+1}$, such that
the kernel $\ker(\mathsf{D})\cap \mathcal{W}$, which is a unital $\ast$-algebra under the product inherited from
$\mathcal{W}\otimes_A \Omega^\bullet$, gives the desired deformation quantization of $A$. 
Because $\nabla_{\mathcal{W}}$ is flat, this construction drastically simplifies and
we do not have to take into account the corrections by curvature dependent terms as in \cite{Fedosov:1994zz}.
The Fedosov differential for our model is given by  
\begin{flalign}
\mathsf{D}:= -\delta + \nabla_{\mathcal{W}}~,
\end{flalign}
where  $\delta: \mathcal{W}\otimes_A\Omega^\bullet\to \mathcal{W}\otimes_A\Omega^{\bullet+1}$ is the $A$-module 
homomorphism defined by linearity and, for all $a\otimes \big( [\varphi_1]^\mathrm{lin}\,\cdots\,[\varphi_n]^\mathrm{lin}\big)
\otimes \lambda \in A\otimes B \otimes 
\bigwedge^\bullet\PhSp_p^{\mathrm{lin}}(\bfM)\simeq \mathcal{W}\otimes_A\Omega^\bullet$,
\begin{flalign}
\delta\Big(a\otimes \big( [\varphi_1]^\mathrm{lin}\,\cdots\,[\varphi_n]^\mathrm{lin}\big)
\otimes \lambda\Big) = \sum_{j=1}^n a\otimes \big([\varphi_1]^\mathrm{lin}\cdots\omi{j}\cdots [\varphi_{n}]^\mathrm{lin} \big)\otimes\big( [\varphi_j]^\mathrm{lin}\wedge \lambda\big)~.
\end{flalign}
It is easy to check that $\delta$ satisfies the graded Leibniz rule, $\delta^2=0$ and
 $\delta\circ \nabla_{\mathcal{W}} + \nabla_{\mathcal{W}}\circ \delta=0$,
from which it follows that $\mathsf{D}$ is a differential on $\mathcal{W}\otimes_A\Omega^\bullet$.
\sk

We now come to the characterization of the kernel $\ker(\mathsf{D})\cap \mathcal{W}$.
Like in \cite{Fedosov:1994zz}, we are making use  of the $A$-module homomorphism 
$\delta^\ast: \mathcal{W}\otimes_A\Omega^\bullet\to \mathcal{W}\otimes_A\Omega^{\bullet-1}$ 
defined by linearity and, for all 
$a\otimes b\otimes \big( [\varphi_1]^\mathrm{lin}\wedge \cdots\wedge [\varphi_n]^\mathrm{lin}\big) \in A\otimes B \otimes 
\bigwedge^\bullet\PhSp_p^{\mathrm{lin}}(\bfM)\simeq \mathcal{W}\otimes_A\Omega^\bullet$,
\begin{flalign}
\delta^\ast\Big(a\otimes b\otimes \big( [\varphi_1]^\mathrm{lin}\wedge \cdots\wedge [\varphi_n]^\mathrm{lin}\big)
\Big) = \sum_{j=1}^n  (-1)^{j+1} \, a\otimes \big(b\,[\varphi_j]^{\mathrm{lin}} \big)\otimes \big( [\varphi_1]^\mathrm{lin}\wedge \cdots\omi{j}\cdots\wedge [\varphi_n]^\mathrm{lin}\big)~.
\end{flalign}
It is easy to check that ${\delta^\ast}^2 =0$ and that $\delta^\ast\circ \delta + \delta\circ\delta^\ast  = (n+m)\,\id$,
when acting on homogeneous elements $a\otimes \big( [\varphi_1]^\mathrm{lin}\,\cdots\,[\varphi_n]^\mathrm{lin}\big)\otimes 
\big( [\varphi_{n+1}]^\mathrm{lin}\wedge \cdots\wedge [\varphi_{n+m}]^\mathrm{lin}\big)  $. The latter property
implies that  
\begin{flalign}\label{eqn:homotopy}
\delta^{-1}\circ \delta + \delta\circ\delta^{-1} +\sigma = \id_{\mathcal{W}\otimes_A\Omega^\bullet}
\end{flalign} 
on all of $\mathcal{W}\otimes_A\Omega^\bullet$, where 
$\delta^{-1}: \mathcal{W}\otimes_A\Omega^\bullet\to \mathcal{W}\otimes_A\Omega^{\bullet-1}$
is defined on homogeneous elements by $\delta^{-1} = \delta^\ast/(n+m)$ for $n+m\neq 0$ and $\delta^{-1}=0$ for $n+m=0$.
The  linear map $\sigma: \mathcal{W}\otimes_A \Omega^\bullet \to A$ is the projection defined by
$\sigma(a\otimes 1\otimes 1) = a$ and $\sigma(a\otimes b\otimes \lambda)=0$ if the degree of $b$ or $\lambda$
is not equal to zero. 
Following the proof of Fedosov \cite{Fedosov:1994zz}, we can show that
the map $\sigma : \ker(\mathsf{D})\cap \mathcal{W} \to A$ is bijective, i.e.\ that for any
$a\in A$ there exists a unique $\mathrm{w}\in \mathcal{W}$, such that $\mathrm{D}(\mathrm{w})=0$ and
$\sigma(\mathrm{w})=a$. We briefly sketch the relevant steps: Let $\mathrm{w}\in \mathcal{W}$ be such that
$0=\mathsf{D}(\mathrm{w}) = -\delta\mathrm{w} + \nabla_{\mathcal{W}}(\mathrm{w})$. Applying $\delta^{-1}$ and
using (\ref{eqn:homotopy}) this yields the equation
\begin{flalign}\label{eqn:constantweylsection}
\mathrm{w} = \sigma(\mathrm{w}) + \delta^{-1}\big(\nabla_{\mathcal{W}}(\mathrm{w})\big)~.
\end{flalign}
Notice that $\sigma(\mathrm{w})$ has degree $(0,0)$ according to the natural grading $(n,m)$ on $\mathcal{W}\otimes_A\Omega^\bullet$
discussed above. The map $\nabla_{\mathcal{W}}$ increases the form-degree by one $(n,m)\mapsto (n,m+1)$, while the
map $\delta^{-1}$ decreases the form-degree by one and increases the $B$-degree by one $(n,m)\mapsto (n+1,m-1)$.
Hence, $\delta^{-1}\circ \nabla_{\mathcal{W}}$ increases the $B$-degree by one $(n,m)\mapsto (n+1,m)$
and equation (\ref{eqn:constantweylsection}) can be solved uniquely by iteration for any initial condition
$\sigma(\mathrm{w}) =a$ (since $A$ is a polynomial algebra this 
requires just a finite number of iterations). For any initial condition 
$\sigma(\mathrm{w}) =a\in A$, the solution $\mathrm{w}$ to (\ref{eqn:constantweylsection}) satisfies $\mathsf{D}(\mathrm{w})=0$
as a consequence of $\mathsf{D}$ being a differential, cf.\ \cite{Fedosov:1994zz}. This establishes
the bijection $\sigma: \ker(\mathsf{D})\cap \mathcal{W} \to A$ and we define a $\star$-product on $A$
by setting, for all $a,a^\prime\in A$,
\begin{flalign}\label{eqn:laststarproduct}
a\star_{\mathrm{F}} a^\prime := \sigma\big(\sigma^{-1}(a)\,\sigma^{-1}(a^\prime)\big)~,
\end{flalign}
where the product between $\sigma^{-1}(a)$ and $\sigma^{-1}(a^\prime)$ is of course taken in $\mathcal{W}$.
\sk

It remains to show that (\ref{eqn:laststarproduct}) coincides with the product (\ref{eqn:starproduct}).
For this let us take $[(\varphi,\alpha)]^m\in A$ and notice that
\begin{flalign}
\sigma^{-1}\big([(\varphi,\alpha)]^m\big) = \sum_{j=0}^m \binom{m}{j}\,[(\varphi,\alpha)]^{m-j}\otimes \left([\varphi]^\mathrm{lin}\right)^j~.
\end{flalign}
From this expression and a slightly tedious calculation one obtains that
\begin{flalign}
[(\varphi,\alpha)]^m\star_\mathrm{F}[(\varphi^\prime,\alpha^\prime)]^n = 
\sigma\Big(\sigma^{-1}\big([(\varphi,\alpha)]^m\big) \,\sigma^{-1}\big([(\varphi^\prime,\alpha^\prime)]^n\big)\Big)
= [(\varphi,\alpha)]^m\star [(\varphi^\prime,\alpha^\prime)]^n~,
\end{flalign}
where the product on the right hand side is (\ref{eqn:starproduct}).
So the three products  (\ref{eqn:starproduct}), (\ref{eqn:starproductfedo}) and (\ref{eqn:laststarproduct})
all coincide on the complexification of $\PA_p(\bfM,\bfJ)$ and hence give the same quantum algebra
$\QA_p(\bfM,\bfJ)$.
\sk

To conclude, we make some remarks on the quantization prescription
pursued in \cite{Sanders:2012sf}, which is described as being a Fedosov
quantization, but in fact differs in essential respects from the Fedosov method \cite{Fedosov:1994zz}. 
Briefly, the method of \cite{Sanders:2012sf} is to construct a
bundle $\mathcal{A}$ of infinitesimal Weyl algebras over an affine space $V$, equipped with a Poisson structure, 
and then to define the quantized algebra as the algebra of flat sections in $\mathcal{A}$ with respect to a certain connection. 
While this basic idea matches exactly with the Fedosov construction,
the starting point chosen in \cite{Sanders:2012sf} is an {\em affine} connection on the tangent bundle of 
$V$ with prescribed affine parallel transport maps between all fibres. The problem with
this choice is that the affine connection does not dualize to the 
cotangent bundle and in particular not to the bundle of infinitesimal Weyl algebras 
$\mathcal{A}$.\footnote{The correct dualization would be to the vector dual bundle of the 
tangent bundle regarded in the category of affine bundles, but this would lead to the following logical 
problem: Instead of replacing the problem of quantizing affine Poisson spaces by quantizing  
linear Poisson spaces (which Fedosov's method does, as explained above),
the choice of affine connection in \cite{Sanders:2012sf} replaces the problem of quantizing affine Poisson spaces
by quantizing affine Poisson spaces in the fibres.}
This issue is sidestepped in \cite{Sanders:2012sf}, by regarding
elements of the infinitesimal Weyl algebras  
as symmetric polynomials acting on the vector space $V_0$ on which $V$ is modeled, which permits a unique parallel transport between
fibres in $\mathcal{A}$ to be defined. This could be regarded as a slightly ad hoc mixture of the quantized and classical
theories, because $V_0$ is analogous to the classical solution space of the homogeneous theory. By contrast, our
prescriptions for improved theories are based on 
classical structures in the classical case and quantum structures in the quantum case. The connection thus employed in \cite{Sanders:2012sf}
is much more rigid than that of \cite{Fedosov:1994zz}, which does
not integrate to a unique parallel transport between fibres (see 
\cite[p. 222]{Fedosov:1994zz}).\footnote{Unique parallel transport depends on theorems on existence and uniqueness of solutions to differential equations that do not necessarily apply to bundles with infinite-dimensional fibres (beyond the Banach case).  A related point is that \cite{Fedosov:1994zz} obtains
a deformation of the algebra of all smooth functions on the classical phase space, while the result of \cite{Sanders:2012sf} is more analogous
to a deformation of an algebra of certain polynomial functions. In our algebraic version of Fedosov's
procedure, the restriction to polynomial functions appears naturally; it should also be possible
to extend our construction to cover Wick polynomials by modifying the underlying Poisson algebra, but
without modifying the fibre algebras, whereas the corresponding extension of \cite{Sanders:2012sf}
would require modification of the fibres.} 
With these thoughts in mind,  the approach in \cite{Sanders:2012sf} might be better
described as `Fedosov-inspired' rather than an application of 
Fedosov's method as such; nonetheless, this procedure does lead to the correct `improved algebra'. 



\begin{thebibliography}{10}

\bibitem[BSZ92]{Baez:1992tj} 
  J.~C.~Baez, I.~E.~Segal and Z.~Zhou,
  ``Introduction to algebraic and constructive quantum field theory,'
  Princeton Series in Physics, Princeton University Press, Princeton (1992).
  
  \bibitem[BDF09]{Brunetti:2009qc} 
    R.~Brunetti, M.~Duetsch and K.~Fredenhagen,
    ``Perturbative Algebraic Quantum Field Theory and the Renormalization Groups,''
    Adv.\ Theor.\ Math.\ Phys.\  {\bf 13}, 1541 (2009)
    [arXiv:0901.2038 [math-ph]].
    
\bibitem[BDHS13]{Benini:2013ita} 
  M.~Benini, C.~Dappiaggi, T.~-P.~Hack and A.~Schenkel,
``A $C^\ast$-algebra for quantized principal $U(1)$-connections on globally hyperbolic Lorentzian manifolds,''
{\it to appear in Communications in Mathematical Physics}
  [arXiv:1307.3052 [math-ph]].
  
\bibitem[BDS14a]{Benini:2012vi} 
  M.~Benini, C.~Dappiaggi and A.~Schenkel,
  ``Quantum field theory on affine bundles,''
    Annales Henri Poincar{\'e} {\bf 15}, 171 -- 211 (2014) 
    [arXiv:1210.3457 [math-ph]].
    
    \bibitem[BDS14b]{Benini:2013tra} 
      M.~Benini, C.~Dappiaggi and A.~Schenkel,
      ``Quantized Abelian principal connections on Lorentzian manifolds,''
      Commun.\ Math.\ Phys.\  {\bf 330}, 123 (2014)
        [arXiv:1303.2515 [math-ph]].
     
      
   \bibitem[Bor62]{Borchers}   
      H.~-J.~Borchers, 
     ``On structure of the algebra of field operators,''
      Nuovo Cimento {\bf 24}, 214--236 (1962). 
    
\bibitem[BFR13]{Brunetti:2013maa} 
  R.~Brunetti, K.~Fredenhagen and K.~Rejzner,
  ``Quantum gravity from the point of view of locally covariant quantum field theory,''
  arXiv:1306.1058 [math-ph].
    
 \bibitem[BFV03]{Brunetti:2001dx} 
  R.~Brunetti, K.~Fredenhagen and R.~Verch,
  ``The generally covariant locality principle: A New paradigm for local quantum field theory,''
  Commun.\ Math.\ Phys.\  {\bf 237}, 31 (2003)
  [math-ph/0112041].
  
  \bibitem[Fed94]{Fedosov:1994zz} 
    B.~V.~Fedosov,
    ``A Simple geometrical construction of deformation quantization,''
    J.\ Diff.\ Geom.\  {\bf 40}, 213 (1994).
  
  
\bibitem[Fer13]{Ferguson:2014}
  M.~Ferguson,
    ``Dynamical Locality of the nonminimally coupled scalar field and enlarged algebra of Wick polynomials,''
  Annales Henri Poincar\'e\ {\bf 14}, 853 (2013)
[math-ph/1203.2151].

  
  \bibitem[Few07]{Fewster:2006iy} 
    C.~J.~Fewster,
    ``Quantum energy inequalities and local covariance. II. Categorical formulation,''
    Gen.\ Rel.\ Grav.\  {\bf 39}, 1855 (2007)
    [math-ph/0611058].
    
  \bibitem[Few13]{Fewster:2012yc} 
  C.~J.~Fewster,
  ``Endomorphisms and automorphisms of locally covariant quantum field theories,''
  Rev.\ Math.\ Phys.\  {\bf 25}, 1350008 (2013)
  [arXiv:1201.3295 [math-ph]].
  
  \bibitem[FP06]{Fewster:2006kt} 
    C.~J.~Fewster and M.~J.~Pfenning,
    ``Quantum energy inequalities and local covariance. I. Globally hyperbolic spacetimes,''
    J.\ Math.\ Phys.\  {\bf 47}, 082303 (2006)
    [math-ph/0602042].
    
    \bibitem[FR13]{Fredenhagen:2011mq} 
      K.~Fredenhagen and K.~Rejzner,
      ``Batalin-Vilkovisky formalism in perturbative algebraic quantum field theory,''
      Commun.\ Math.\ Phys.\  {\bf 317}, 697 (2013)
      [arXiv:1110.5232 [math-ph]].

  \bibitem[FV12a]{Fewster:2011pe} 
  C.~J.~Fewster and R.~Verch,
  ``Dynamical locality and covariance: What makes a physical theory the same in all spacetimes?,''
  Annales Henri Poincare {\bf 13}, 1613 (2012)
  [arXiv:1106.4785 [math-ph]].
  
  \bibitem[FV12b]{Fewster:2011pn} 
  C.~J.~Fewster and R.~Verch,
  ``Dynamical locality of the free scalar field,''
  Annales Henri Poincare {\bf 13}, 1675 (2012)
  [arXiv:1109.6732 [math-ph]].
  
  \bibitem[HW01]{Hollands:2001nf} 
    S.~Hollands and R.~M.~Wald,
    ``Local Wick polynomials and time ordered products of quantum fields in curved space-time,''
    Commun.\ Math.\ Phys.\  {\bf 223}, 289 (2001)
    [gr-qc/0103074].
    
    \bibitem[HW02]{Hollands:2001fb} 
      S.~Hollands and R.~M.~Wald,
      ``Existence of local covariant time ordered products of quantum fields in curved space-time,''
      Commun.\ Math.\ Phys.\  {\bf 231}, 309 (2002)
      [gr-qc/0111108].
      
      \bibitem[HW05]{Hollands:2004yh} 
        S.~Hollands and R.~M.~Wald,
        ``Conservation of the stress tensor in interacting quantum field theory in curved spacetimes,''
        Rev.\ Math.\ Phys.\  {\bf 17}, 227 (2005)
        [gr-qc/0404074].

\bibitem[IZ80]{IZbook}
  C.~Itzykson and J.~B.~Zuber,
  ``Quantum Field Theory,''
  New York, USA: McGraw-Hill (1980).

  \bibitem[Mei71]{Meisters:1971}
  G.~H.~Meisters, 
  ``Translation-invariant linear forms and a formula for the Dirac measure,''
  J.\ Functional Analysis {\bf 8}, 173--188 (1971).
  
  \bibitem[Pei52]{Peierls:1952}
  R.~E.~Peierls, 
  ``The commutation laws of relativistic field theory,''
  Proc.\ Roy.\ Soc.\ {\bf A 214}, 143--157 (1952).

  \bibitem[PS13a]{Pinamonti:2013zba} 
    N.~Pinamonti and D.~Siemssen,
    ``Scale-Invariant Curvature Fluctuations from an Extended Semiclassical Gravity,''
    arXiv:1303.3241 [gr-qc].
    
    \bibitem[PS13b]{Pinamonti:2013wya} 
      N.~Pinamonti and D.~Siemssen,
      ``Global Existence of Solutions of the Semiclassical Einstein Equation,''
      arXiv:1309.6303 [math-ph].
      
   \bibitem[SDH14]{Sanders:2012sf} 
      K.~Sanders, C.~Dappiaggi and T.~-P.~Hack,
     ``Electromagnetism, local covariance, the Aharonov-Bohm effect and Gauss' law,''
     Commun.\ Math.\ Phys.\  {\bf 328}, 625 (2014)
       [arXiv:1211.6420 [math-ph]].
      
   \bibitem[Uhl62]{Uhlmann}
   A.~Uhlmann,
   ``\"Uber die Definition der Quantenfelder nach Wightman und Haag,''
   Wiss.\ Zeit.\ Karl Marx Univ.\ {\bf 11}, 213--217 (1962).  

  
  \bibitem[Ver01]{Verch:2001bv} 
    R.~Verch,
    ``A spin statistics theorem for quantum fields on curved space-time manifolds in a generally covariant framework,''
    Commun.\ Math.\ Phys.\  {\bf 223}, 261 (2001)
    [math-ph/0102035].
    
    \bibitem[Ver12]{Verch:2011bx} 
      R.~Verch,
      ``Local covariance, renormalization ambiguity, and local thermal equilibrium in cosmology,''
     in Quantum Field Theory and Gravity. Conceptual and mathematical advances in the search for a unified framework, eds.\
      Finster, F., M\"uller, O., Nardmann, M., Tolksdorf, J., and Zeidler, E. (Birkh\"auser, 2012), pp.\ 229-256,
     [arXiv:1105.6249 [gr-qc]].
\end{thebibliography}
\end{document}